\documentclass[11pt]{article}
\usepackage[margin=1in]{geometry}
\usepackage{amsmath,amsfonts,amssymb,amsthm}
\usepackage{graphicx}
\usepackage{enumerate}
\usepackage{bbm}
\usepackage{verbatim}
\usepackage{hyperref,color}
\usepackage[capitalize,nameinlink]{cleveref}
\usepackage[dvipsnames]{xcolor}
\hypersetup{
	colorlinks=true,
	pdfpagemode=UseNone,
    citecolor=OliveGreen,
    linkcolor=NavyBlue,
    urlcolor=Magenta,
	pdfstartview=FitW
}
\usepackage{appendix}
\crefname{appsec}{Appendix}{Appendices}
\usepackage{tikz}
\usepackage{array}
\usepackage[ruled,linesnumbered]{algorithm2e}
\SetKwComment{Comment}{/* }{ */}
\usepackage{multicol}

\theoremstyle{plain}
\newtheorem{thm}{Theorem}[section]
\newtheorem{theorem}[thm]{Theorem}

\newtheorem{lemma}[thm]{Lemma}

\newtheorem{corollary}[thm]{Corollary}

\newtheorem{clm}[thm]{Claim}

\newtheorem{fact}[thm]{Fact}

\theoremstyle{definition}

\newtheorem{definition}[thm]{Definition}

\newtheorem*{ass*}{Assumption}
\newtheorem*{assumption*}{Assumption}

\theoremstyle{remark}

\newtheorem{remark}[thm]{Remark}

\crefname{lem}{Lemma}{Lemmas}
\crefname{lemma}{Lemma}{Lemmas}
\crefname{thm}{Theorem}{Theorems}
\crefname{theorem}{Theorem}{Theorems}
\crefname{defn}{Definition}{Definitions}
\crefname{definition}{Definition}{Definitions}
\crefname{fact}{Fact}{Facts}
\crefname{fact}{Fact}{Facts}
\crefname{clm}{Claim}{Claims}
\crefname{claim}{Claim}{Claims}
\crefname{prop}{Proposition}{Propositions}
\crefname{proposition}{Proposition}{Propositions}
\crefname{algocf}{Algorithm}{Algorithms}

\newcommand{\E}{\mathbb{E}}

\newcommand{\Ent}{\mathrm{Ent}}

\newcommand{\norm}[1]{\left\lVert #1 \right\rVert}
\newcommand{\ceil}[1]{\left\lceil #1 \right\rceil}
\newcommand{\floor}[1]{\left\lfloor #1 \right\rfloor}

\newcommand{\poly}{\mathrm{poly}}

\newcommand{\eps}{\varepsilon}

\newcommand{\N}{\mathbb{N}}

\newcommand{\R}{\mathbb{R}}

\renewcommand{\AA}{\mathcal{A}}
\newcommand{\BB}{\mathcal{B}}
\newcommand{\CC}{\mathcal{C}}
\newcommand{\DD}{\mathcal{D}}

\newcommand{\FF}{\mathcal{F}}

\newcommand{\II}{\mathcal{I}}

\newcommand{\KK}{\mathcal{K}}

\newcommand{\MM}{\mathcal{M}}

\newcommand{\PP}{\mathcal{P}}
\newcommand{\QQ}{\mathcal{Q}}

\renewcommand{\SS}{\mathcal{S}}
\newcommand{\TT}{\mathcal{T}}

\newcommand{\XX}{\mathcal{X}}

\renewcommand{\QQ}{\mathcal{K}}
\newcommand{\uu}{u}
\renewcommand{\complement}{\mathsf{c}}
\newcommand{\bb}{b}

\newcommand{\kl}[2]{D_{\mathrm{KL}}\left(#1 \,\Vert\, #2\right)}
\newcommand{\tv}[2]{d_{\mathrm{TV}}\left(#1 , #2\right)}
\newcommand{\chitwo}[2]{\chi^2\left(#1 \,\Vert\, #2\right)}

\newcommand{\bkl}{\varphi_{\mathrm{kl}}}
\newcommand{\Ber}{\mathrm{Ber}}
\newcommand{\maxcut}{\mathsf{max\text{-}cut}}

\newcommand{\coorora}{\mathsf{Coordinate\ Oracle}}
\newcommand{\subcora}{\mathsf{Subcube\ Oracle}}
\newcommand{\stanora}{\mathsf{General\ Oracle}}
\newcommand{\pairora}{\mathsf{Pairwise\ Oracle}}

\newcommand{\idtest}{\textsc{ID-TEST}}

\newcommand{\alg}{\mathsf{Alg}}
\newcommand{\ora}[1]{\mathsf{Ora}[#1]}

\begin{document}

\title{Complexity of High-Dimensional Identity Testing with Coordinate Conditional Sampling}

\author{
Antonio Blanca\thanks{Department of Computer Science and Engineering, Pennsylvania State University. Email: ablanca@cse.psu.edu. Research supported in part by NSF grant CCF-1850443.}
\and
	Zongchen Chen\thanks{School of Computer Science, Georgia Institute of Technology. Email: chenzongchen@gatech.edu.} 
	\and 
	Daniel \v{S}tefankovi\v{c}\thanks{Department of Computer Science, University of Rochester. Email: stefanko@cs.rochester.edu. Research supported in part by NSF grant CCF-1563757.}
	\and Eric Vigoda\thanks{Department of Computer Science, University of California, Santa Barbara. Email: vigoda@ucsb.edu. Research supported in part by NSF grant CCF-2147094.}
}
\date{\today}

\maketitle

\begin{abstract}%
	We study the identity testing problem for high-dimensional distributions. 
	Given as input an explicit distribution~$\mu$, an $\varepsilon>0$, and 
	access to sampling oracle(s) for a hidden distribution~$\pi$, the goal in identity testing is to distinguish whether the two distributions $\mu$ and $\pi$ are identical or are at least $\eps$-far apart.
	When there is only access to full samples from the hidden distribution $\pi$,
	it is known that exponentially many samples (in the dimension) may be needed for identity testing, and hence previous works have studied identity testing with additional access to various ``conditional'' sampling oracles.
	We consider a significantly weaker conditional sampling oracle, which we call the $\coorora$,
	and provide a computational and statistical characterization of the identity testing problem in this new model.
	
	We prove that if an analytic property known as approximate tensorization of entropy holds for an $n$-dimensional visible distribution $\mu$,
	then there is an efficient identity testing algorithm for any hidden distribution~$\pi$ using $\widetilde{O}(n/\eps)$ queries to the $\coorora$.
	Approximate tensorization of entropy is a pertinent condition as recent works have established it for a large class of high-dimensional distributions. 
	We also prove a computational phase transition: 
	for a well-studied class of $n$-dimensional distributions, specifically
	sparse antiferromagnetic Ising models over $\{+1,-1\}^n$, we show that in the regime where approximate tensorization of entropy fails, there is no
	efficient identity testing algorithm unless $\mathsf{RP}=\mathsf{NP}$.
	We complement our results with 
	a matching $\Omega(n/\eps)$
	statistical lower bound for the sample complexity of identity testing in the $\coorora$ model.
\end{abstract}

\maketitle

\section{Introduction}

A fundamental problem in statistics and machine learning is the identity testing problem (also known as the goodness-of-fit problem).  Roughly speaking, we are explicitly given a visible distribution $\mu$ and oracle access to samples from an unknown/hidden distribution $\pi$; the goal is to determine if these distributions are identical using as few samples from $\pi$ as possible.  

The complexity of identity testing for general distributions is now well-understood; this includes conditions on the visible and hidden distributions which enable efficient identity testing; see~\cite{Canonne-survey,Canonne22} for a comprehensive survey.
An intriguing line of work considers a different perspective: what additional assumptions \emph{on the sampling oracle} for the hidden distribution are required to ensure efficient identity testing. We present tight results with more modest oracle assumptions than considered previously.

Let us begin with a formal definition of the classical identity testing framework.
Let $\XX$ be a finite state space of size $N = |\XX|$, and let $d(\cdot, \cdot)$ denote a metric or divergence between distributions over $\XX$; the standard choices for $d(\cdot, \cdot)$ are total variation distance (TV distance) or Kullback–Leibler divergence (KL divergence).
For a distribution $\mu$ over $\XX$ and a parameter $\eps > 0$, denote by $\idtest(d,\eps; \mu)$ the identity testing problem for $\mu$: given as input the full description of the visible distribution $\mu$, and given access to a sampling oracle for an unknown distribution $\pi$, our goal is to distinguish between the cases $\pi = \mu$ vs.~$d(\pi,\mu) \ge \eps$
with probability at least $2/3$.  

For a distribution $\mu$ over $\XX$, there are efficient identity testing algorithms with sample complexity $O(\sqrt{N}/\eps^2)$ which
matches, asymptotically, the information-theoretic lower bound; see~\cite{VV,Paninski} for landmark results and \cite{Chan2014,ADK15,VV,DK16,Goldreich16,DGKPP21} for other relevant works. 
(We recall that the sample or query complexity of an identity testing algorithm is the number of queries it sends to the sampling oracle.) 

In practice, data is often high-dimensional, 
which raises the question of whether identity testing can be solved more effectively for high-dimensional distributions;
this will be our focus.
To be more precise, let $\QQ = \{1,\dots,k\}$ be a label (spin/color) set and let $\XX = \QQ^n$ be a product space of dimension~$n$. 
We study the identity testing problem $\idtest(d,\eps; \mu)$ for $n$-dimensional distributions $\mu$ over $\XX$.

Identity testing for high-dimensional distributions has recently attracted some attention, see, e.g.,~\cite{DP17,DDK19,testing-colt,testing-colt2,CDKS20,BGKV21,BCY22+}. 
The focus is on visible distributions $\mu$ that have a $\poly(n)$ size description or parametrization; otherwise one could not hope to design efficient testing algorithms.
Such distributions include product distributions (including the uniform distribution), Bayesian nets, and undirected graphical models (also known as spin systems) among~others.

The goal is to design identity testing algorithms with $\poly(n)$ sample complexity and running times.
It is known, however, that identity testing may require a super-polynomial (in $n$) number of samples \cite{testing-colt,testing-colt2}.
(The algorithms for the general identity testing problem have sample complexity $\Omega(k^{n/2}/\varepsilon^2)$ in the high-dimensional setting since $|\XX| = k^n$.)

Consequently, in order to design efficient algorithms, there are two types of further conditions that one may attach to the identity testing problem. 
The first approach is to restrict the unknown distribution~$\pi$ to be in some particular class of distributions;  a natural example is to require that $\pi$ is from the same class as $\mu$.
For example, \cite{BGKV21} studies the setting where both $\mu$ and $\pi$ are product distributions, \cite{CDKS20,DP17} requires $\mu$ and $\pi$ to be Bayesian nets, \cite{DDK19} studies the problem when $\mu$ and $\pi$ are Ising models. More recently, \cite{BCY22+} considers the case where $\mu$ is a product distribution, and $\pi$ is a Bayesian net.
While such an approach leads to fruitful results for testing high-dimensional distributions, 
it is not ideal from a practical perspective, where $\pi$ can be, for example, a ``noisy'' version of $\mu$ and may not necessarily belong to a nice class of distributions.

An alternative approach to overcome the apparent intractability of identity testing in the high-dimensional setting is to assume access to stronger sampling oracles from the hidden distribution~$\pi$; specifically, access to conditional sampling oracles for $\pi$ 
(in addition to the sampling oracle for $\pi$). 
This approach for high-dimensional distributions is the focus of this paper. 

There are several types of conditional sampling oracles, and here we mention the most popular choices.
The first is the general conditional sampling oracle---see~\cite{CRS,CFGM,FJOPS}---which given any subset $\mathcal{X}'$ of the space $\XX$ generates a sample from the projection of $\pi$ to $\mathcal{X}'$; that is, the oracle returns an element $x$ from $\mathcal{X}'$ with probability $\pi(x)/\pi(\mathcal{X}')$.
This oracle is not well-suited for the high-dimensional setting because the query subset $\mathcal{X}'$ could be exponentially large in $n$, and thus one could not hope to formulate the queries to the oracle efficiently (unless restricted to a special class of subsets~$\mathcal{X}'$). 

The second is the pairwise conditional sampling oracle ($\pairora$) which takes a pair of configurations and generates a sample from the distribution restricted to these two choices: given $x, y \in \XX$ the oracle returns $x$ with probability $\pi(x)/(\pi(x)+\pi(y))$ and $y$ otherwise; see~\cite{CRS}. The queries for $\pairora$ can be easily formulated for high-dimensional distributions, and identity testing has been studied in this setting. Recently,~\cite{Nar} provided an identity testing algorithm for the $\pairora$ model with $\widetilde{O}(\sqrt{n}/\varepsilon^2)$ sample complexity and a matching statistical lower bound; the $\widetilde{O}$ notation hides poly-logarithmic factors in $n$ and $1/\varepsilon$.

The other conditional oracle previously studied in the high-dimensional setting is the subcube conditional sampling oracle ($\subcora$) introduced by 
Bhattacharyya and Chakraborty
\cite{BC} and also studied in \cite{CCKLW,CJLW}. 
A query to the $\subcora$ consists of a subset
$\Lambda \subseteq [n] = \{1,\dots,n\}$ of variables and a configuration $x \in \QQ^{\Lambda}$ on $\Lambda$. If $\pi(x) > 0$, the $\subcora$ returns a sample $x' \in \QQ^{[n] \setminus \Lambda}$ from the conditional distribution $\pi(\cdot \,|\, x)$ (see~Definition~\ref{def:subcora}).
For the $\subcora$, an identity testing algorithm using $\widetilde{O}(n^2/\eps^2)$ queries was given in~\cite{BC};  improved algorithms were presented for uniformity testing in~\cite{CCKLW} and for testing juntas in~\cite{CJLW}.

In this work, we study identity testing for high-dimensional distributions under
a weaker conditional sampling oracle, which we call the $\coorora$.
The $\coorora$ corresponds to the $\subcora$ restricted to query sets $\Lambda$ where $|\Lambda|=n-1$; that is, we fix the configuration at all but one coordinate and look at the conditional distribution at this particular coordinate given a fixed configuration on the remaining coordinates. Hence, access to the $\coorora$ is a much weaker assumption than access to the $\subcora$. 
{
	We also note that the $\subcora$ model can be significantly harder to simulate. For instance, for the classical ferromagnetic Ising model simulating the $\coorora$ is trivial, but sampling conditionally on arbitrary configurations, as required by the $\subcora$, is 
	computationally hard~\cite{GJ-Ising}.}

Access to the $\coorora$ is also a weaker assumption than access to the $\pairora$ in the following sense.
When $k=2$ and $\XX=\{0,1\}^n$, 
$\coorora$ access corresponds to $\pairora$ access restricted to pairs of configurations that differ in \emph{exactly one coordinate}.
When $k\ge 3$, one can simulate an $\delta$-approximate $\coorora$ with $\pairora$ access in $\poly(k,\log(1/\delta))$ time (or a perfect one with $\poly(k)$ expected time) using a Markov chain; see~Remark~\ref{rmk:pair-coord} for the details.
{
	In addition, as in the case of the $\subcora$, simulating the $\pairora$ can be computationally more demanding than simulating the $\coorora$. For example, 
	in the context of the ferromagnetic Ising model on an $n$-vertex bounded degree graphs, 
	a query to the $\coorora$ will require $O(1)$ random bits, but 
	queries to the $\pairora$ may require $\Omega(n)$ random bits.}

We provide a computational and statistical characterization of the identity testing problem in the $\coorora$ model.
Our focus is on imposing no conditions on the hidden distribution $\pi$, other than access to $\coorora$, and explore which conditions on the visible distribution $\mu$ are necessary and sufficient for identity testing.
We mention that the $\coorora$ oracle has already been implicitly used in \cite{CCKLW} for uniformity testing (i.e., the special case of testing whether $\pi$ is the uniform distribution).

\medskip
\noindent\textbf{Algorithmic results.\ }
For our algorithmic work we consider the identity testing problem
under KL divergence, which we denote by $\kl{\cdot}{\cdot}$ and is formally defined in Section \ref{sec:prelim}. 
From an algorithmic perspective, the choice of KL divergence is a natural one since, by Pinsker's inequality, 
a testing algorithm for $\idtest(\kl{\cdot}{\cdot},2\eps^2; \mu)$ yields one for $\idtest(\tv{\cdot}{\cdot},\eps; \mu)$ (i.e., for identity testing under TV distance) albeit with potentially sub-optimal sample complexity and running time; the reverse is not true in general.

We start by introducing a key analytic property for the visible distribution, known as {\em approximate tensorization of entropy}~\cite{CMT}, which we will show is a sufficient (and essentially also necessary) condition for efficient identity testing in the high-dimensional setting.
Approximate tensorization of entropy roughly states that the entropy of a distribution is bounded by the sum of the average conditional entropies at each coordinate.  

\begin{definition}[Approximate Tensorization of Entropy]
	\label{def:AT}
	A distribution $\mu$ fully supported on $\KK^n$ satisfies \emph{approximate tensorization of entropy} with constant $C$ if 
	for any distribution $\pi$ over~$\QQ^n$: 
	\begin{equation}
	\label{eqn:AT-defn}
	\kl{\pi}{\mu} \le C \sum\nolimits_{i=1}^n \E_{x \sim \pi_{n \setminus i}} \Big[ \kl{\pi_i(\cdot \mid x)}{\mu_i(\cdot \mid x)} \Big], 
	\end{equation}
	where 
	$\pi_{n \setminus i}(\cdot)$ denotes the marginal distribution of $\pi$ on $[n] \setminus \{i\}$, 
	and $\pi_i(\cdot \mid x)$ and $\mu_i(\cdot \mid x)$ denote the marginals of $\pi$ and $\mu$, respectively, on the $i$-th coordinate conditional on $x$.
\end{definition}

\noindent
The constant $C$ achieves the minimum $C=1$ when $\mu$ is a product distribution. More details about approximate tensorization and equivalent formulations are provided in~Section \ref{sec:AT-defn}. 

Approximate tensorization of entropy is known to imply optimal mixing times of single-site update Markov chains, known as the Gibbs sampler or Glauber dynamics~\cite{Cesi,CLV-STOC21}. It is also used to establish modified log-Sobolev inequalities and the concentration of Lipschitz functions under the distribution~\cite{BGmlsi,MTbook}. 

There are a plethora of recent results establishing approximate tensorization in a wide variety of settings. In particular, \cite{CLV20} showed that the spectral independence condition introduced by \cite{ALO20} implies approximate tensorization of entropy for 
sparse undirected graphical models (i.e., spin systems on bounded degree graphs). Furthermore, recent works showed that spectral independence (and hence approximate tensorization) is implied by certain forms of correlation decay~\cite{CLV20,CGSV21,FGYZ21}, path coupling for local Markov chains~\cite{BCCPSV22,Liu}, and the stability of the partition function~\cite{CLV-FOCS21}.  
As such, approximate tensorization is now known to hold with constant $C = O(1)$ (independent of $n$) for a variety of high-dimensional distributions; see, e.g.,~\cite{CLV-FOCS21,BCCPSV22,Liu,chen2022algorithms,galanis2022fast,friedrich2022spectral}.

We show that approximate tensorization of the visible distribution $\mu$ yields an efficient identity testing algorithm, 
provided access to the $\coorora$ and the $\stanora$ for the hidden distribution~$\pi$. 
Access to the $\stanora$ oracle (i.e., to independent full samples from $\pi$) is a standard assumption for testing under conditional sampling oracles.
In particular, the $\stanora$ corresponds to $\subcora$ restricted to $\Lambda=\emptyset$,
so access to the $\subcora$ implies access to the $\stanora$, and previous work under the $\pairora$ assumes access to the $\stanora$ as well. 

For our algorithmic result, we have four additional basic assumptions on the visible distribution~$\mu$. Specifically, we require that:
\begin{enumerate}[(i)]
	\setlength{\itemsep}{0pt}
	\item $\mu$ has a description (parametrization) of $\poly(n)$ size;
	\item the $\coorora$ can be implemented efficiently for the \emph{visible} distribution $\mu$; 
	\item $\mu$ is $\eta$-{\emph balanced}: there is a lower bound $\eta$ so that the conditional probability of any label $a\in \QQ$ at any coordinate $i$, fixing any configuration on $[n] \setminus \{i\}$, is at least $\eta$ (see Section~\ref{def:balanced});
	\item $\mu$ is fully supported on $\QQ^n$. 
\end{enumerate}

\noindent
We discuss these assumptions in detail below (see Remark~\ref{rmk:assumptions}). Our algorithmic result for the $\coorora$ model follows.

\begin{theorem}
	\label{thm:alg-main}
	Given a distribution $\mu$ over $\XX = \QQ^n$ satisfying (i)-(iv) and \emph{Approximate Tensorization} with constant $C$, there is a testing algorithm 
	for $\idtest(\kl{\cdot}{\cdot},\eps; \mu)$
	with access to the $\coorora$ and $\stanora$ with $\widetilde{O}(n/\eps)$ sample complexity and polynomial running time.
\end{theorem}

\noindent   
We refer the reader to Theorem \ref{thm:alg-main-detailed} for a more precise theorem statement indicating the explicit dependence on $C$ and $\eta$ in the sample complexity.  See also Section~\ref{subsec:application} for applications of Theorem \ref{thm:alg-main} to several well-studied high-dimensional distributions.

We shall see in what follows that our algorithmic result for the $\coorora$ model in Theorem \ref{thm:alg-main} is tight, both statistically and computationally; that is, we establish a matching $\Omega(n/\eps)$ sample complexity lower bound and show that there is a class of high-dimensional distributions
where identity testing is computationally hard in exactly the same settings where approximate tensorization of entropy does not hold. 
(These results also hold for the easier problem of testing under TV; see Theorems~\ref{thm:comp-hardness} and~\ref{thm:lb-coor-tv}.)

A surprising feature of our algorithm is that it bypasses sampling from visible distribution $\mu$; it does not even require the concentration of any statistics under $\mu$. As in some of the previous algorithms for high-dimensional testing---e.g., those in~\cite{DDK19,CCKLW}---
our algorithm starts by ``localizing'' the testing problem (i.e., reducing it 
to a one-dimensional setting).
For this, we crucially use the Approximate Tensorization of entropy of the visible distribution.

We then consider the problem of testing general (one-dimensional) distributions under KL divergence.
It turns out that this problem has been largely overlooked in the literature (the aforementioned known results 
for identity testing are all under TV distance). This is likely because there are pairs of distributions with infinite KL divergence but arbitrarily small TV distance, and so testing under KL divergence is considered unsolvable in a worst-case sense;~see \cite{DKW18}.
However, we can aim for algorithms with sample complexities that depend on the visible distribution; i.e., instance-specific bounds instead of worst-case ones, as done in~\cite{VV,DK16,BCG}. 

We provide here an algorithm for the classical identity testing problem (that is, only access to the $\stanora$ is assumed) for general distributions under KL divergence; the sample complexity of our algorithm depends on the visible distribution (see~Lemma \ref{lem:KL-id-intro}).
This is a key technical development towards establishing Theorem \ref{thm:alg-main}, and one we believe could be of independent interest.

\begin{remark}
	\label{rmk:assumptions}
	    We pause now to discuss assumptions (i)-(iv) in Theorem~\ref{thm:alg-main}. 
	As mentioned, condition (i) is necessary as otherwise one can not hope to design testing algorithms with $\poly(n)$ running times.
	Condition (ii) formally states that for any coordinate $i$, and any fixed assignment $\sigma$ for the other $n-1$ coordinates, we can compute the conditional distribution at $i$ given $\sigma$ in polynomial time.
	This is equivalent to requiring that a step of the Gibbs Sampler Markov chain for $\mu$ can be implemented efficiently; we believe
	(ii) is a mild~assumption.
	
	The notion of $\eta$-balancedness 
	in condition (iii) is a byproduct of working with KL-divergence and is closely related to other coordinate marginal conditions 
	that are required for efficient learning and sampling; specifically, under the assumption that $\mu$ has full support, it is equivalent to the notions $\delta$-biased 
	in~\cite{klivans2017learning} and of $b$-marginally bounded distributions from~\cite{CLV-FOCS21,BCCPSV22}.
	Finally, we note that condition (iv) is also a byproduct of working with KL divergence but can be relaxed; we could require instead that the support of $\pi$ is a subset of the support of $\mu$. We emphasize that these conditions are all for the visible distribution $\mu$, and that we impose no restrictions on the hidden distribution $\pi$ (other than oracle access).
	For example, the uniform distribution, product distributions, and undirected graphical models (e.g., the Ising and hard-core models) satisfy the conditions in Theorem~\ref{thm:alg-main}.
\end{remark}

\noindent\textbf{Computational hardness results.\ }
We show next that the algorithmic result in Theorem \ref{thm:alg-main} is computationally tight. 
In particular, for the antiferromagnetic Ising model (defined below), we establish the following computational phase transition for identity testing in the $\coorora$ model: (i) when approximate tensorization holds the problem can be solved efficiently, and (ii) when approximate tensorization does not hold, there is no polynomial-time testing algorithm unless $\mathsf{RP}=\mathsf{NP}$. 

We do not directly prove that identity testing is hard when approximate tensorization fails.
We show instead that the same strong correlations that cause approximate tensorization to fail, 
combined with the hardness of identifying the ground states of the model in the presence of strong correlations, imply the hardness of identity testing.
(The ground states are the most likely configurations in the model, and for the antiferromagnetic Ising model correspond to the maximum cuts of the graph.)

We introduce the Ising model next, which is the simplest and most well-studied example of an undirected graphical model.  
Given a graph $G=(V,E)$, the set of configurations of the model is denoted by $\Omega=\{+1,-1\}^V$.  For a real-valued parameter~$\beta$, 
the probability of a configuration $\sigma\in\Omega$ is given by the Gibbs or Boltzmann distribution:
\begin{equation}
\label{eq:ising:def}
\mu_{G,\beta}(\sigma) = \frac{1}{Z_{G,\beta}} \cdot \exp\Big(\beta\sum\nolimits_{\{v,w\}\in E} \sigma_v\sigma_w\Big),
\end{equation}
where the normalizing constant $Z_{G,\beta}$ is known as the partition function.
When $\beta>0$ the model is ferromagnetic/attractive and when~$\beta<0$ then the model is antiferromagnetic/repulsive; see~Section~\ref{subsec:sparse:ising} for a more general definition of the model.

The antiferromagnetic Ising model undergoes an intriguing computational phase transition at the threshold $\beta_c(d) = -\frac{1}{2} \ln (\frac{d}{d-2})$ for the parameter $\beta$.
This threshold corresponds to the so-called uniqueness/non-uniqueness phase transition on the infinite $d$-regular tree defined as follows.
Let $p^+_\ell$ denote the marginal probability that the root of the complete $d$-regular tree of depth $\ell$ (i.e., the tree where all internal vertices have degree $d$ and all leaves are on the same level)
has label/spin $+1$ when one fixes the leaves to the all $+1$ configuration. 
Similarly, let $p^-_\ell$ denote the analogous marginal probability for the root to be $+1$ when the leaves are instead fixed to the all $-1$ configuration. When $\beta<\beta_c(d)$, then in the limit as $\ell\rightarrow\infty$ the two marginals are the same, i.e., $\lim_{\ell\rightarrow\infty} p^+_\ell = \lim_{\ell\rightarrow\infty} p^-_\ell$; this is known as the {\em tree uniqueness region} since implies that there is a unique Gibbs distribution for the infinite $d$-regular tree.  On the other hand, when $\beta>\beta_c(d)$ then the limits are different; this is called the {\em tree non-uniqueness region} as there are multiple Gibbs distributions the infinite $d$-regular tree.
A key consequence for general graphs is the following rough statement: in graphs of maximum degree at most $d$, when $\beta < \beta_c(d)$ long-range correlations die off, whereas 
when $\beta > \beta_c(d)$ long-range correlations persist and marks the onset of hardness for several computational problems (e.g., counting and sampling) on graphs of degree at most $d$; see~\cite{sly2008uniqueness} for further details.

For constant $d \ge 3$ and all $0>\beta>\beta_c(d)$, 
the 
approximate sampling and counting (i.e., approximating the partition function $Z_{G,\beta}$) problems can be solved efficiently on any graph of maximum degree $d$~\cite{CLV-STOC21}. 
Moreover, approximate tensorization holds in this regime, and hence Theorem \ref{thm:alg-main} applies for identity testing in the $\coorora$ model.
In contrast, it is also known that when $\beta<\beta_c(d)$ there are no polynomial-time approximate sampling or counting algorithms unless $\mathsf{RP}=\mathsf{NP}$~\cite{SlySun,GSV:ising}.  

We establish here the computational hardness of identity testing in the $\coorora$ model in the same parameter regime $\beta < \beta_c(d)$, which thereby exhibits a similar computational phase transition for identity testing for the class of antiferromagnetic Ising models. 

\begin{theorem}
	\label{thm:comp-hardness}
	For sufficiently large constant $d \ge 3$ and constant $\beta<0$, 
	consider identity testing for the family of antiferromagnetic Ising models on $n$-vertex graphs of max degree $d$ with parameter~$\beta$. 
	\begin{enumerate}[(i)]
		\setlength{\itemsep}{0pt}
		\item If $\beta > {\beta_c(d)}$, then there exists a polynomial-time algorithm 
		for identity testing under KL divergence with access to the $\coorora$ and the $\stanora$ with sample complexity $\widetilde{O}(n/\eps)$;
		\item If $\beta < \beta_c(d)$, then there is no polynomial-time algorithm for identity testing under TV distance (and hence under KL divergence) with access to the $\coorora$ and the $\stanora$ unless $\mathsf{RP}=\mathsf{NP}$.
	\end{enumerate}
\end{theorem}

\noindent
There are few analogous computational hardness results for identity testing; most lower-bound results in this setting are information-theoretic. The few examples appeared in~\cite{testing-colt,testing-colt2}, and 
these earlier results apply to the identity testing problem 
with access only to $\stanora$ and require both the hidden and visible models to be Ising models. 
In our current setting, the visible model is an Ising model, but the hidden is an arbitrary high-dimensional distribution. This is a significant conceptual difference, and the techniques from~\cite{testing-colt,testing-colt2} do not easily extend (see~Remark~\ref{rmk:diff}).

At a high level, as in~\cite{testing-colt}, we prove the hardness result in Theorem \ref{thm:comp-hardness}(ii) using a reduction from the maximum cut problem. That is, given a graph $G=(V,E)$, we construct a testing instance that if solved, would find the maximum cut of $G$.
In this approach, constructing a testing instance of small degree is a key challenge, and the ``degree reducing'' gadgets from~\cite{testing-colt,testing-colt2} no longer work in our setting.

Instead, we use a gadget introduced in \cite{Sly} to establish the computational hardness of approximate counting antiferromagnetic spin systems. 
An interesting technical aspect of our proof is that we are required to design polynomial-time sampling algorithms to simulate the hidden oracles.
This is difficult for us because sampling antiferromagnetic Ising models throughout the non-uniqueness regime, i.e., for all $\beta < \beta_c(d)$, is a notoriously hard problem (the problem is $\mathsf{NP}$-hard even for regular graphs).
We manage to design efficient sampling algorithms for our testing instances
using the recent algorithmic result of Koehler, Lee, and Risteski~\cite{KLR} that give an approximate sampling algorithm for Ising models when the edge interaction matrix has low rank, in conjunction with the sampling methods from~\cite{JKP} that use polymer models. 
A detailed overview of our reduction is given in Section~\ref{subsec:hardness}.
We mention that in the reductions in~\cite{testing-colt,testing-colt2}, sampling is trivial, since there it is assumed that $\beta \ll \beta_c(d)$ (specifically, $|\beta| d=\Omega(\log{n})$) and the instance is bipartite, so the Gibbs distribution concentrates in the configurations that align with the bi-partition; see Remark~\ref{rmk:diff} for a detailed account of the novelties in our reduction to establish Theorem \ref{thm:comp-hardness}(ii).

Finally, we mention that the hardness result in~Theorem \ref{thm:comp-hardness}(ii) extends to \emph{any} conditional sampling oracle that could be implemented in polynomial time for the antiferromagnetic Ising model, and thus applies to identity testing in the $\pairora$ model, complementing the algorithmic results from~\cite{CRS,Nar}. On the other hand, they do not extend to the $\subcora$ model
since we do not know how to simulate this oracle efficiently.

\medskip
\noindent\textbf{Statistical lower bounds.\ }
We present next an information-theoretic lower bound for identity testing problem in the $\coorora$ model that matches the sample complexity of our testing algorithm for this model (Theorem \ref{thm:alg-main}). 
Our lower bound is for the special case of uniformity testing under TV distance when $k=2$; 
i.e., the visible distribution is the uniform distribution over~$\{0,1\}^n$.

\begin{theorem}
	\label{thm:lb-coor-tv}
	Let $\mu$ be the uniform distribution over $\{0,1\}^n$.
	Then, any algorithm for the problem $\idtest(\tv{\cdot}{\cdot},\eps; \mu)$ with access to both the $\coorora$ and the $\stanora$ requires $\Omega(n/\eps^2)$~samples. 
\end{theorem}

\noindent
A direct corollary of this result is that solving the identity testing problem under KL divergence requires $\Omega(n/\eps)$ samples in the $\coorora$ model, thus showing that the sample complexity of our algorithm in Theorem \ref{thm:alg-main} is asymptotically tight (up to logarithmic in $n$ and $1/\varepsilon$ factors).

Our proof of Theorem \ref{thm:lb-coor-tv} follows a well-known strategy. We construct a family of ``bad'' distributions~$\BB$, each of which has TV distance at least~$\eps$ from the uniform distribution $\mu$ over $\{0,1\}^n$. The lower bounds follow from the fact that, for this carefully constructed family $\BB$,
one can not distinguish between sequences of independent samples from $\mu$ or from a distribution $\pi$ chosen uniformly at random from $\BB$. 
However, since our setting is adaptive, i.e., the choice of conditional queries of the testing algorithm may depend on the output to previous ones,
we need to consider \emph{query histories}, as in \cite{CRS,Nar}. 
(Roughly speaking, a query history is a sequence of queries that the testing algorithm sends the oracle along with the outputs from the oracle.)
To show that two query histories are indistinguishable (under $\mu$ or $\pi$), we use ideas from \cite{CRS} and the so-called \emph{hybrid argument} in cryptography; see~\cite{hybrid}. 

\medskip
\noindent\textbf{New results for the $\subcora$ model.\ } While the main focus of this work is the study of identity testing under weaker oracle assumptions (i.e., the $\coorora$ model),
we also provide new results for identity testing in the previously studied $\subcora$ model. Our first results for this model is an improved identity testing algorithm.

\begin{theorem}
	\label{thm:alg-subcora}
	Let $\mu$ be an $\eta$-balanced distribution fully supported on $\KK^n$ that has a $\poly(n)$ size parameterization. If we can compute the marginal probability at any coordinate conditioned on any partial configuration on any subset of coordinates, then there is an identity testing algorithm for $\idtest(\kl{\cdot}{\cdot},\eps; \mu)$ for the $\subcora$ model with $\widetilde{O}(n/\eps)$ sample complexity and running time that depends on the time it takes to compute the coordinate conditional marginals.
\end{theorem}

This algorithm, compared to the best-known algorithm for the $\subcora$ model in~\cite{BC}, additionally requires that $\mu$ is $\eta$-balanced,
but improves the sample complexity significantly from $\widetilde{O}(n^2/\eps^2)$ to $\widetilde{O}(n/\eps)$.
In addition, compared to Theorem \ref{thm:alg-main}, this result for the $\subcora$ does \emph{not} require approximate tensorization of entropy.
In fact, we point out several relevant settings where Theorem \ref{thm:alg-subcora} applies, but approximate tensorization fails (or we do not know if it holds) and hence Theorem \ref{thm:alg-main} does not apply: 
undirected graphical models (e.g., Ising model) on trees, Bayesian networks, mixtures of product distributions, and high-temperature Ising models and monomer-dimer models (i.e., weighted matchings) on arbitrary graphs. 
We remark that, similar to Theorem \ref{thm:alg-main}, Theorem \ref{thm:alg-subcora} also holds under the weaker assumption that the support of $\mu$ contains the support of $\pi$;
see Theorems~\ref{thm:alg-subc} and~\ref{thm:robust-alg-subc} for more details.

We also provide a matching lower bound for uniformity testing in the $\subcora$. 

\begin{theorem}
	\label{thm:lb-subc-kl}
	Let $\mu$ be the uniform distribution over $\{0,1\}^n$.
	Then, any algorithm for the problem $\idtest(\kl{\cdot}{\cdot},\eps; \mu)$ with access to the $\subcora$ requires $\Omega(n/\eps)$ samples. 
\end{theorem}

\noindent
Note that when $\mu$ is the uniform distribution, then $\eta=\Theta(1)$, so
Theorems \ref{thm:alg-subcora} and \ref{thm:lb-subc-kl} provide asymptotically matching sample complexity bounds for identity testing under KL divergence.
Interestingly, if one considers uniformity testing under TV distance and $\subcora$ access, then the recent work \cite{CCKLW} shows that $\widetilde{O}(\sqrt{n}/\eps^2)$ oracle queries suffice. 
As far as we know, it is unclear if the testing algorithm from \cite{CCKLW} with sublinear sample complexity can be used for other high-dimensional distributions, e.g., general product distributions.

Our last result concerns \emph{tolerant} identity testing in the $\subcora$ model.
In this problem, the goal is to distinguish between the cases $\kl{\pi}{\mu} \le \delta$ and $\kl{\pi}{\mu} \ge \delta+\eps$ for $\delta,\eps > 0$;
identity testing corresponds to $\delta = 0$.
We show that, under the same assumptions as in Theorem \ref{thm:alg-subcora}, one can estimate $\kl{\pi}{\mu}$ within additive error $\eps$ using $\widetilde{O}(n^4/\eps^4)$ queries to the $\subcora$. 

\begin{theorem}
	\label{thm:intro:kl-estimate-subc}
	Let $\mu$ be an $\eta$-balanced visible distribution fully supported on $\KK^n$ that has a $\poly(n)$ size parametrization. 
	Suppose we can compute the marginal probability for $\mu$ at any coordinate conditioned on any partial configuration on a subset of coordinates. 
	Given access to the $\subcora$ for a hidden distribution $\pi$, 
	there is an algorithm that for any $\eps > 0$ computes $\widehat{S}$ such that, with probability at least $2/3$, we have
	$
	| \widehat{S} - \kl{\pi}{\mu} | \le \eps.
	$
	The sample complexity of the algorithm is 
	$
	\widetilde{O}\left( {n^4}/{\eps^4} \right). 
	$
	The running time of the algorithm depends on the time it takes to compute the coordinate conditional marginals. 
\end{theorem}

\section{Overview of Techniques}

We present proof overviews for our main results in the $\coorora$ model:
our testing algorithm (Theorem \ref{thm:alg-main}), the computational hardness (Theorem \ref{thm:comp-hardness}(ii)), and the lower bound (Theorem \ref{thm:lb-coor-tv}).

\subsection{Algorithmic result for \texorpdfstring{$\coorora$}{Coordinate Oracle} model:~\texorpdfstring{Theorem \ref{thm:alg-main}}{Theorem 1.2}}
\label{subsec:ov-alg}

Suppose $\mu$ is the visible distribution and let $\pi$ be an arbitrary distribution over $\QQ^n$. 
If approximate tensorization of entropy holds for $\mu$ with constant $C$, the following holds:
\[
\kl{\pi}{\mu} \le Cn \; \E_{(i,x)} \left[ \kl{p^x_i}{q^x_i} \right],
\]
where $i \in [n]$ is a uniformly random coordinate, $x \in \QQ^{n\setminus i}$ is generated from the marginal distribution $\pi_{n \setminus i}$ of $\pi$ on $[n] \setminus \{i\}$, $p^x_i = \pi_i(\cdot \mid x)$, and $q^x_i = \mu_i(\cdot \mid x)$ (see~Definition~\ref{def:AT}).
Therefore, to distinguish between the cases $\pi = \mu$ and $\kl{\pi}{\mu}\ge \eps$, it suffices to distinguish between:
\[
p^x_i = q^x_i \text{~for all pairs $(i,x)$}
\quad \text{vs.} \quad
\E_{(i,x)} \left[ \kl{p^x_i}{q^x_i} \right] \ge \frac{\eps}{Cn}.
\]
This is the first step towards localizing the testing problem to a single coordinate.
Now, under the $\eta$-balanced assumption for $\mu$, we have that $0 \le \kl{p^x_i}{q^x_i} \le \ln(1/\eta)$.
Hence, if $\E_{(i,x)} \left[ \kl{p^x_i}{q^x_i} \right] \ge \frac{\eps}{Cn}$, 
one can show via a reverse Markov inequality that there exists an integer $\ell \ge 1$, such that $2^\ell = O(n)$ and
\begin{equation}\label{eq:rev:mark}
{\Pr}_{(i,x)}\big( \kl{p^x_i}{q^x_i} \ge  2^{\ell} \cdot \frac{\eps}{2Cn} \big) \ge \frac{1}{2^\ell D},
\end{equation}
where $D = \Theta\big(\log(\frac{n \cdot \log(1/\eta)}{\varepsilon})\big)$; see~Lemma \ref{lem:reverse-Markov} for a precise statement.

With~\eqref{eq:rev:mark}, it is not difficult to find a pair $(i,x)$ such that $\kl{p^x_i}{q^x_i} \ge  2^{\ell}\cdot \frac{\eps}{2 Cn}$. 
This can done by first generating $O(\poly(n) \cdot D)$ 
pairs $(i_t,x_t)$, by choosing $i_t \in [n]$ uniformly at random and then using the
$\stanora$ to sample the partial configuration $x_t$ on $[n]\setminus\{i_t\}$.
Then, we can exhaustively check for each $\ell$ (note that $\ell = O(\log n)$) whether among the generated pairs $(i_t,x_t)$'s there is one, say $(i,x)$, satisfying that $\kl{p^x_i}{q^x_i} \ge  2^{\ell}\cdot\frac{\eps}{2 Cn}$. By~\eqref{eq:rev:mark}, this will likely be the case.

In conclusion, we reduce (or localized identity testing for $\mu$ to solving identity testing for the one-dimensional distributions $p:=p^x_i$ and $q:=q^x_i$ on a domain of size $k$ with respect to KL divergence. 
We assume $q$ can be computed for the visible distribution $\mu$ (this is assumption (iii) in Theorem \ref{thm:alg-main}), and we have access to a sampling oracle for $p^x_i$ using the $\coorora$ for~$\pi$. 

As mentioned earlier, in the distribution testing literature, testing under KL divergence has been overlooked. This is because there are pairs of distributions with infinite KL divergence but arbitrarily small TV distance, which entails that identity testing requires arbitrarily many samples even though the KL divergence is arbitrarily large.
For example, this happens when $q$ is the distribution on a single point $0$ and $p$ is the Bernoulli distribution with arbitrarily small mean~\cite{DKW18}.
As such, identity testing under KL divergence has been considered unsolvable in the sense of worst-case sample complexity for arbitrary $p$ and~$q$. 

However, the identity testing problem under KL divergence makes perfect sense for specific visible distributions $q$
if we are interested in the instance-specific sample complexity instead of the worst-case one, as in~\cite{VV} under TV distance. Namely, for a given distribution $q$, what is the number of samples required, potentially depending on $q$, for the identity testing problem for $q$ under KL divergence? 
We give next a first attempt at solving this problem. The sample complexity of our testing algorithm depends on the minimum probability $\eta = \min_{a \in \KK} q(a)$. 

\begin{lemma}
	\label{lem:KL-id-intro}
	Let $k \in \N^+$ 
	and let $\eps > 0$, $\eta \in (0,1/2]$. 
	Given a visible distribution $q$ over domain $\QQ$ of size $k$ such that $q(a) \ge \eta$ for any $a \in \KK$, and given sample access to an unknown distribution $p$ over $\QQ$, there exists a polynomial-time identity testing algorithm that distinguishes with probability at least $2/3$ between the cases
	$p = q$ or $\kl{p}{q} \ge \eps$.
	The sample complexity of the identity testing algorithm is
	$
	O\big(\min\big\{ \frac{1}{\eps \sqrt{\eta}}, \frac{\sqrt{k} \ln(1/\eta)}{\eps^2} \big\} \big)
	$
	for $k \ge 3$ and 
	$
	O\big( \frac{\ln(1/\eta)}{\eps} \big) 
	$
	for $k=2$.
\end{lemma}

\noindent
We remark that the dependency on $\eta$ in the sample complexity is inevitable; see Remark~\ref{rmk:eta:dep}.

A natural first approach to identity testing under KL divergence to prove~Lemma \ref{lem:KL-id-intro} is a reduction to testing under TV distance via the so-called reversed Pinsker's inequality: $\kl{p}{q} \le (2/\eta)\tv{p}{q}^2$ (see Lemma \ref{lem:kl-chi2-l2}). 
The sample complexity of such algorithm is $O(\sqrt{k}/(\eps\eta))$. This is not optimal: for example, if $q$ is the uniform distribution over $\QQ$ one has $\eta = 1/k$ and so the sample complexity is $O(k^{3/2}/\eps)$, but one would expect the sample complexity to be $O(\sqrt{k}/\eps)$,
by analogy to what happens for testing in TV distance. 
A better reduction is to testing under $\ell_2$ distance via the inequality: $\kl{p}{q} \le (1/\eta) \norm{p-q}_2^2$ (see Lemma \ref{lem:kl-chi2-l2}). 
We then need to distinguish between $p=q$ and $\norm{p-q}_2 \ge \sqrt{\eps\eta}$, which allows us to apply results from \cite{DK16} and obtain an algorithm with sample complexity of $O(\norm{q}_2/(\eps\eta))$; this time we get the $O(\sqrt{k}/\eps)$ sample complexity bound when $q$ is the uniform distribution. 

However, two major challenges ought to be solved for this approach to work.
First, while $\norm{q}_2$ can be bounded for certain specific distributions $q$ (e.g., the uniform distribution), in general we do not have a bound for $\norm{q}_2$.
This can be solved via the \emph{flattening} method from \cite{DK16} which, roughly speaking, constructs a new testing instance (i.e., distributions $p'$ and $q'$ over $\QQ'$) that is equivalent to the initial one with the additional property that $\norm{q'}_2$ is small. 
The idea is to divide ``heavy'' elements (those $a\in\QQ$ with large density $q(a)$) into many copies so that $q'(a') \approx 1/\ell$ for all $a' \in \QQ'$ where $\ell = |\QQ'|$; i.e., $q'$ is close to~uniform. 

The second challenge is that, even if when $\norm{q}_2$ small, the dependency on $\eta$ could still be inverse-polynomial. This is particularly problematic when $\eta$ decays sharply as $k$ grows; e.g., $\eta = 2^{-k}$. 
To overcome this, we divide the domain $\QQ$ into two parts, those with small density $q(a) < \zeta$ and those with large density $q(a) \ge \zeta$ for some parameter $\zeta$. We then deal with the two parts separately by running different testing algorithms on each. 
This can be viewed as a simple application of the \emph{bucketing} technique from~\cite{BFFKRW} using only two buckets.

While flattening and bucketing were previously known, the novelty of our approach is to combine them to get a stronger bound for the sample complexity, specifically by selecting the right scale $\ell$ for flattening and the right threshold $\zeta$ for bucketing. 
Our bound achieves $O(\sqrt{k}/\eps)$ for $q$ with $\eta = \Theta(1/k)$ such as the uniform distribution, and also maintains a $\sqrt{k}$ dependency even for biased $q$ of tiny $\eta$, with only a logarithmic dependency on $1/\eta$. 
For details see Lemmas \ref{lem:KL-id-test} and \ref{lem:Ber-KL-id-test}.

\subsection{Computational hardness in the \texorpdfstring{$\coorora$}{Coordinate Oracle} model:~\texorpdfstring{Theorem \ref{thm:comp-hardness}(ii)}{Theorem 1.3}}
\label{subsec:hardness}

We establish hardness of the identity testing problem as stated in Theorem \ref{thm:comp-hardness}(ii) for the antiferromagnetic Ising model with $\coorora$ and $\stanora$ access via a reduction from the maximum cut problem.
Let $\{G=(V_{G},E_{G}),k\}$ be an instance of the maximum cut problem. That is, we want to check whether $\maxcut(G) < k$ or $\maxcut(G) \ge k$.
In our reduction,
we construct an identity testing instance for the antiferromagnetic Ising model, feed it as input to a presumed testing algorithm, and claim that the output of the algorithm solves $\{G=(V_{G},E_{G}),k\}$ with probability at least $2/3$; this is not possible unless $\mathsf{RP}=\mathsf{NP}$.

We start by constructing the multi-graph $F = (V_F,E_F)$ by adding two special vertices $s$ and $t$ to $G$; i.e., $V_F = V_G \cup \{s,t\}$.
These two vertices are connected
with $N^2-k$ edges, where $N = |V_G|$. We also add $N$ edges between $s$ and each vertex of $V_G$ and do the same for $t$ so that:
\begin{enumerate}
	\setlength{\itemsep}{0pt}
	\item When $\maxcut(G) < k$, the cut $(\{s,t\},V_G)$ of size $2N^2$ is the unique maximum cut of $F$; 
	\item When $\maxcut(G) \ge k$, there exists another cut in $F$, other than $(\{s,t\},V_G)$, of size~$\ge 2N^2$.
\end{enumerate}
This is because for $S \subset V_G$, the cut $(S \cup \{s\},V_G\setminus S \cup \{t\})$ of $F$ will have size:
$\maxcut(G) + |S|N + |V_G\setminus S|N + N^2-k = 2N^2 + \maxcut(G) - k, 
$
which is $\ge 2N^2$ only when $\maxcut(G) \ge k$.

We consider the antiferromagnetic Ising model on $F$.
There is a natural bijection between the cuts of $F$ and the configurations of the Ising model.
In particular, each cut $(S,V_F\setminus S)$ of $F$ corresponds to exactly two Ising configurations: vertices in $S$ are assigned $+1$ and those in $V_F\setminus S$ are assigned $-1$ (and vice versa). From the definition of the model (see~\eqref{eq:ising:def}) we also see that the ``ground states'' of the antiferromagnetic Ising model on $F$, that is the configurations of maximum probability in the Gibbs distribution, correspond precisely to the maximum cuts of $F$.

Let $\Omega$ be the set of all cuts of $F$ and let $\Omega_0$ be the set of all cuts $(S,V_F\setminus S)$ of $F$ 
\emph{except} those where $s \in S$, $t \in V_F\setminus S$, and the corresponding cut for $G$, i.e., $(S\setminus\{s\},V_F\setminus \{S,t,s\})$, has size $\ge k$.
This way, if $\maxcut(G) < k$, then $\Omega_0 = \Omega$, and
if $\maxcut(G) \ge k$, then $\Omega \setminus \Omega_0$ contains
the cuts of $F$ corresponding to cuts of $G$ of size $\ge k$.

We set the visible distribution of our testing instance to be the Gibbs distribution $\mu_{{F},\beta}$ of the antiferromagnetic Ising model on ${F}$ with $\beta < \beta_c(d) < 0$ in the tree non-uniqueness region.
The hidden distribution will be $\mu_{{F},\beta} (\cdot \mid  \Omega_0)$; that is, $\mu_{{F},\beta}$ conditioned on configurations that correspond to cuts in $\Omega_0$.
Our construction ensures that 	
if $\maxcut(G) < k$, then $\Omega = \Omega_0$ and so $\mu_{{F},\beta}  (\cdot \mid \Omega_0) = \mu_{{F},\beta}$.
Moreover, when
$\maxcut(G) \ge k$, we have $\Omega \neq \Omega_0$ and $\mu_{{F},\beta}  (\cdot \mid \Omega_0) \neq \mu_{{F},\beta}$.
In fact, it can be shown that the TV distance
between $\mu_{{F},\beta}  (\cdot \mid \Omega_0)$ and $\mu_{{F},\beta}$ is $1-o(1)$;
intuitively, this is because $\Omega \setminus \Omega_0$ contains large cuts of $F$ that account for a non-trivial portion of the probability mass of $\mu_{{F},\beta}$.

Our reduction is then completed by generating samples from $\mu_{{F},\beta}  (\cdot \mid \Omega_0)$ and giving these samples and $\mu_{{F},\beta}$ to the identity testing algorithm as input. 
The testing algorithm is guaranteed to succeed with probability at $2/3$.
If the algorithm detects that the samples did not come from $\mu_{{F},\beta}$, it means that $\maxcut(G) \ge k$; otherwise, it means that $\maxcut(G) < k$.
Hence, we have a polynomial running time algorithm that solves the maximum cut problem with probability at least $2/3$, which is not possible unless $\mathsf{RP}=\mathsf{NP}$.

There are two important complications in this approach. First, $F$ is a multi-graph of unbounded degree, and our goal is to establish hardness for the class of antiferromagnetic Ising models graphs of maximum degree $d = O(1)$ when $\beta < \beta_c(d)$. Second, we do not know how to generate samples from $\mu_{{F},\beta}  (\cdot \mid \Omega_0)$  efficiently in polynomial time. 

Let us address first how we solve the issue of $F$ being a multi-graph with large maximum degree. For this, we use a ``degree reducing'' gadget;
the one we use was introduced in~\cite{Sly} to establish the hardness of approximate counting and sampling antiferromagnetic spin systems.
Specifically, each vertex of $F$ is replaced by a gadget $H$ which 
consists of a (nearly) $d$-regular random bipartite graph with a relatively small number of trees attached to it.
Being more precise, the leaves of each tree will be identified with unique vertices on the same side of the bipartite graph; see~Section~\ref{subsec:vertexgadget} for the precise construction.
The root of these trees are called \emph{ports} and are used to connect
the gadgets as dictated by the edges of $F$. 
This results in a simple $d$-regular graph $\widehat{F}$. 

A key feature of the gadget $H$ is that 
in the tree non-uniqueness region $\beta < \beta_c(d)$, a sample from $\mu_{H,\beta}$ will have mostly $+1$'s on one side of $H$ and mostly $-1$'s on the other, or vice versa.
Hence there are two possible ``phases'' for the gadget which we use to simulate the spin of the corresponding vertex in $F$; i.e., the phase of the gadget 
is mapped to the spin of the corresponding vertex of $F$.
Therefore, in a configuration in $\widehat{F}$,
the phase of all the gadgets determine a cut for $F$, and thus one for $G$. 
Consequently, the reduction described above from the maximum cut problem to identity testing using $F$ can be done using $\widehat{F}$ instead.

The second technical complication is that we are required to sample from $\mu_{\widehat{F},\beta}(\cdot\mid\Omega_0)$. 
For this, we observe first that sampling a phase assignment from $\mu_{\widehat{F},\beta}(\Omega_0)$ is straightforward (see Lemma~\ref{lemma:dominant:phase}). 
We then sample the port configuration given the phase vector from $\Omega_0$.
This is done via a rejection sampling procedure by noting that the marginal distribution on the ports is within $o(1)$
total variation distance of a suitably defined product distribution. 
Once the port configuration is sampled within the desired accuracy, we sample the configuration on each gadget (independently) given the configuration of the ports. 
For this we use a hybrid approach: we use the recent algorithm from~\cite{KLR} for low-rank Ising models for one range of values of $\beta$ (i.e., when $|\beta| \sqrt{d} = O(1)$) and polymer models---see~\cite{JKP}---for the other. To use these algorithms, we fleshed out the spectrum of the incidence matrix of the gadget. Note that simulating the $\coorora$  for the Ising model is straightforward as the spin probability is a function of the number of neighboring $+1$ and $-1$ spins.

\begin{remark}
	\label{rmk:diff}
	In~\cite{testing-colt}, hardness of identity testing was established when both the visible and hidden distributions are antiferromagnetic Ising models on graphs of bounded degree also via a reduction from the maximum cut problem. As such, we believe it is meaningful to detail the conceptual and technical differences, as well as some similarities, between the reduction described above and the one from~\cite{testing-colt}.
	Conceptually, in~\cite{testing-colt} the hidden distribution $\pi$ is assumed to be from the same class as $\mu$,
	so the testing problem in consideration is easier. 
	In fact, this problem is not hard for all $\beta < \beta_c(d) < 0$ since when $|\beta|d = O(\log n)$ it can be solved using the learning algorithm from~\cite{klivans2017learning} to learn $\pi$.
	Only when $|\beta|d = \Omega(\log n)$, 
	this variant of identity testing becomes computationally hard, and this is precisely what is established in~\cite{testing-colt}.
	Our goal here is to show hardness throughout the entire non-uniqueness regime $\beta < \beta_c(d)$ (not only for $|\beta|d = \Omega(\log n)$),
	so the hidden distribution in our reduction can not be an Ising model. Our hidden distribution $\mu_{\widehat{F},\beta}(\cdot\mid\Omega_0)$ is instead a \emph{conditional} antiferromagnetic Ising distribution, 
	and, as noted, sampling from it is challenging.
	
	At a technical level, the necessary assumption in~\cite{testing-colt} that $|\beta|d = \Omega(\log n)$ simplifies matters significantly.
	In particular, the degree reducing gadgets there simply consist of random regular bipartite graphs; when $\beta d = \omega(\log n)$, sampling from the antiferromagnetic Ising model on these gadgets is trivial since $1-o(1)$ of the probability mass is concentrated on two trivial configurations ($+1$ in one side of the bipartite graph, $-1$ in the other side and vice versa).
	When $\beta < \beta_c(d)$, the correlations in the model are super-polynomially weaker; i.e., there is no such strong concentration in the ground states. 
	As such, we must use a more sophisticated degree-reducing gadget (the one from~\cite{Sly} as discussed earlier),
	and consider the phase of the gadget to simulate spin assignments to vertices.
	In terms of similarities, the construction of the multi-graph $F$ from the max-cut instance detailed above is nearly identical to the construction in~\cite{testing-colt}, i.e., $F$ is essentially the same, but $\widehat{F}$ is not since we must to use a different gadget.
\end{remark}

\section{Preliminaries}
\label{sec:prelim}

In this section we gather a number of standard definitions and results that we will refer to in our proofs. 
Let $k,n \in \N^+$ be integers. 
Let $\QQ = \{1,\dots, k\}$ denote a finite alphabet set of size $k$, and let $\pi$ be an arbitrary distribution over $\QQ^n$.
Throughout the paper, we use $n$ in the subscript and superscript to represent the set $[n] = \{1,\dots, n\}$ and use $n \setminus i$ to represent the set $[n] \setminus \{i\}$ to ease the notation.

\subsection{Coordinate conditional sampling oracle}
We recall next the formal definitions of the various sampling oracles discussed in the paper.

\begin{definition}
	\label{def:subcora}
	The sampling oracles for the hidden distribution $\pi$ are defined as follows:
	\begin{itemize}
		\item \emph{General Sampling Oracle} ($\stanora$): Generate a sample $x$ from $\pi$.
		\item \emph{Coordinate Conditional Sampling Oracle} ($\coorora$): Given $i \in [n]$ and $x \in \QQ^{n \setminus i}$ as inputs to the oracle:
		\begin{itemize}
			\item If $\pi(X_{n \setminus i} = x) > 0$, the oracle samples $a \in \QQ$ from the conditional marginal distribution $\pi(X_i = \cdot \mid X_{n \setminus i} = x)$;
			\item If $\pi(X_{n \setminus i} = x) = 0$, the oracle outputs $a \in \QQ$ arbitrarily.
		\end{itemize}
		\item \emph{Subcube Conditional Sampling Oracle} ($\subcora$): Given $\Lambda \subseteq [n]$ and $x \in \QQ^{\Lambda}$ as inputs to the oracle:
		\begin{itemize}
			\item If $\pi(X_\Lambda = x) > 0$, the oracle samples $x' \in \QQ^{[n] \setminus \Lambda}$ from the conditional distribution $\pi(X_{V \setminus \Lambda} = \cdot \mid X_\Lambda = x)$;
			\item If $\pi(X_\Lambda = x) = 0$, the oracle outputs $x'\in \QQ^{V \setminus \Lambda}$ arbitrarily.
		\end{itemize}
		\item \emph{Pairwise Conditional Sampling Oracle} ($\pairora$):
		Given $x,y \in \QQ^n$, the oracle returns~$x$ with probability $\pi(x)/(\pi(x)+\pi(y))$ and $y$ otherwise.
	\end{itemize}
\end{definition}

We provide next two brief remarks noting that access to a $\coorora$ is a weaker assumption than access to a $\pairora$ or a $\subcora$.

\begin{remark}\label{rmk:pair-coord}
	$\pairora$ is generally a stronger oracle than $\coorora$.
	When $k=2$ and the state space is the binary hypercube $\{0,1\}^n$, this is obvious since the $\coorora$ essentially generates samples conditioned in the set $\{x,y\}$ where $x$ and $y$ differ in \emph{exactly one coordinate}, while $\pairora$ can handle any pair vectors $x,y \in \QQ^n$.
	If $k\ge 3$ is a constant (independent of $n$), then one can also simulate an $\eps$-approximate $\coorora$ with $\pairora$ access in $\poly(k,\log(1/\eps))$ time (or a perfect one with $\poly(k)$ expected time) using a Markov chain. 
	Given a query $(i,x)$ where $i \in [n]$ and $x \in \QQ^{n \setminus i}$, 
	to generate a random value at the coordinate $i$ conditional on $x$, 
	one can simulate the Markov chain that in each step picks an element $a \in \KK$ uniformly at random and lets $a_{t+1} = a$ with probability $\mu(x_{i,a}) / (\mu(x_{i,a}) + \mu(x_{i,a_t}))$ and $a_{t+1} = a_t$ otherwise; $x_{i,a}$ denotes the vector with the $i$-th coordinating being $a$ and all other coordinates given by $x$. 
	Every step of the Markov chain can be perfectly implemented with the $\pairora$, and a simple coupling argument shows that the $\eps$-mixing time is $\poly(k,\log(1/\eps))$. 
	For perfect sampling with $\poly(k)$ expected time, one can use the Coupling from the Past Method; see~\cite{CFTP}.
\end{remark}

\begin{remark}\label{rmk:subcube}
	The $\subcora$ subsumes the $\coorora + \stanora$ combination implying that:
	\begin{itemize}
		\item Algorithms with both $\coorora + \stanora$ access give algorithms for the setting where there is $\subcora$ access;
		\item Lower bounds for the $\subcora$ model imply lower bounds the $\coorora+\stanora$.
	\end{itemize} 
\end{remark}

\subsection{Identity testing}
We provide next the formal definition of the identity testing problem
for a distribution $\mu$ over $\QQ^n$.
Let $d$ be any metric or divergence for distributions over $\QQ^n$. 

\smallskip
\noindent\textbf{$\idtest(d,\eps; \mu)$}.\\
\noindent\textbf{Input}: Description of a distribution $\mu$ over $\QQ^n$; \\
\noindent\textbf{Provided}: Access to 
$\coorora + \stanora$ for 
an unknown distribution $\pi$ over $\QQ^n$. \\
\noindent\textbf{Goal}:
Determine whether $\pi = \mu$ or $d(\pi, \mu) \ge \eps$.
\smallskip

Let $\FF$ denote a family of distributions (with varying dimensions), each of which is supported on $\QQ^n$ for some integer $n \in \N^+$ and can be represented with $\poly(n)$ parameters.  
We say an algorithm $\AA$ is an identity testing algorithm for the family $\FF$ if for every $\mu \in \FF$ it solves $\idtest(d,\eps; \mu)$ with probability at least $2/3$.  
Note that the unknown distribution $\pi$ does not necessarily belong to the family $\FF$.


\subsection{Coordinate balancedness and marginal boundedness}
\label{def:balanced}
We say a distribution $\mu$ supported on $\QQ^n$ is {\em $\eta$-balanced}, if for every $i \in [n]$, every $x \in \QQ^{n \setminus i}$ with $\mu(X_{n \setminus i} = x) > 0$, and every $a \in \QQ$, one has
\[
\text{either~~} \mu(X_i = a \mid X_{n \setminus i} = x) = 0, 
\quad
\text{or~~} \mu(X_i = a \mid X_{n \setminus i} = x) \ge \eta. 
\]
On the other hand,
we say the distribution $\mu$ is {\em $b$-marginally bounded} if for every $\Lambda \subseteq [n]$, every $x \in \QQ^{\Lambda}$ with $\mu(X_{\Lambda} = x) > 0$, every $i \in [n] \setminus \Lambda$, and every $a \in \QQ$, one has
\[
\text{either~~} \mu(X_i = a \mid X_{\Lambda} = x) = 0,
\quad
\text{or~~} \mu(X_i = a \mid X_{\Lambda} = x) \ge b. 
\]
Note that marginal boundedness is a generalization of coordinate balance, and in particular, if in the definition of $b$-marginally bounded one restricts to $\Lambda$ where $|\Lambda|=n-1$ then we obtain $b$-balanced.  Hence, any $b$-marginally bounded distribution is also $b$-balanced. 
Moreover, if $\mu$ has full support, then both notions are equivalent.
See also Remark~\ref{rmk:marginal-bound} for a weaker version of marginal boundedness. 
Related notions to marginal boundedness 
appeared in~\cite{klivans2017learning,CLV-STOC21}.

\subsection{Approximate tensorization of entropy}
\label{sec:AT-defn}

Let $\mu$ be a distribution over $\QQ^n$. 
For any non-negative function $f: \QQ^n \to \R_{\ge 0}$, the expectation of $f$ is defined to be $\mu(f) = \sum_{x \in \QQ^n} \mu(x) f(x)$, and the (relative) entropy of $f$ is defined as $$
\Ent_\mu(f) = \mu(f \ln f) - \mu(f) \ln(\mu(f)),
$$ 
with the convention that $0 \ln 0 = 0$. 

Given a coordinate $i$ and a partial configuration $x \in \QQ^{n\setminus i}$ on all coordinates but $i$, one can define the entropy of the function $f$ with respect to the conditional distribution $\mu_i(\cdot \mid x)$, which we denote by $\Ent_i^x(f)$. 
Furthermore, we regard $\Ent_i^x(f)$ as a function of $x$ and its expectation, when $x$ is generated from $\mu_{n\setminus i}$, is denoted as $\mu[\Ent_i(f)]$. 
We are now ready to give the formal definition of approximate tensorization of entropy in the functional inequality form, as in~\cite{Cesi,CMT,CP}. 

\begin{definition}[Approximate Tensorization of Entropy: Functional Form]
	\label{def:AT-fun}
	We say that a distribution $\mu$ over $\QQ^n$ satisfies approximate tensorization of entropy with constant $C$ if 
	for any nonnegative function $f: \QQ^n \to \R_{\ge 0}$ one has
	\begin{equation}
		\label{eq:e:func}
	\Ent (f) \le C \sum_{i=1}^n \mu[\Ent_i(f)]. 
	\end{equation}
\end{definition}

As mentioned in the introduction, 
approximate tensorization is an important tool for proving functional inequalities like the modified log-Sobolev inequality (MLSI); it is is also useful for deriving optimal mixing time bounds for the Glauber dynamics.
Although it is most often stated in this functional inequality form, 
mainly because of several useful analytic properties, 
in this paper we will consider its probabilistic version, as in \cite{Marton19,GSS19}.  

For two distributions $\mu$ and $\pi$ over a discrete state space $\QQ^n$, we write $\pi \ll \mu$ if $\mu(x) = 0$ implies $\pi(x) = 0$ for any $x \in \QQ^n$, i.e., the support of $\pi$ is contained in the support of $\mu$. 
The Kullback--Leibler (KL) divergence is defined as
\[
\kl{\pi}{\mu} = \sum_{x \in \QQ^n} \pi(x) \ln \left( \frac{\pi(x)}{\mu(x)} \right). 
\]
The following definition of approximate tensorization is slightly more general than~Definition~\ref{def:AT} from the introduction.   
\begin{definition}[Approximate Tensorization of Entropy: Probabilistic Form]
	\label{def:AT-prob}
	We say a distribution $\mu$ over $\QQ^n$ satisfies approximate tensorization of entropy with constant $C$ if 
	for any distribution $\pi$ over $\QQ^n$ such that $\pi \ll \mu$ one has
	\begin{equation}
	\label{eqn:AT-defn-prob}
	\kl{\pi}{\mu} \le C \sum_{i=1}^n \E_{x \sim \pi_{n \setminus i}} \Big[ \kl{\pi_i(\cdot \mid x)}{\mu_i(\cdot \mid x)} \Big]. 
	\end{equation}
\end{definition}

Note that in Definition~\ref{def:AT} we required that $\mu$ has full support, instead of the more general assumption $\pi \ll \mu$.
We remark that in \eqref{eqn:AT-defn-prob} the partial configuration $x \in \QQ^{n \setminus i}$ is drawn from $\pi$ rather than $\mu$. 
It is easy to check that the two definitions Definitions~\ref{def:AT-fun} and~\ref{def:AT-prob} are equivalent to each other by letting $f = \pi/\mu$; see \cite{Marton19}. 

\begin{remark}
	To the best of our knowledge, there is no known analog of the probabilistic form of approximate tensorization of entropy 
	\eqref{eqn:AT-defn-prob} for other $f$-divergences, even when the visible distribution $\mu$ is a product measure. 
	An analog of~\eqref{eq:e:func} the functional form of the same condition, does exist for the variance functional, which can be translated into a more intricate, weighted version of \eqref{eqn:AT-defn-prob} for $\chi^2$-divergence, with the entropy replaced by variance; the weights depend on ratios of the two densities and can be exponentially large. Although it is possible to establish such an inequality for certain high-dimensional distributions, how to use it for algorithmic purposes remains unclear. This contributes to our rationale for selecting KL in our work. 
\end{remark}	

\section{Identity Testing via Approximate Tensorization}

For integer $k\ge 2$ and real $C \ge 1, \eta>0$, let $\FF_k(C,\eta)$ denote the family of all distributions over $\QQ^n$ (for any $n \in \N^+$) with $\poly(n)$ many parameters that are $\eta$-balanced and satisfy approximate tensorization of entropy with constant $C$. 
The goal of this section is to give an identity testing algorithm for the family $\FF_k(C,\eta)$ in terms of the KL divergence. 
We observe that this also implies a tester for the TV distance by the Pinsker's inequality.

For applications in Section~\ref{subsec:application} all the parameters $k,C,\eta$ are constants independent of $n$.
In the theorem below, however, we consider a more general setting where these parameters are functions of the dimension $n$ with only mild assumptions on their growth rate. 
This allows us to have a clearer picture on the sample complexity and the dependency on all the parameters involved.

\begin{theorem}
	\label{thm:alg-main-detailed}
	Let $k = k(n)$ be an integer and $C = C(n) \ge 1$, $\eta = \eta(n) \in (0,1/2]$ be reals.
	Suppose that $$\max\{\log C, \log\log(1/\eta)\} = O(\log n).$$ 
	Then, there is an identity testing algorithm for the family $\FF_k(C,\eta)$ with query access to both the $\coorora$ and the $\stanora$ and for KL divergence with distance parameter $\eps > 0$. 
	The query complexity of the identity testing algorithm is 
	\[
	O\left( \min\left\{ \frac{C}{\sqrt{\eta}} \cdot \frac{n}{\eps} \log^3 \left(\frac{n}{\eps}\right),\, C^2 \sqrt{k} \log\Big(\frac{1}{\eta}\Big) \cdot \frac{n^2}{\eps^2} \log^2 \left(\frac{n}{\eps}\right) \right\} \right). 
	\]
	The running time of the algorithm is polynomial in all parameters ($1/\eta$ for the first bound, and $\log(1/\eta)$ for the second) and also proportional to the time of computing the conditional marginal distributions $\mu_i(\cdot \mid x)$ for any $i \in [n]$ and any feasible $x \in \QQ^{n \setminus i}$. 
	Furthermore, if $k=2$, i.e., we have a binary domain $\QQ = \{0,1\}$, the query complexity can be improved to
	\[
	O\left( C\log\Big(\frac{1}{\eta}\Big) \cdot \frac{n}{\eps} \log^3 \left(\frac{n}{\eps}\right) \right).
	\]
\end{theorem}

\subsection{Algorithm}

Before presenting our algorithm, we first give 
a well-known fact, e.g., see \cite[Section 8.2.4]{goldreich2017introduction} and \cite[Proposition 6.7]{Nar}.
\begin{lemma}
	\label{lem:reverse-Markov}
	Let $\eps, M > 0$ be reals and let $L = \ceil{\log_2(M/\eps)}$. 
	If $Y$ is a non-negative random variable such that $Y \le M$ always 
	and $\E Y \ge \eps$, then there exists a non-negative integer $\ell \le L$ such that
	\[
	\Pr(Y \ge 2^{\ell-1} \eps) \ge \frac{1}{2^\ell (L+1)}.
	\]
\end{lemma}

\begin{proof}
	Suppose for sake of contradiction that for all $0 \le \ell \le L$ it holds
	\[
	\Pr(Y \ge 2^{\ell-1} \eps) < \frac{1}{2^\ell (L+1)}.
	\]
	Notice that $2^L \eps \ge M$. 
	Then we have
	\begin{align*}
	\E Y 
	&= \int_{0}^M \Pr(Y \ge y) dy
	= \int_0^{\eps/2} \Pr(Y \ge y) dy + \sum_{\ell=0}^{L} \int_{2^{\ell-1} \eps}^{2^\ell \eps} \Pr(Y \ge y) dy\\
	&\le \frac{\eps}{2} + \sum_{\ell=0}^L (2^\ell \eps - 2^{\ell-1} \eps) \Pr(Y \ge 2^{\ell-1} \eps) \\
	&< \frac{\eps}{2} + \sum_{\ell=0}^L 2^{\ell-1} \eps \cdot \frac{1}{2^\ell (L+1)}
	= \eps,
	\end{align*}
	which is a contradiction.
\end{proof}

For $i \in [n]$ and $x \in \QQ^{n \setminus i}$, we define $q^x_i = \mu_i(\cdot \mid x)$
to be a distribution over $\QQ$ induced by the pair $(i,x)$ from $\mu$, where we think of $i$ and $x$ as the parameters. 
Similarly, we define $p^x_i = \pi_i(\cdot \mid x)$ with respect to $\pi$. 

Recall that the approximate tensorization of entropy for $\mu$ can be written as
\begin{equation*}
\kl{\pi}{\mu} \le C \sum_{i=1}^n \E_{x \sim \pi_{n \setminus i}} \Big[ \kl{\pi_i(\cdot \mid x)}{\mu_i(\cdot \mid x)} \Big]
= Cn \; \E_{(i,x)} \left[ \kl{p^x_i}{q^x_i} \right],
\end{equation*}
where $i \in [n]$ is a uniformly random coordinate and $x$ is generated from the marginal distribution $\pi_{n \setminus i}$.
Therefore, the original identity testing problem boils down to the following testing problem:
\[
{\Pr}_{(i,x)}(p^x_i = q^x_i) = 1
\quad \text{v.s.} \quad
\E_{(i,x)} \left[ \kl{p^x_i}{q^x_i} \right] \ge \eps'
\]
where $\eps' = \eps/(Cn)$. 
Notice that $\kl{p^x_i}{q^x_i} \le \ln(1/\eta)$ for all $(i,x)$ assuming $\eta$-balancedness and $p \ll q$. 
By Lemma \ref{lem:reverse-Markov} it further boils down to the following sequence of testing problems:
let $L = \ceil{\log_2(\ln(1/\eta)/\eps')}$ and for each $\ell \le L$, for a random pair $(i,x)$, distinguish between
$p^x_i = q^x_i$ surely
versus
$${\Pr}_{(i,x)}\left( \kl{p^x_i}{q^x_i} \ge 2^{\ell-1} \eps' \right) \ge \frac{1}{2^\ell (L+1)}.$$

For this testing problem, we sample $(i,x)$ for $O\left( 2^\ell (L+1) \right)$ times so that we get to see the event $\kl{p^x_i}{q^x_i} \ge 2^{\ell-1} \eps'$,
and when it happens the problem is reduced to a classical identity testing setting on a finite state space where we can apply previously known identity testing algorithm. 
To accomplish this we also give an identity testing algorithm for the KL divergence, which is missing in the literature; see Lemmas \ref{lem:KL-id-test} and \ref{lem:Ber-KL-id-test}. 

We give a few more definitions before presenting our algorithm formally. 
For a distribution $\pi$ over $\XX = \QQ^n$, we define the set $\XX'$ by 
\[
\XX' = \left\{ (i,x): i \in [n],\, x \in \QQ^{n \setminus i} \right\}
\]
to be the set of all pairs $(i,x)$ where $i$ is one coordinate and $x$ contains the values of all coordinates other than $i$. 
We then define a distribution $\pi'$ over $\XX'$ by
\[
\pi'(i,x) = \frac{1}{n} \, \pi_{n \setminus i}(x) = \frac{1}{n} \, \pi \left( X_{n \setminus i} = x \right),
\]
so that a sample from $\pi'$ can be obtained in the following way: first pick $i \in [n]$ uniformly at random, and then sample $x$ from the marginal distribution $\pi_{n \setminus i}$.

\begin{algorithm}[t]
	\caption{Identity Testing for $\FF_k(C,\eta)$ for KL divergence}
	\label{alg:id-test}
	\KwIn{Description (parametrization) of a given distribution $\mu \in \FF_k(C,\eta)$, query access to both $\coorora$ and $\stanora$ for an unknown distribution $\pi$, and distance parameter $\eps > 0$.} 
	
	\smallskip
	$\eps' \gets \eps / (Cn)$\;
	$L \gets \ceil{\log_2(\ln(1/\eta)/\eps')}$\; 
	\For{$0 \le \ell \le L$}{
		$\eps_\ell \gets 2^{\ell-1} \eps'$ \Comment*[r]{Distance parameter}
		$\delta \gets 2^{-2L-6}$ \Comment*[r]{Failure probability}
		$T_\ell \gets 2^{\ell+2} (L+1)$ \Comment*[r]{Need $T_\ell$ samples of $(i,x)$ to see $\kl{p^x_i}{q^x_i} \ge \eps_\ell$}
		\For{$t = 1,2, \dots, T_\ell$}{
			Sample $(i,x)$ from $\pi'$ via $\stanora$ for $\pi$\; \label{line:sample-i-x}
			Call $\AA_{\textsc{kl-id}}$ from Lemmas \ref{lem:KL-id-test} and \ref{lem:Ber-KL-id-test} to distinguish between $p^x_i$ and $q^x_i$ with distance parameter $\eps_\ell$ and failure probability $\delta$
			(samples from $p^x_i$ are obtained via $\coorora$ for $\pi$) \Comment*[r]{Check whether $\kl{p^x_i}{q^x_i} \ge \eps_\ell$}
			\label{line:call-tester}
			
			\If{$\AA_{\textsc{kl-id}}$ returns No (i.e., $\kl{p^x_i}{q^x_i} \ge \eps_\ell$)}
			{\KwOut{No (i.e., $\kl{\pi}{\mu} \ge \eps$), and the algorithm ends\;}}
		}
	}


	\KwOut{Yes (i.e., $\pi = \mu$)}
\end{algorithm}

Our algorithm is given in Algorithm~\ref{alg:id-test}, which also appeared in the previous work \cite{CCKLW} for uniformity testing over the binary hypercube $\{0,1\}^n$. 
We now give our proof of Theorem \ref{thm:alg-main-detailed}.

\begin{proof}[Proof of~Theorem \ref{thm:alg-main-detailed}]
	Suppose first that $\pi = \mu$. 
	Then each time we call the KL tester in Line~\ref{line:call-tester}, it returns Yes with probability at least $1-\delta$ since $p^x_i = q^x_i$ for any $(i,x)$. 
	If every time the result is Yes then Algorithm~\ref{alg:id-test} will return Yes (i.e., $\pi = \mu$). 
	By a simple union bound, the probability that Algorithm~\ref{alg:id-test} mistakenly outputs No is at most 
	$$ \sum_{\ell=0}^L T_\ell \cdot \delta = \sum_{\ell=0}^L 2^{\ell+2} (L+1) \cdot 2^{-2L-6} \le 2^{L+3} (L+1) \cdot 2^{-2L-6} \le \frac{1}{8},$$
	where the last inequality is due to $L+1 \le 2^L$.  
	
	Next assume that $\kl{\pi}{\mu} \ge \eps$. Then by approximate tensorization of entropy we have 
	$$\E_{(i,x)} \left[ \kl{p^x_i}{q^x_i} \right] \ge \eps'$$ where $\eps' = \eps/(Cn)$. 
	By Lemma \ref{lem:reverse-Markov}, there exists a non-negative integer $\ell \le L$ such that
	$$\Pr\left( \kl{p^x_i}{q^x_i} \ge 2^{\ell-1} \eps' \right) \ge \frac{1}{2^\ell (L+1)}.$$
	For this $\ell$, the algorithm repeats for $T_\ell = 2^{\ell+2} (L+1)$ times to find such a pair $(i,x)$ via the general sampling oracle; the probability that the algorithm fails to find such $(i,x)$ is upper bounded by
	\[
	\left( 1-\frac{1}{2^\ell (L+1)} \right)^{T_\ell} \le 
	\exp\left( - \frac{T_\ell}{2^\ell (L+1)} \right)
	= e^{-4} \le \frac{1}{50}.  
	\]
	In the case that such a pair is successfully found, the KL tester in Line~\ref{line:call-tester} will return No with probability at least $1-\delta$, and hence then Algorithm~\ref{alg:id-test} returns No.
	Therefore, if Algorithm~\ref{alg:id-test} wrongly outputs Yes then either a good pair $(i,x)$ is not found, or the KL tester in Line~\ref{line:call-tester} makes a mistake on a good pair $(i,x)$. 
	The probability of outputting Yes is then upper bounded by $1/50 + \delta \le 1/8$. 
	
	Finally, Lemma \ref{lem:KL-id-test}, combined with the amplification technique for failure probability (e.g., see \cite[Lemma 1.1.1]{Canonne22}), implies that the number of samples required by Algorithm~\ref{alg:id-test} is at most
	\begin{align*}
	& \sum_{\ell = 0}^L T_\ell \cdot O\left( \min\left\{ \frac{\ln (1/\delta)}{\eps_\ell \sqrt{\eta}},\, \frac{\sqrt{k} \ln(1/\eta) \ln (1/\delta)}{\eps_\ell^2} \right\} \right) \\
	={}& \sum_{\ell = 0}^L O\left( \min\left\{ \frac{C L^2 n}{\eps \sqrt{\eta}},\, \frac{C^2 \sqrt{k} L^2 n^2 \ln(1/\eta)}{2^\ell \eps^2} \right\} \right) \\
	={}& O\left( \min\left\{ \frac{C L^3  n}{\eps \sqrt{\eta}},\, \frac{C^2 \sqrt{k} L^2 n^2 \ln(1/\eta)}{\eps^2} \right\} \right).
	\end{align*}
	Since $L = O(\log (n/\eps))$ under our assumptions $\log(C) = O(\log n)$ and $\log\log(1/\eta) = O(\log n)$, we obtain the sample complexity upper bound from the theorem. 
	For $k = 2$ the sample complexity is obtained in the same way, using Lemma \ref{lem:Ber-KL-id-test} instead.  
\end{proof}

\subsection{Identity testing for KL divergence on general domain}

In this and next subsection, we prove Lemma \ref{lem:KL-id-intro} from Section~\ref{subsec:ov-alg}, which is also a key subroutine of Algorithm~\ref{alg:id-test}. 
We first consider general $k\ge 2$ in this subsection, and then give an improved sample complexity bound for $k=2$ in the next subsection. 

\begin{lemma}
	\label{lem:KL-id-test}
	Let $k \in \N^+$ be an integer, 
	and let $\eps > 0$, $\eta \in (0,1/2]$ be reals. 
	Given a target distribution $q$ over domain $\QQ$ of size $k$ such that either $q(a) = 0$ or $q(a) \ge \eta$ for any $a \in \KK$, and given sample access to an unknown distribution $p \ll q$ over $\QQ$, there exists a polynomial-time identity testing algorithm that distinguishes with probability at least $2/3$ between the two cases 
	\begin{equation}
	p = q
	\quad\text{and}\quad
	\kl{p}{q} \ge \eps
	\end{equation}
	with sample complexity
	$
	O\Big( \min\big\{ \frac{1}{\eps \sqrt{\eta}},\, \frac{\sqrt{k} \ln(1/\eta)}{\eps^2} \big\} \Big). 
	$
\end{lemma}

The running time is polynomial in $1/\eta$ if we apply the first bound, and $\log(1/\eta)$ for the second.

\begin{remark}
	\label{rmk:eta:dep}
	We remark that the dependency on $\eta$ in the sample complexity is inevitable.
	This is because, if $\eta \rightarrow 0$, $\tv{p}{q}$ could tend to $0$ but $\kl{p}{q}$ can be independent of $\eta$.
	For example, if $p$ and $q$ are Bernoulli random variables with means $\hat p,\hat q \in (0,1/2)$, respectively,
	one can have $\kl{p}{q} = 0.1$, for some $\hat p=\hat p(\hat q)$, with $\tv{p}{q} = |\hat p-\hat q| \to 0$ as $\hat q \to 0$.
	In Lemma \ref{lem:KL-id-intro} we show that the sample complexity for identity testing with respect to KL divergence depends, in the worst case, logarithmically on $1/\eta$. 
	In particular, for uniformity testing under KL divergence, Lemma \ref{lem:KL-id-intro} gives an $O(\sqrt{k}/\eps)$ sample complexity which matches what one would expect given the $O(\sqrt{k}/\eps^2)$ sample complexity for uniformity testing under TV distance and Pinsker's inequality.
\end{remark}

%

We need the following standard inequalities between statistical divergences;
see \cite{Canonne22} for more. 

\begin{lemma}
	\label{lem:kl-chi2-l2}
	Let $q$ be a distribution fully supported on a finite set $\QQ$, and let $\eta = \min_{a \in \QQ} q(a)$. Then for any distribution $p$ over $\QQ$ with $p \ll q$ it holds
	\[
	\kl{p}{q} 
	\le \chitwo{p}{q}
	\le \frac{1}{\eta}\, \norm{p-q}_2^2
	\le \frac{2}{\eta}\, \tv{p}{q}^2.
	\]
\end{lemma}

%


Our KL tester uses the following identity testing algorithm from \cite{DK16} for $\ell_2$ distance, and also the flattening technique proposed there to reduce the $\ell_2$ norm. 
See also \cite[Theorem 2.2.2]{Canonne22} for an exposition of the algorithms and techniques.

\begin{lemma}[\cite{DK16}]
	\label{lem:best-l2-id-test}
	Given the distance parameter $\eps > 0$, full description of the target distribution $q$ with domain $\QQ$ of size $k$, and general sample access to an unknown distribution $p$ over $\QQ$, there exists a polynomial-time identity testing algorithm $\AA_{\ell_2\textsc{-id}}$ that distinguishes with probability at least $2/3$ between the two cases 
	\begin{equation}
	\norm{p-q}_2 \le \frac{\eps}{2}
	\quad\text{and}\quad
	\norm{p-q}_2 \ge \eps
	\end{equation}
	with sample complexity
	$
	O\Big( \max\big\{  \frac{\norm{q}_2}{\eps^2},\, \frac{1}{\eps} \big\} \Big). 
	$
\end{lemma}

We are now ready to give our proof of Lemma \ref{lem:KL-id-test}. 

\begin{proof}[Proof of~Lemma \ref{lem:KL-id-test}]
	Without loss of generality we may assume that $q$ is fully supported over $\QQ$; i.e., $q(a) \ge \eta$ for all $a \in \KK$. 
	We establish the two bounds in the lemma separately using two testing algorithms depending on the range of parameters, both based on the $\ell_2$ tester in \cite{DK16}.
	To clarify, our testing algorithm will check the two bounds $1/(\eps\sqrt{\eta})$ and $(\sqrt{k}/\eps^2) \ln(1/\eta)$, find the smaller one, and run the algorithm for that bound.

	\medskip
	\noindent\textbf{Algorithm A.}
	We construct a new instance $p',q'$ (including the oracle for $p'$) of the identity testing problem from $p, q$ such that $\eta/2 \le q'(a) \le \eta$ for all $a \in \QQ'$, where $\QQ'$ is a new domain of size $k' = \Theta(1/\eta)$; this is achieved using the flattening technique from \cite{DK16}.  
	We show that the new identity testing problem with $p',q'$ is equivalent to the original one with $p,q$, requiring the same number of samples. And we can apply Lemma \ref{lem:best-l2-id-test} to the new instance $p',q'$ with nicer properties to obtain a better bound on the sample complexity instead of doing it directly to $p,q$. 
	For each $a \in \QQ$, define 
	\[
	k_a = \floor{\frac{q(a)}{\eta}} + 1.
	\]
	We split each $a \in \KK$ into $k_a$ distinct copies, denoted by $a_1,\dots,a_{k_a}$, which constitute the new domain $\QQ'$; namely, 
	\[
	\QQ' = \{a_i: a \in \KK,\, 1\le i \le k_a\}.
	\] 
	Notice that the size $k' = |\QQ'|$ of the new domain is bounded by 
	\[
	k' = \sum_{a \in \KK} k_a 
	\le \sum_{a \in \KK} \left(\frac{q(a)}{\eta} + 1\right)
	= \frac{1}{\eta} + k
	\le \frac{2}{\eta},
	\]
	where the last inequality follows from $\eta \le 1/k$ since $q$ is fully supported on $\QQ$. 
	The new target distribution $q'$ is given by for every $a \in \KK$ and $i \in [k_a]$,
	\[
	q'(a_i) = \frac{q(a)}{k_a}, 
	\]
	and similarly $p'$ is given by $p'(a_i) = p(a)/k_a$. 
	We can easily transform a sample from $p$ into a sample from $p'$: if we receive $a$ as a sample from $p$, then we can compute $k_a$ (since the target distribution $q$ is given in full description) and generate $i \in [k_a]$ uniformly at random, so that $a_i$ is a sample from the distribution $p'$. 
	The crucial fact here is that the KL divergence (more generally, any $f$-divergence) is preserved under flattening. 
	Indeed, observe that
	\[
	\kl{p'}{q'} 
	= \sum_{a \in \QQ} \sum_{i \in [k_a]} p'(a_i) \ln\frac{p'(a_i)}{q'(a_i)}
	= \sum_{a \in \QQ} \sum_{i \in [k_a]} \frac{p(a)}{k_a} \ln\frac{p(a)}{q(a)}
	= \kl{p}{q}.
	\]
	Thus, we only need to solve the identity testing problem for the flattened distribution $p'$ and $q'$. 
	Moreover, we observe that for all $a \in \KK$ and $i \in [k_a]$ it holds
	\[
	\frac{\eta}{2} \le q'(a_i) = \frac{q(a)}{k_a} \le \eta,
	\]
	since we have
	\[
	\frac{2q(a)}{\eta} \ge k_a = \floor{\frac{q(a)}{\eta}} + 1 \ge \frac{q(a)}{\eta},
	\]
	where the first inequality follows from $q(a) \ge \eta$. 
	Therefore, we observe that
	\[
	\norm{q'}_2^2 = \sum_{a_i \in \QQ'} q'(a_i)^2 
	\le \eta \sum_{a_i \in \QQ'} q'(a_i) = \eta.
	\]
	
	Note that $\kl{p'}{q'} \le (2/\eta) \norm{p'-q'}_2^2$ by Lemma \ref{lem:kl-chi2-l2}.  
	Applying Lemma \ref{lem:best-l2-id-test}, we are able to distinguish between $p'=q'$ versus $\norm{p'-q'}_2^2 \ge \eps \eta / 2$, and hence between $p'=q'$ versus $\kl{p'}{q'} \ge \eps$, using
	\[
	O\left( \max\left\{  \frac{\norm{q'}_2}{\eps \eta},\, \frac{1}{\sqrt{\eps \eta}} \right\} \right)
	= O\left( \frac{1}{\eps\sqrt{\eta}} \right)
	\]
	samples from the unknown distribution $p'$. 
	This then gives an identity testing algorithm for $p$ and $q$ for KL divergence using the same number of samples from $p$.

	\medskip
	\noindent\textbf{Algorithm B.}
	The previous algorithm works well when $1/\eta$ is not too large. 
	To get a better dependency on $1/\eta$ as in the second bound, more work is required.
	Our first step is still flattening the distributions, but up to the scale $1/k$ instead of $\eta$.
	This is done exactly in \cite{DK16} and \cite[Theorem 2.2.2]{Canonne22}. 
	Let $k_a = \floor{kq(a)}+1$ for each $a \in \QQ$ and let $q'(a_i) = q(a)/k_a$ for each $a \in \QQ$ and $i \in [k_a]$.
	The flattened distributions $p'$ and $q'$ satisfies the following properties:
	\begin{enumerate}[(a)]
		\item Given an explicit description of $q$, one can efficiently give an explicit description of the flattened distribution $q$;
		\item Given access to the sampling oracle for $p$, one can efficiently generate samples from the flattened distribution $p'$; 
		\item The KL divergence is preserved, i.e., $\kl{p'}{q'} = \kl{p}{q}$;
		\item The size of the new domain is $k' \le 2k$; 
		\item For every $a_i \in \QQ'$, we have $\eta/2 \le q'(a_i) \le 2/k'$; \label{itm:d}
		\item We have $\norm{q'}_2 \le \sqrt{2/k'}$. \label{itm:e}
	\end{enumerate}
	The proofs of these properties are the same as before or as in \cite[Theorem 2.2.2]{Canonne22} so we omit here. 
	We only mention the lower bound on $q'(a_i)$: 
	since $k_a = \floor{kq(a)}+1 \le kq(a) + q(a)/\eta$ we have that
	\[
	q'(a_i) = \frac{q(a)}{k_a} \ge \frac{q(a)}{kq(a) + \frac{q(a)}{\eta}} = \frac{\eta}{k\eta + 1} \ge \frac{\eta}{2}.  
	\]
	Therefore, it suffices to consider the identity testing problem with respect to distributions $p'$ and $q'$ satisfying properties (\ref{itm:d}) and (\ref{itm:e}). 
	For ease of notation, in the rest of the proof we assume that our $p,q$ are already flattened to satisfy (\ref{itm:d}) and (\ref{itm:e}), instead of writing $p',q'$ and $k'$. 
	
	Our second step is to divide elements in $\QQ$ into two classes, those with larger probability mass and those with smaller one, and to upper bound the KL divergence by dealing with the two classes separately. 
	Let 
	\[
	\zeta = \frac{\eps}{10 k \ln(2/\eta)}, 
	\]
	and let $\QQ_1 = \{a \in \KK: q(a) \ge \zeta \}$ and $\QQ_2 = \{a \in \KK: \eta/2 \le q(a) < \zeta \}$.
	Hence, $\{\QQ_1,\QQ_2\}$ forms a partition of $\QQ$.
	We upper bound the KL divergence of $p$ and $q$ as follows.
	Observe that
	\begin{equation}\label{eq:KL-1-2}
	\kl{p}{q} = \sum_{a \in \QQ_1} p(a) \ln \frac{p(a)}{q(a)} + \sum_{a \in \QQ_2} p(a) \ln \frac{p(a)}{q(a)}.
	\end{equation}
	For the second term, we have
	\begin{equation}\label{eq:KL-first}
	\sum_{a \in \QQ_2} p(a) \ln \left(\frac{p(a)}{q(a)}\right) \le 
	\sum_{a \in \QQ_2} p(a) \ln(2/\eta) = \left( \ln(2/\eta) \right) p(\QQ_2).
	\end{equation}
	For the first term, we have
	\begin{align}
	\sum_{a \in \QQ_1} p(a) \ln \left(\frac{p(a)}{q(a)}\right)
	&\le \sum_{a \in \QQ_1} p(a) \left( \frac{p(a)}{q(a)} - 1 \right) \nonumber\\
	&= \sum_{a \in \QQ_1} (p(a) - q(a)) \left( \frac{p(a)}{q(a)} - 1 \right)
	+ \sum_{a \in \QQ_1} q(a) \left( \frac{p(a)}{q(a)} - 1 \right) \nonumber\\
	&= \sum_{a \in \QQ_1} \frac{(p(a) - q(a))^2}{q(a)} + p(\QQ_1) - q(\QQ_1) \nonumber\\
	&\le \frac{1}{\zeta} \norm{p-q}_2^2 + q(\QQ_2) - p(\QQ_2), \label{eq:KL-second}
	\end{align}
	where the last inequality is because $q(a) \ge \zeta$ for $a \in \QQ_1$. 
	Therefore, combining \eqref{eq:KL-1-2}, \eqref{eq:KL-first} and \eqref{eq:KL-second} we obtain
	\begin{equation}\label{eq:KL-upper}
	\kl{p}{q} \le \frac{1}{\zeta} \norm{p-q}_2^2 + \left( \ln(2/\eta) - 1 \right) p(\QQ_2) + q(\QQ_2).
	\end{equation}
	In particular, \eqref{eq:KL-upper} directly implies the following fact. 
	\begin{fact}\label{fact:kl-2cases}
		If $\kl{p}{q} \ge \eps$, then 
		\[
		\text{either~~} p(\QQ_2) \ge \frac{\eps}{5\ln(2/\eta)}, \text{~~~~or~~} \norm{p-q}_2^2 \ge \frac{4}{5}\eps\zeta.
		\]
	\end{fact}
	\noindent
	To see this, suppose on contrary that $p(\QQ_2) < \eps/(5\ln(2/\eta))$ and $\norm{p-q}_2^2 < \frac{4}{5}\eps\zeta$. 
	Since we know 
	\[
	q(\QQ_2) \le \zeta k = \frac{\eps}{10\ln(2/\eta)} < \frac{\eps}{5\ln(2/\eta)}, 
	\]
	we deduce from \eqref{eq:KL-upper} that
	\[
	\eps \le \kl{p}{q} < \frac{4}{5} \eps + \frac{1}{5}\eps = \eps,
	\]
	which is a contradiction.

	Our identity testing algorithm proceeds by conducting two tests independently. 
	In the first test, we try to distinguish between $p(\QQ_2) = q(\QQ_2)$ and $p(\QQ_2) \ge \eps/(5\ln(2/\eta))$ with failure probability $1/6$ and sample complexity $m_1 = O(\ln(1/\eta) / \eps)$. (If $\QQ_2 = \emptyset$ then we do nothing in this first stage.)
	To be more precise, let $X$ (respectively $Y$) be the indicator of the event that a sample drawn from $p$ (respectively $q$) is contained in $\QQ_2$. 
	So both $X$ and $Y$ are Bernoulli random variables, where the expectation $q(\QQ_2) > 0$ of $Y$ is known to the algorithm, while the expectation $p(\QQ_2)$ of $X$ is unknown but we have sample access to $X$ via samples from $p$. 
	We would like to distinguish between the two cases $p(\QQ_2) = q(\QQ_2)$, i.e.\ $X$ and $Y$ are the same, and $p(\QQ_2) \ge \eps/(5\ln(2/\eta))$, i.e.\ $X$ and $Y$ are far from each other since $q(\QQ_2) \le \eps/(10\ln(2/\eta))$. 
	This is a standard property testing problem for Bernoulli random variables.
	We use the testing algorithm from Lemma \ref{lem:Ber-id-test} for 
	$$ \gamma = \frac{\eps}{5q(\QQ_2)\ln(2/\eta)} - 1 \ge 1 $$ 
	with failure probability $1/6$, using $m_1 = O(\ln(1/\eta) / \eps)$ samples from $X$.

	In the second stage, we run the tester from Lemma \ref{lem:best-l2-id-test} to distinguish between $p=q$ and $\norm{p-q}_2^2 \ge \frac{4}{5}\eps\zeta$ with failure probability $1/6$. 
	Let $m_2$ be the number of samples that the $\ell_2$ tester uses, 
	and we obtain from Lemma \ref{lem:best-l2-id-test} that
	\[
	m_2 = O\left( \max\left\{ \frac{\norm{q}_2}{\eps \zeta},\, \frac{1}{\sqrt{\eps \zeta}} \right\} \right)
	= O\left( \frac{\sqrt{k} \ln(1/\eta)}{\eps^2} \right), 
	\]
	where we use the property (\ref{itm:e}) from flattening.  
	
	Suppose in both tests the outputs are Yes (i.e., $p(\QQ_2) = q(\QQ_2)$ in the first and $p=q$ in the second), then our identity testing algorithm will output Yes (i.e., $p=q$). 
	If in at least one test the output is No, then our identity testing algorithm will output No (i.e., $\kl{p}{q} \ge \eps$). 
	To finish up the proof, we still need to bound the failure probability and the number of samples needed for our testing algorithm. 
	Suppose first that $p=q$, and hence $p(\QQ_2) = q(\QQ_2)$. 
	Our testing algorithm wrongly outputs No if at least one of the two tests makes a mistake and outputs No. By a simple union bound, the probability of this is at most $1/6+1/6 = 1/3$.
	On the other hand, if $\kl{p}{q} \ge \eps$, then either $p(\QQ_2) \ge \eps/(5\ln(2/\eta))$ or $\norm{p-q}_2^2 \ge \frac{4}{5}\eps\zeta$ by Fact~\ref{fact:kl-2cases}, and so at least one of the two tests should output No if it does not make a mistake. Hence, the failure probability is at most $1/6$. 
	Finally, the number of samples we need is
	\[
	m_1 + m_2 = O\left( \frac{\ln(1/\eta)}{\eps} \right) + O\left( \frac{\sqrt{k} \ln(1/\eta)}{\eps^2} \right)
	= O\left( \frac{\sqrt{k} \ln(1/\eta)}{\eps^2} \right). 
	\]
	This establishes the second bound of the lemma. 
\end{proof}

\subsection{Identity testing for KL divergence on binary domain}

If $k=2$, i.e., we have a binary domain $\QQ = \{0,1\}$, then the sample complexity for the KL tester is better.

For $p \in [0,1]$, the Bernoulli distribution denoted by $\Ber(p)$ is the distribution over $\{0,1\}$ such that $\Pr(X = 1) = p$.
We record below the standard Chernoff bounds. 

\begin{lemma}[Chernoff bounds]
	\label{lem:Chernoff}
	Suppose $X_1,\dots,X_m$ are independent Bernoulli random variables from $\Ber(p)$ where $p \in [0,1]$. 
	Let $\hat{p} = \frac{1}{m} \sum_{i=1}^m X_i$ denote the sample mean. 
	Then for all $\delta \ge 0$,
	\begin{align*}
	\Pr\left( \hat{p} \le (1-\delta) p \right) &\le e^{-\delta^2 p m /2};\\
	\Pr\left( \hat{p} \ge (1+\delta) p \right) 
	&
	\le e^{-\delta^2 p m /(2+\delta)}
	\le
	\begin{cases}
	e^{-\delta^2 p m /3}, & 0\le \delta \le 1;\\
	e^{-\delta p m /3}, & \delta \ge 1.
	\end{cases}
	\end{align*}
\end{lemma}

The following is a folklore fact. 
\begin{lemma}
	\label{lem:Ber-id-test}
	Let $\gamma > 0$ be a real number.
	Given $q \in (0,1/(1+\gamma)]$ and sample access to $\Ber(p)$ with unknown $p \in [0,1]$, there exists a polynomial-time identity testing algorithm that distinguishes with probability at least $2/3$ between the two cases 
	\begin{equation}
	p = q
	\quad\text{and}\quad
	p \ge (1+\gamma) q
	\end{equation}
	with sample complexity
	\[
	O\left( \frac{1+\gamma}{\gamma^2 q} \right)
	= \begin{cases}
	O\left( \dfrac{1}{\gamma^2 q} \right), & 0 < \gamma \le 1;\\[1em]
	O\left( \dfrac{1}{(1+\gamma) q} \right), & \gamma \ge 1.
	\end{cases}
	\] 
\end{lemma}

\begin{proof}
	Let $\hat{p}$ denote the sample mean of $m$ independent samples from $\Ber(p)$ where
	\[
	m = \ceil{\frac{10(1+\gamma)}{\gamma^2 q}}.
	\] 
	If $\hat{p} \le (1+\gamma/2) q$, then the tester concludes $p = q$; otherwise, it concludes $p \ge (1+\gamma) q$. 
	
	Suppose first $p=q$. Then by the Chernoff bound Lemma \ref{lem:Chernoff} we have
	\[
	\Pr\left( \hat{p} \ge \left( 1+\frac{\gamma}{2} \right) q \right) 
	\le \exp\left( -\frac{\gamma^2 qm}{2(4+\gamma)} \right) \le \frac{1}{3}. 
	\]
	If $p \ge (1+\gamma) q$, then again by the Chernoff bound Lemma \ref{lem:Chernoff} we have
	\begin{multline*}
	\Pr\left( \hat{p} \le \left( 1+\frac{\gamma}{2} \right) q \right)
	\le \Pr\left( \hat{p} \le \left(\frac{1+\gamma/2}{1+\gamma}\right) p \right)
	= \Pr\left( \hat{p} \le \left(1 - \frac{\gamma}{2(1+\gamma)}\right) p \right) \\
	\le \exp\left( - \frac{\gamma^2 p m}{8(1+\gamma)^2} \right) 
	\le \exp\left( - \frac{\gamma^2 q m}{8(1+\gamma)} \right) 
	\le \frac{1}{3}.
	\end{multline*}
	Finally, for $0 < \gamma \le 1$ one has $(1+\gamma)/\gamma^2 \le 2/\gamma^2$, and for $\gamma \ge 1$ one has $(1+\gamma)/\gamma^2 \le 4/(1+\gamma)$, which completes the proof of the lemma. 
\end{proof}

We now give our testing algorithm for Bernoulli random variables. 
\begin{lemma}
	\label{lem:Ber-KL-id-test}
	Let $\eps > 0$ be a real number.
	Given $q \in (0,1)$ and sample access to $\Ber(p)$ with unknown $p \in [0,1]$, there exists a polynomial-time identity testing algorithm that distinguishes with probability at least $2/3$ between the two cases 
	\begin{equation}
	p = q
	\quad\text{and}\quad
	\kl{\Ber(p)}{\Ber(q)} \ge \eps
	\end{equation}
	with sample complexity
	\[
	O\left( \frac{\ln(1/\eta)}{\eps} \right)
	\]
	where $\eta = \min\{q, 1-q\}$.
\end{lemma}

\begin{proof}
	We may assume without loss of generality that $q \le 1/2$ and $\eta = q$, as otherwise we can flip the Bernoulli. 
	For $q \in (0,1)$ and $p \in [0,1]$, we define
	\[
	\bkl(p,q) 
	= \kl{\Ber(p)}{\Ber(q)}
	= p \ln\frac{p}{q} + (1-p) \ln \frac{1-p}{1-q}.
	\]
	The testing algorithm is as follows.
	Let $S<T$ be parameters which, as will be clear soon, depend on the distance parameter $\eps$ and the mean $q$ of the given Bernoulli (note that both $\eps$ and $q$ are known to the algorithm). 
	Compute the sample mean $\hat{p}$ for $p$ using $m$ independent samples from $\Ber(p)$.
	The testing algorithm determines $p=q$ or $\bkl(p,q) \ge \eps$ by checking whether $\hat{p}$ belongs to the interval $[S,T]$ or not. 
	More specifically, if $\hat{p} \in [S,T]$, then it outputs $p = q$.
	If $\hat{p} \notin [S,T]$, then it outputs $\bkl(p,q) \ge \eps$. 
	We need to choose suitable $S$ and $T$ so that the algorithm is accurate with high probability and the number of samples required is minimized.
	Given $q \in (0,1/2]$ and $\eps>0$, we will consider three separate cases. 
	

	\medskip
	\noindent\textbf{Case 1: $\eps \le 2q$.}
	%
	We choose $S = q - \sqrt{\eps q/8}$, $T = q + \sqrt{\eps q/8}$, 
	and let 
	\[
	m = \ceil{\frac{64}{\eps}}.
	\] 
	be the number of samples.
	If $p=q$, then by the Chernoff bound (Lemma \ref{lem:Chernoff}) we have
	\[
	\Pr\left( \hat{p} \notin [S,T] \right)
	\le
	\Pr\left( |\hat{p} - q| \ge \sqrt{\frac{\eps q}{8}} \right) \le 2 \exp\left( - \frac{\eps m}{24} \right) \le \frac{1}{3},
	\]
	where we use $\sqrt{\eps q / 8} \le q/2$ by the assumption $\eps \le 2q$. 
	
	Now suppose $\bkl(p,q) \ge \eps$.
	By Lemma \ref{lem:kl-chi2-l2}, we have
	\[
	\eps \le \bkl(p,q) 
	\le \frac{2}{q}(p-q)^2,
	\]
	and hence either $p \le q - \sqrt{\eps q/2}$ or $p \ge q + \sqrt{\eps q/2}$.
	Suppose $p \le q - \sqrt{\eps q/2}$. 
	If $p=0$ then trivially $\Pr\left( \hat{p} \in [S,T] \right) = 0$ since $S = q - \sqrt{\eps q/8} \ge q/2 > 0$. 
	If $0 < p \le q - \sqrt{\eps q/2}$, then again by the Chernoff bound (Lemma \ref{lem:Chernoff}) we have
	\[
	\Pr\left( \hat{p} \in [S,T] \right)
	\le \Pr\left( \hat{p} \ge S \right)
	\le
	\exp\left( - \frac{(S-p)^2 m}{2p+(S-p)} \right)
	\le \exp\left( - \frac{\eps m}{16} \right) \le \frac{1}{3},
	\]
	where the second to last inequality follows from $S+p \le 2q$ and $(S-p)^2 \ge \eps q/8$. 
	If $p \ge q + \sqrt{\eps q/2}$, then Lemma \ref{lem:Chernoff} gives
	\[
	\Pr\left( \hat{p} \in [S,T] \right)
	\le \Pr\left( \hat{p} \le T \right)
	\le
	\exp\left( - \frac{(p-T)^2 m}{2p} \right)
	\le \exp\left( - \frac{\eps m}{32} \right) \le \frac{1}{3},
	\]
	where the second to last inequality follows from the fact that $(p-T)^2 / (2p)$ is minimized at $p = q + \sqrt{\eps q/2} \le 2q$ and hence
	\[
	\frac{(p-T)^2}{2p} \ge \frac{(q + \sqrt{\eps q/2}-T)^2}{2(q + \sqrt{\eps q/2})}
	\ge \frac{\eps}{32}.
	\]

	\medskip
	\noindent\textbf{Case 2: $2q < \eps \le 2q \ln(1/q)$.}
	(This case is possible only for $q < 1/e$.) 
	Again if $\bkl(p,q) \ge \eps$ then either $p \le q - \sqrt{\eps q/2}$ or $p \ge q + \sqrt{\eps q/2}$.
	But since $\eps > 2q$, we have $q - \sqrt{\eps q/2} < 0$ 
	and hence it must be $p \ge q + \sqrt{\eps q/2}$.
	This means that, we need to distinguish between $p = q$ versus
	\[
	p \ge q + \sqrt{\eps q/2} \ge 2q
	\]
	as $\eps > 2q$. 
	Therefore, we can apply the identity tester from Lemma \ref{lem:Ber-id-test} for $\gamma = 1$ with sample complexity
	\[
	m = O\left( \frac{1}{q} \right) 
	= O\left( \frac{\ln(1/q)}{\eps} \right)
	\]
	since $\eps \le 2q \ln(1/q)$.

	\medskip
	\noindent\textbf{Case 3: $\eps > \max \{2q,\, 2q \ln(1/q)\}$.}
	Just as in Case 2, if $\bkl(p,q) \ge \eps$ then one must have $p \ge q$, since $p < q$ implies 
	\[
	\bkl(p,q) 
	\le \bkl(0,q) 
	= \ln\left(\frac{1}{1-q}\right)
	\le \frac{q}{1-q} \le 2q.
	\]
	Since $p \ge q$, we have $1-p \le 1-q$ and thus 
	\[
	\eps \le \bkl(p,q) = p \ln\frac{p}{q} + (1-p) \ln \frac{1-p}{1-q}
	\le p \ln\frac{p}{q} \le p \ln\frac{1}{q}. 
	\]
	Therefore, it suffices to distinguish between $p=q$ and $p \ge \eps / \ln(1/q) > 2q$. 
	The identity tester from Lemma \ref{lem:Ber-id-test} for $\gamma = \eps/(q\ln(1/q)) - 1 \ge 1$ can achieve $2/3$ success probability with sample complexity
	\[
	m 
	= O\left( \frac{\ln(1/q)}{\eps} \right).
	\]
	This completes the proof of the lemma.
\end{proof}

\subsection{Applications}
\label{subsec:application}

Here we give several applications of Theorem \ref{thm:alg-main-detailed}. 

\subsubsection{Product distributions}
\label{subsub:prod}

For each $i \in [n]$ let $\mu_i$ be an arbitrary distribution over $\QQ$, 
and define a product distribution $\mu = \mu_1 \otimes \cdots \otimes \mu_n$ over $\QQ^n$. 
It is well-known that every product distribution satisfies approximate tensorization of entropy with an optimal constant $C=1$.

\begin{lemma}[\cite{Cesi,MSW03,CMT}]
	\label{lem:prod-AT}
	Let $\mu$ be any product distribution over $\QQ^n$. 
	For any distribution $\pi$ over $\QQ^n$ such that $\pi \ll \mu$, we have
	\[
	\kl{\pi}{\mu} \le \sum_{i=1}^n \E_{x \sim \pi_{n \setminus i}} \Big[ \kl{\pi_i(\cdot \mid x)}{\mu_i(\cdot \mid x)} \Big].
	\]
	Namely, every product distribution satisfies approximate tensorization of entropy with constant $1$. 
\end{lemma} 

For a product distribution $\mu$, define $\eta(\mu) = \min_{i \in [n]} \min_{a \in \QQ:\, \mu_i(a) > 0} \mu_i(a)$.
Observe that $\mu$ is $\eta(\mu)$-balanced. 
Let $\PP(\eta)$ denote the collection of all product distributions $\mu$ such that $\eta(\mu) \ge \eta$.
The following corollary follows immediately from Theorem \ref{thm:alg-main-detailed} and Lemma \ref{lem:prod-AT}.

\begin{corollary}
	Let $\eta \in (0,1/2]$ be real. 
	There is a polynomial-time identity testing algorithm for the family $\PP(\eta)$ of $\eta$-balanced product distributions with query access to both $\coorora$ and $\stanora$ and for KL divergence with distance parameter $\eps > 0$. 
	The query complexity of the identity testing algorithm is $O((n/\eps) \log^3(n/\eps))$.
\end{corollary}

\subsubsection{Sparse Ising models in the uniqueness region}
\label{subsec:sparse:ising}

An Ising model is a tuple $(G,\beta,h)$ where
\begin{itemize}
	\item $G = (V,E)$ is a finite simple graph;
	\item $\beta: E \to \R$ is a function of edge couplings;
	\item $h: V \to \R$ is a function of vertex external fields.
\end{itemize}
We may also view $\beta$ and $h$ as vectors; in particular, we write $\beta_{uv}$ to represent the edge coupling of an edge $\{u,v\} \in E$, and write $h_v$ to represent the external field of a vertex $v \in V$.

The Gibbs distribution of an Ising model $(G,\beta,h)$ is given by
\[
\mu_{(G,\beta,h)}(\sigma) = \frac{1}{Z_{(G,\beta,h)}} \exp \left( \sum_{\{u,v\} \in E} \beta_{uv} \sigma_u \sigma_v + \sum_{v\in V} h_v \sigma_v \right),
\quad \forall \sigma \in \{+,-\}^V,
\]
where 
\[
Z_{(G,\beta,h)} = \sum_{\sigma \in \{+,-\}^V} \exp \left( \sum_{\{u,v\} \in E} \beta_{uv} \sigma_u \sigma_v + \sum_{v\in V} h_v \sigma_v \right)
\]
is the partition function.


\begin{definition}[The family $\II\SS(\Delta,\delta,h^*)$ of Ising models in tree-uniqueness]
	For an integer $\Delta \ge 3$ and reals $\delta \in (0,1),h^* > 0$, let $\II\SS(\Delta,\delta,h^*)$ be the family of Gibbs distributions of Ising models $(G,\beta,h)$ satisfying:
	\begin{enumerate}[(1)]
		\item The maximum degree of $G$ is at most $\Delta$;
		\item We have $(\Delta-1) \tanh(\beta^*) \le 1- \delta$, where $\beta^* = \max_{\{u,v\} \in E} |\beta_{uv}|$ denotes the maximum edge coupling in absolute value; 
		\item For each $v \in V(G)$, we have $|h_v| \le h^*$. 
	\end{enumerate}
\end{definition}

Recent works towards establishing optimal mixing of Glauber dynamics have shown approximate tensorization of entropy for the family $\II\SS(\Delta,\delta,h^*)$. 

\begin{lemma}[\cite{CLV20,CLV-STOC21}]
	\label{lem:Ising-AT}
	For any integer $\Delta \ge 3$ and reals $\delta \in (0,1),h^* > 0$, 
	there exists a constant $C = C(\Delta,\delta,h^*) \ge 1$, such that
	every Ising distribution $\mu$ from the family $\II\SS(\Delta,\delta,h^*)$ satisfies approximate tensorization of entropy with constant $C$.
\end{lemma}

We then deduce the following corollary from Theorem \ref{thm:alg-main-detailed} and Lemma \ref{lem:Ising-AT}. 

\begin{corollary}
	Suppose $\Delta \ge 3$ is an integer and $\delta \in (0,1), h^* > 0$ are reals. 
	There is a polynomial-time identity testing algorithm for the family $\II\SS(\Delta,\delta,h^*)$ of Ising models with query access to both $\coorora$ and $\stanora$ and for KL divergence with distance parameter $\eps > 0$. 
	The query complexity of the identity testing algorithm is $O((n/\eps) \log^3(n/\eps))$.
\end{corollary}

\subsubsection{Distributions satisfying Dobrushin uniqueness condition}
\label{subsubsec:Dob}

Let $\mu$ be a distribution over $\QQ^n$. 
For $i,j \in [n]$, the Dobrushin influence of $i$ on $j$ is given by
\[
a_{u,v} = \max_{(x,x') \in \CC_{i,j}}\, \tv{\mu_j(\cdot \mid X_{n \setminus j} = x)}{\,\mu_j(\cdot \mid X_{n \setminus j} = x')}, 
\]
where $\CC_{i,j}$ denotes the collection of all pairs $(x,x')$ of vectors in $\QQ^{n \setminus i}$ such that $\mu(X_{n \setminus j} = x) > 0$, $\mu(X_{n \setminus j} = x') > 0$, and $x,x'$ either are the same or differ exactly at the coordinate $i$. 
The Dobrushin influence matrix $A$ is an $n \times n$ matrix with entries given as above. 
Note that $A$ is not symmetric in general. 

For $b\in (0,1/2]$, we say the distribution $\mu$ is $b$-marginally bounded if for every $\Lambda \subseteq [n]$, every $x \in \QQ^{\Lambda}$ with $\mu(X_{\Lambda} = x) > 0$, every $i \in [n] \setminus \Lambda$, and every $a \in \QQ$, one has
\[
\text{either~~} \mu(X_i = a \mid X_{\Lambda} = x) \ge b 
\quad
\text{or~~} \mu(X_i = a \mid X_{\Lambda} = x) = 0. 
\]
Note that though seemingly similar, the notion of marginal boundedness is not the same as the coordinate balancedness defined in Section~\ref{def:balanced}. 
We observe that any $b$-marginally bounded distribution is also $b$-balanced.

For $\delta \in (0,1)$ and $b \in (0,1/2]$, let $\DD(\delta,b)$ be the family of all distributions over $\QQ^n$ satisfying the following conditions:
\begin{enumerate}[(1)]
	\item The Dobrushin influence matrix $A$ of $\mu$ satisfies $\norm{A}_2 \le 1-\delta$;
	\item $\mu$ is $b$-marginally bounded.
\end{enumerate}
Marton proved that every distribution from the family $\DD(\delta,b)$ satisfies approximate tensorization of entropy.

\begin{lemma}[\cite{Marton19}]
	\label{lem:Marton}
	Suppose $\delta \in (0,1)$ and $b \in (0,1/2]$ are reals. 
	Every distribution $\mu$ from the family $\DD(\delta,b)$ satisfies approximate tensorization of entropy with constant $C = 1/(b \delta^2)$.
\end{lemma}

The following corollary follows from Theorem \ref{thm:alg-main-detailed} and Lemma \ref{lem:Marton}. 

\begin{corollary}\label{coro:dob}
	Suppose $\delta \in (0,1)$ and $b \in (0,1/2]$ are reals. 
	There is a polynomial-time identity testing algorithm for the family $\DD(\delta,b)$ with query access to both the $\coorora$ and the $\stanora$ and for KL divergence with distance parameter $\eps > 0$. 
	The query complexity of the identity testing algorithm is $O((n/\eps) \log^3(n/\eps))$.
\end{corollary}

For Ising models, there is also a stronger version of Dobrushin uniqueness in literature. 

\begin{definition}[The family $\II\SS_{\mathrm{D}}(\delta,h^*)$ of Ising models in Dobrushin-uniqueness]
	For $\delta \in (0,1)$ and $h^* > 0$, let $\II\SS_{\mathrm{D}}(\delta,h^*)$ be the family of Gibbs distributions of Ising models $(G,\beta,h)$ satisfying:
	\begin{enumerate}[(1)]
		\item For each $v \in V(G)$, we have $\sum_{u \in N(v)} |\beta_{uv}| \le 1-\delta$; 
		\item For each $v \in V(G)$, we have $|h_v| \le h^*$. 
	\end{enumerate}
\end{definition}

Notice that in the Ising model we have $a_{u,v} \le \tanh (|\beta_{uv}|) \le |\beta_{uv}|$ for $\{u,v\} \in E$ and $a_{u,v} = 0$ for non-edges. 
So we have $\II\SS_{\mathrm{D}}(\delta,h^*) \subseteq \DD(\delta,b)$ for $b \ge 1/(e^{2(h^*+1)}+1)$. 
Hence, the following corollary follows immediately from Corollary~\ref{coro:dob}.

\begin{corollary}
	Suppose $\delta \in (0,1)$ and $h^* > 0$ are reals. 
	There is a polynomial-time identity testing algorithm for the family $\II\SS_{\mathrm{D}}(\delta,h^*)$ of Ising models with query access to both $\coorora$ and $\stanora$ and for KL divergence with distance parameter $\eps > 0$. 
	The query complexity of the identity testing algorithm is $O((n/\eps) \log^3(n/\eps))$.
\end{corollary}

\subsection{Identity testing for TV distance}
\label{subsec:tv-vs-kl}

One of the main goals of this paper is to give efficient identity testing algorithms \emph{without any restriction} on the noisy, unknown distribution $\pi$. 
However, since we work with KL divergence in most parts of our algorithmic results, one assumption we have to make is that the support of the hidden distribution $\pi$ is contained in that of the visible $\mu$, denoted by $\pi \ll \mu$. 
This is necessary for the KL divergence $\kl{\pi}{\mu}$ to be finite.
However, we emphasize that this assumption is fairly mild and does not introduce any restriction in many settings for the following two reasons.
(1) In many cases the visible distribution $\mu$ is already fully supported on $\QQ^n$ and hence the hidden one $\pi$ can be arbitrary, e.g., $\mu$ is the uniform distribution or from an Ising model.
(2) Testing algorithms for KL divergence can be easily applied as a black box to obtain identity testing algorithms for TV distance, where in the latter we do not require $\pi \ll \mu$.
Here we show how our identity testing algorithm Algorithm~\ref{alg:id-test} can be used to test for TV distance. 

\begin{lemma}
	Suppose $\AA_{\textsc{kl-id}}$ is an identity testing algorithm for a family $\FF$ of distributions with query access to both $\coorora$ and $\stanora$ and for KL divergence with distance parameter $\eps > 0$. 
	The query complexity of $\AA_{\textsc{kl-id}}$ is $m(n,1/\eps)$ and the running time of $\AA_{\textsc{kl-id}}$ is polynomial in $n$ and $1/\eps$.  
	Then there exists a polynomial-time identity testing algorithm $\AA_{\textsc{tv-id}}$ for $\FF$ with the same query access and for TV distance with distance parameter $\eps > 0$.
	The query complexity of $\AA_{\textsc{tv-id}}$ is $O(m(n,2/\eps^2) + 1/\eps)$.
\end{lemma}

\begin{proof}
	Let $\XX_\mu \subseteq \XX$ denote the support of $\mu$.
	By the law of total probability we have $\pi(\cdot) = \pi(\XX_\mu) \,\pi(\cdot \mid \XX_\mu) + \pi(\XX_\mu^\complement) \,\pi(\cdot \mid \XX_\mu^\complement)$ where $\XX_\mu^\complement = \XX \setminus \XX_\mu$ is the complement. 
	Therefore, we obtain from the triangle inequality that
	\[
	\tv{\pi}{\mu} \le \pi(\XX_\mu^\complement) + \tv{\pi(\cdot \mid \XX_\mu)}{\mu}.
	\] 
	In particular, if $\tv{\pi}{\mu} \ge \eps$, then either $\pi(\XX_\mu^\complement) \ge \eps/2$ or $\tv{\pi(\cdot \mid \XX_\mu)}{\mu} \ge \eps/2$, where the latter implies $\kl{\pi(\cdot \mid \XX_\mu)}{\mu} \ge \eps^2/2$ via the Pinsker's inequality. 
	
	Our testing algorithm $\AA_{\textsc{tv-id}}$ runs in two stages. 
	In the first stage, 
	we distinguish between $\pi(\XX_\mu^\complement) = 0$ versus $\pi(\XX_\mu^\complement) \ge \eps/2$ using $O(1/\eps)$ samples from $\pi$, 
	and we say $\pi$ passes this stage if none of these samples is in $\XX_\mu^\complement$. 
	In particular, by choosing suitable constants we can make the failure probability at most $1/3$, i.e., if $\pi(\XX_\mu^\complement) \ge \eps/2$ then the probability that $\pi$ passes is at most $1/3$. 
	Observe that if $\pi(\XX_\mu^\complement) = 0$ then it always passes the first stage. 
	
	In the second stage, we test between $\pi(\cdot \mid \XX_\mu) = \mu$ versus $\kl{\pi(\cdot \mid \XX_\mu)}{\mu} \ge \eps^2/2$, using $\AA_{\textsc{kl-id}}$ with failure probability $1/3$. 
	Note that if we saw samples that belong to $\XX_\mu^\complement$ when running $\AA_{\textsc{kl-id}}$, either from calls of $\coorora$ or from calls of $\stanora$, we can safely conclude that $\pi \neq \mu$ and hence $\tv{\pi}{\mu} \ge \eps$.
	Otherwise, these samples can be viewed as generated perfectly from the conditional distribution $\pi(\cdot \mid \XX_\mu)$. 
	We say $\pi$ passes the second stage if $\AA_{\textsc{kl-id}}$ outputs Yes (i.e., $\pi(\cdot \mid \XX_\mu) = \mu$).
	
	If $\pi$ passes both stages then $\AA_{\textsc{tv-id}}$ outputs Yes (i.e., $\pi = \mu$); otherwise it outputs No (i.e., $\tv{\pi}{\mu} \ge \eps$).
	Observe that, if $\pi = \mu$ then it passes the first stage always and passes the second stage with probability at least $2/3$. 
	Meanwhile, if $\tv{\pi}{\mu} \ge \eps$ then either $\pi(\XX_\mu^\complement) \ge \eps/2$ or $\kl{\pi(\cdot \mid \XX_\mu)}{\mu} \ge \eps^2/2$. 
	If $\pi(\XX_\mu^\complement) \ge \eps/2$ then it passes the first stage with probability at most $1/3$. 
	And if $\kl{\pi(\cdot \mid \XX_\mu)}{\mu} \ge \eps^2/2$ it passes the second stage with probability at most $1/3$. 
	Hence, the probability that $\pi$ passes both stages is at most $1/3$. 
	Therefore, $\AA_{\textsc{tv-id}}$ is a polynomial-time identity testing algorithm with sample complexity $O(m(n,2/\eps^2) + 1/\eps)$. 
\end{proof}

\section{Hardness of Identity Testing When Approximate Tensorization Fails}


In this section we show that approximate tensorization is essentially a necessary condition for efficient identity testing, in the sense that there are high-dimensional distributions, specifically the antiferromagnetic Ising model, for which either approximate tensorization holds with constant $C=O(1)$ (and thus there is an efficient identity algorithm from Theorem \ref{thm:alg-main}), or there is no polynomial-time identity testing algorithm with $\stanora$ and $\coorora$ access unless $\mathsf{RP}=\mathsf{NP}$.

We prove the hardness result in Theorem \ref{thm:comp-hardness} from the introduction in the following sections. We use a reduction from the maximum cut problem to identity testing. In particular, given a hard maximum cut instance, we construct an identity testing instance whose outputs provides the maximum cut. Our reduction is inspired by the one in~\cite{testing-colt}, but we use a different ``degree reducing'' gadget (namely, the one from~\cite{Sly}), and we are also required to design an algorithm to sample from the hidden model we construct. This is challenging because sampling from the antiferromagnetic Ising model is $\mathsf{NP}$-hard in general, but for our instance we manage to do it using a hybrid approach. Specifically, we use the recent algorithm from~\cite{KLR} for low-rank Ising model for one range of parameters and polymer models~\cite{JKP} for the other. Both algorithms rely on the fact that the graph in our testing is a random bipartite graph with trees attached to it that happens to be a good expander. 

Our proof is organized as follows.
First, we introduce our degree reducing gadget in Section~\ref{subsec:vertexgadget}. The testing instance construction and the reduction is then provided in Section~\ref{subsection:reduction}.
Finally, Section~\ref{subsection:sampling} and Section~\ref{subsection:sampling:proofs} contain our sampling algorithm.

\subsection{The degree reducing gadget}
\label{subsec:vertexgadget}

The gadget construction has as parameters integers $n \ge 1$, 
$d \ge 3$ and real numbers $0 < \theta,\psi < 1/8$.
Let $\ell = 2 \lfloor \frac \psi 2 \log_{d-1} n \rfloor$,
$t = (d-1)^{\lfloor \theta \log_{d-1} n \rfloor}$
and	$m = t(d-1)^{\ell}$.
The gadget is constructed as follows:
\begin{enumerate}
	\item Let $\hat G = (V_{\hat G},E_{\hat G})$ 
	be a random bipartite graph with $n+m$ vertices on each side.
	\item For $s \in \{+,-\}$, let the vertices on the $s$-side of $\hat G$ be $W_s \cup U_s$, where $|W_s| = n$ and $|U_s| = m$.
	\item Let $M_1,\dots,M_{d-1}$ be $d-1$ random perfect matchings between $W_+ \cup U_+$ and $W_- \cup U_-$; that is, each $M_i$ is drawn uniformly at random from the set of all perfect matching between $W_+ \cup U_+$ and $W_- \cup U_-$;
	\item Let $M'$ be a random perfect matching between $W_+$ and $W_-$;
	\item Set $E_{\hat G} = M' \cup \left(\bigcup_{i=1}^{d-1} M_i \right)$;
	\item Construct collections $\mathcal T_+$ and $\mathcal T_-$ each of $t$ disjoint $(d-1)$-ary trees of height $\ell$. 
	\item Adjoin $\mathcal T_+$ (resp., $\mathcal T_-$) to $\hat G$
	by identifying each vertex of $U_+$ (resp., of $U_-$) with one of the leafs of the trees in $\mathcal T_+$ (resp., $\mathcal T_-$). We denote the set of roots of the trees in $\mathcal T_+$ (resp., $\mathcal T_-$) by $R_+$ (resp., $R_-$).
\end{enumerate}

Let $G=(V_{G},E_{G})$ be the random multi-graph resulting from this construction.

\subsection{The reduction}
\label{subsection:reduction}

Let $(K=(V_{K},E_{K}),k)$ be an instance of the maximum cut problem. Namely, we want to distinguish between
the cases $\maxcut(K) < k$ and $\maxcut(K) \ge k$, where $\maxcut(K)$ denotes the size of the maximum cut of the graph $K$.

Let $N = |V_K|$; we may assume that $N = n^{\theta/12}$, where $n$ and $\theta$ are the parameters for the degree reducing gadget construction in the previous section.
Form the multi-graph $F = (V_F,E_F)$ by adding two special vertices $s$ and $t$ to $K$ (i.e., $V_F = V_{K} \cup \{s,t\}$), 
connecting $s$ and $t$ with $N^2-k$ edges, 
and adding $N$ edges between each $s$ and $t$ and each vertex in $V_K$;
note that $F$ has $|E_K|+3N^2-k$.
This construction ensures that: 
\begin{enumerate}
	\item When $\maxcut(K) < k$, then $(\{s,t\},V_K)$ is the unique maximum cut of $F$ and has size $2N^2$; 
	\item When $\maxcut(K) \ge k$, then there exists another cut in $F$ whose size is at least $2N^2$; this cut is obtained by taking the maximum cut for $K$ and adding $s$ and $t$ to opposite sides of it.
\end{enumerate}

Next, we generate an instance $G = (V_G,E_G)$ of the degree reducing gadget
from Section~\ref{subsec:vertexgadget}.
We then obtain the multi-graph $\widehat{F}=(V_{\widehat{F}},E_{\widehat{F}})$ by replacing every vertex $v \in V_F$ with a copy $G$; we label each copy of $G$ by $G^v$ and let $R_+^v$ and $R_-^v$ denote $R_+$ and $R_-$ for $G_v$.
Moreover, for each edge $\{u,v\} \in E_F$, we add a matching 
of size $n^{3\theta/4}$ between $R_+^v$ and $R_+^u$, and another matching of the same size between $R_-^v$ and $R_-^u$.
Note that $\widehat{F}$ is a $d$-regular multi-graph.

We will consider the antiferromagnetic Ising model on the multi-graph $\widehat{F}$.
(See~Section~\ref{subsec:sparse:ising} for the definition of the Ising model on a simple graph. The definition extends to the multi-graph setting by simply considering multi-edges in the summation.)
For a configuration $\sigma \in \{+1,-1\}^{V_{G_v}}$, we define its \emph{phase} $Y_v(\sigma)$ as $+1$ if the number of vertices assigned $-1$ in $W_+$ is greater than the number of vertices assigned $-1$ in $W_-$; otherwise we set $Y(\sigma)=-1$.
For a configuration $\sigma \in \{+1,-1\}^{V_{\widehat{F}}}$,
we let $Y(\sigma)$ denote the \emph{phase vector} of $\sigma$, which contains as coordinates the phase of $\sigma$ in each gadget $G_v$.

Let $\Omega = \{+1,-1\}^{V_F}$ be the set of all phase vectors. 
Let $\xi_{st}^+ \in \Omega$ (resp., $\xi_{st}^- \in \Omega$) be the phase vector that assigns $+1$ (resp., $-1$) to $s,t$ and $-1$ (resp., $+1$) to every other gadget in $\widehat{F}$. Let $\Omega_{st} = \{\xi_{st}^+ ,\xi_{st}^- \}$.
Observe that each phase vector $Y(\sigma)$ corresponds to a cut in the graph $F$, with the phase determining the side of the cut for each vertex.

Let $\Omega'_0 \subseteq \Omega$ be the collection of all phase vectors corresponding to cuts $(\{s\} \cup U, \{t\} \cup V_F \setminus U)$ of $F$, which in turn correspond to cuts $(U,V_K\setminus U)$ of $K$ of size $< k$. 
Let $\Omega_0$ be $\Omega'_0$ together with the phase vectors for cuts $(\{s,t\} \cup U, V_F \setminus U)$ of $F$.
Then:
\begin{enumerate}
	\item if $\maxcut(K) < k$, then $\Omega_0 = \Omega$;
	\item if $\maxcut(K) \ge k$, then $\Omega_0 \subsetneq \Omega$ and $\Omega \setminus \Omega_0$ contains at least one phase vector corresponding to a cut $(\{s\} \cup U, \{t\} \cup V_F \setminus U)$ of $F$, where $( U, V_K \setminus U)$ is a maximum cut for $K$.
\end{enumerate}

We are now ready to describe our instance for the identity testing problem. Let $\beta < \beta_c(d) := -\frac{1}{2} \ln (\frac{d}{d-2})$; 
this parameter regime corresponds to the so-called tree uniqueness region for $(d-1)$-ary infinite trees.
The visible distribution of our testing instance will be the Gibbs distribution $\mu_{\widehat{F},\beta}$ for the antiferromagnetic Ising model on $\widehat{F}$. 
The hidden distribution will be $\mu_{\widehat{F},\beta} (\cdot \mid Y(\sigma) \in \Omega_0)$; that is, $\mu_{\widehat{F},\beta} $ conditioned on the phase vector being in $\Omega_0$.
Our construction ensures that 	
if $\maxcut(K) < k$, then $\mu_{\widehat{F},\beta}  (\cdot \mid Y(\sigma) \in \Omega_0) = \mu_{\widehat{F},\beta}$.
In addition, we have the following fact.

\begin{lemma}
	\label{lemma:tv:1}
	If $\maxcut(K) \ge k$ and $\beta < \beta_c(d)$, then $\tv{\mu_{\widehat{F},\beta} (\cdot \mid Y(\sigma) \in \Omega_0)}{\mu_{\widehat{F},\beta}} = 1-o(1)$. 
\end{lemma}
\begin{proof} 
	Observe that
	$$
	\tv{\mu_{\widehat{F},\beta} (\cdot \mid Y(\sigma) \in \Omega_0)}{\mu_{\widehat{F},\beta}} = \sum_{\sigma: Y(\sigma) \in \Omega\setminus\Omega_0} \mu_{\widehat{F},\beta}(\sigma).
	$$
	Since $\maxcut(K) \ge k$, the set $\Omega \setminus \Omega_0$ contains (at least) the phase vector corresponding to a maximum cut of $F$. Hence, $\sum_{\sigma: Y(\sigma) \in \Omega\setminus\Omega_0} \mu_{\widehat{F},\beta}$ is at least the probability that a sample from $\mu_{\widehat{F},\beta}$
	reveals a maximum cut for $F$.
	The results in~\cite{Sly,GSV:ising,cai2016hardness} imply that this probability is indeed $1-1/2^{n^{\theta/4}}$, as desired.
	Specifically, the argument in the proof of Theorems 1 and 2 in~\cite{Sly} shows that this holds (under certain conditions) for the hard-core model;~\cite{GSV:ising} extends the argument for any antiferromagnetic spin system (including the Ising model); and Lemma 22 from~\cite{cai2016hardness} shows that the required condition holds for all $\beta < -\frac{1}{2} \ln (\frac{d}{d-2})$ in the tree uniqueness region.				
\end{proof}	

The idea of our reduction is to provide this testing instance to 
a presumed polynomial-time identity testing algorithm and use its output to determine
whether $\mu_{\widehat{F},\beta}  (\cdot \mid Y(\sigma) \in \Omega_0) = \mu_{\widehat{F},\beta}$ or $$\tv{\mu_{\widehat{F},\beta} (\cdot \mid Y(\sigma) \in \Omega_0)}{\mu_{\widehat{F},\beta} (\cdot)} = 1-o(1).$$
This gives whether $\Omega_0 = \Omega$ or not, and thus whether the $\maxcut(K) < k$ or not, which would solve the maximum cut problem in randomized polynomial time and  imply that there is no polynomial-time identity testing algorithm unless $\mathsf{RP}=\mathsf{NP}$.

All that remains to complete the reduction is that we show how to sample (in polynomial time) from the hidden distribution $\mu_{\widehat{F},\beta} (\cdot \mid Y(\sigma) \in \Omega_0)$
and how to simulate the $\coorora$ for it.

Simulating the conditional marginal oracle for $\mu_{\widehat{F},\beta} (\cdot \mid Y(\sigma) \in \Omega_0)$ is straightforward. Given a vertex $v \in V_{\widehat{F}}$ and   
a configuration $\sigma \in \{+1,-1\}^{ V_{\widehat{F}} \setminus \{v\}}$, we can first check if $Y(\sigma) \not
\in \Omega_0$; if this the case, we output $\{+1,-1\}$ arbitrarily.
Otherwise, we sample from the vertex marginal $\mu_{\widehat{F},\beta}(\cdot\mid\sigma)$, which can be done in $O(d)$ time.
Sampling from $\mu_{\widehat{F},\beta}  (\cdot \mid Y(\sigma) \in \Omega_0)$ is much trickier, but it can be done relaying heavily on the structure of the graph $\widehat{F}$; note that the problem of approximately sampling antiferromagnetic is computationally hard, even in the bounded degree case. We prove the following.

\begin{lemma}
	\label{lemma:sampling}
	For any $\varepsilon \in (0,1)$ and any phase vector $\mathcal Y \in \Omega_0$
	there is an algorithm that generates a sample from a distribution $\mu_{\textsc{alg}}$ such that 
	$\tv{\mu_{\textsc{alg}}}{\mu_{\widehat{F},\beta}(\cdot \mid \mathcal Y)} \le \varepsilon + \ln(1/\varepsilon) e^{-\Omega(n^{\theta/4})}$
	with running time $\poly(|V_{\widehat{F}}|,1/\varepsilon)$.
\end{lemma}

The proof of this lemma is provided in Section~\ref{subsection:sampling}. We are now ready to prove~Theorem \ref{thm:comp-hardness} from the introduction.

\begin{proof}[Proof of~Theorem \ref{thm:comp-hardness}]
	The first part of the theorem was proved in Section~\ref{subsec:sparse:ising}.
	For the second part, suppose there is an identity testing algorithm with polynomial running time and sample complexity. 
	
	Let $(K=(V_K,E_K),k)$ be and instance of the maximum cut problem with $|V_K| = n^{\theta/12}$.
	Set $\mu_{\widehat F,\beta}$ to be the visible distribution
	and $\mu_{\widehat{F},\beta}  (\cdot \mid Y(\sigma) \in \Omega_0)$ to be the hidden one. 
	Suppose $L = \poly(n)$ is the sample complexity of the testing algorithm in this instance,
	Generate a set $\mathcal S$ of $L$ samples from the distribution $\mu_{\textsc{alg}}$ from Lemma~\ref{lemma:sampling} setting $\varepsilon = 1/(100L)$, so that
	$$
	\tv{\mu_{\textsc{alg}}^{\otimes L}}{\mu_{\widehat{F},\beta}^{\otimes L}  (\cdot \mid Y(\sigma) \in \Omega_0)} \le L \cdot 	\tv{\mu_{\textsc{alg}}}{\mu_{\widehat{F},\beta}  (\cdot \mid Y(\sigma) \in \Omega_0)} \le \frac {1}{50},
	$$
	where $\mu_{\textsc{alg}}^{\otimes L}$ and $\mu_{\widehat{F},\beta}^{\otimes L}  (\cdot \mid Y(\sigma) \in \Omega_0)$ denote the product distributions corresponding to $L$ independent samples from $\mu_{\textsc{alg}}$ and $\mu_{\widehat{F},\beta}  (\cdot \mid Y(\sigma) \in \Omega_0)$ respectively.
	
	Our algorithm for solving $(K=(V_K,E_K),k)$ gives $\mathcal{S}$ to the testing algorithm.	
	Recall that our construction ensures that 	
	if $\maxcut(K) < k$, then $\mu_{\widehat{F},\beta}  (\cdot \mid Y(\sigma) \in \Omega_0) = \mu_{\widehat{F},\beta}$ and 
	that if $\maxcut(K) \ge k$ then 
	\begin{equation}
	\label{eq:tv:p}
	\tv{\mu_{\widehat{F},\beta} (\cdot \mid Y(\sigma) \in \Omega_0)}{\mu_{\widehat{F},\beta}} = 1-o(1);
	\end{equation}
	see Lemma~\ref{lemma:tv:1}.
	
	If $\pi^{\otimes L}$ is the optimal coupling of the distributions $\mu_{\textsc{alg}}^{\otimes L}$ and $\mu_{\widehat{F},\beta}^{\otimes L}  (\cdot \mid Y(\sigma) \in \Omega_0)$, and $(\mathcal{S},\mathcal{S}')$ is sampled from $\pi^{\otimes L}$, then $\mathcal{S}' = \mathcal{S}$ with probability at least $49/50$,
	$\mathcal{S} \sim \mu_{\textsc{alg}}^{\otimes L}$ and $\mathcal{S}' \sim \mu_{\widehat{F},\beta}^{\otimes L}  (\cdot \mid Y(\sigma) \in \Omega_0)$. Therefore,
	if \eqref{eq:tv:p} holds (i.e., $\maxcut(K) \ge k$), then
	\begin{align}
	\label{eq:coupling}
	&\Pr[ \textsc{Tester}~\text{outputs}~\textsc{Yes}~\text{when given samples}~\mathcal{S}~\text{where}~\mathcal{S}\sim\mu_{\textsc{alg}}^{\otimes L}]\notag\\
	={}& \Pr[  \textsc{Tester}~\text{outputs}~\textsc{Yes}~\text{when given samples}~\mathcal{S}~\text{where}~(\mathcal{S},\mathcal{S}')\sim\pi^{\otimes L}]\notag \\
	\leq{}&	\Pr[ \textsc{Tester}~\text{outputs}~\textsc{Yes}~\text{when given samples}~\mathcal{S}'~\text{where}~(\mathcal{S},\mathcal{S}')\sim\pi^{\otimes L} ]+\pi^{\otimes L}(\mathcal{S}\neq \mathcal{S}' )\notag\\
	={}&\Pr[\textsc{Tester}~\text{outputs}~\textsc{Yes}~\text{when given samples}~\mathcal{S}'~\text{where}~\mathcal{S}'\sim\mu_{\widehat{F},\beta}^{\otimes L} (\cdot \mid Y(\sigma) \in \Omega_0)]+\pi^{\otimes L}(\mathcal{S}\neq \mathcal{S}' )\notag\\
	\le{}&  \frac{1}{3} + \frac{1}{50} = \frac{53}{150}.
	\end{align}	
	Hence, the \textsc{Tester} returns \textsc{No} with probability at least $3/5$ in this case.
	
	Now, when $\maxcut(K) < k$ and $\mu_{\widehat{F},\beta}  (\cdot \mid Y(\sigma) \in \Omega_0) = \mu_{\widehat{F},\beta}$,
	we can analogously deduce that the \textsc{Tester} returns \textsc{Yes} with probability at least $2/3$.		
	Therefore, our algorithm can solve any maximum cut instance $(K=(V_K,E_K),k)$ in polynomial time with probability at least $3/5$, and the result follows.
\end{proof}

\subsection{Sampling conditional on the phase vector: proof of \texorpdfstring{Lemma~\ref{lemma:sampling}}{Lemma 5.2}}
\label{subsection:sampling}

We start with a number of definitions and facts required to describe and analyze our algorithm to establish Lemma~\ref{lemma:sampling}. The proofs of these facts
are provided in Section~\ref{subsection:sampling:proofs}.
The first lemma states that it essentially suffices to sample from the simpler conditional distribution $\mu_{\widehat{F},\beta} (\cdot \mid Y(\sigma) \in \Omega_0)$.

\begin{lemma}
	\label{lemma:dominant:phase}
	$\tv{\mu_{\widehat{F},\beta} (\cdot \mid Y(\sigma) \in \Omega_{st})}{\mu_{\widehat{F},\beta} (\cdot \mid Y(\sigma) \in \Omega_0)} \le \frac{1}{2^{n^{\theta/4}}}$.
\end{lemma}

\noindent
We call the roots in $\bigcup_{v \in V_{{F}}} (R_+^v \cup R_-^v)$ used to connect the degree reducing gadgets \emph{ports}.
Let $P$ denote the set of all ports of $\widehat{F}$; we also use $P^v \subset P$ to denote the set of ports of the gadget $G_v$.

For a configuration $\{+1,-1\}^{P_v}$, 
let $Z_{G_v,\beta}(\sigma_{P_v})$ denote the sum of the weights of all the configurations on $G_v$ that agree with $\sigma_{P_v}$. We will need an approximation algorithm for this quantity and an approximate sampling algorithm for $\mu_{G_v,\beta}(\cdot\mid \sigma_{P_v})$.
A fully polynomial-time randomized approximation scheme (FPRAS) for $Z_{G_v,\beta}(\sigma_{P_v})$
is an algorithm that for every $\varepsilon  > 0$ and $\delta \in (0,1)$ outputs $\hat Z$ so that, with probability at least $1-\delta$, $e^{-\varepsilon} \hat Z \leq  Z_{G_v,\beta}(\sigma_{P_v}) \leq  e^\varepsilon \hat Z$ and runs in time polynomial in $|V_{G_v}|$ , $1/\varepsilon$ and $\log(1/\delta)$. A polynomial-time sampling algorithm
for $\mu_{G_v,\beta}(\cdot\mid \sigma_{P_v})$ is a randomized algorithm that for every $\varepsilon  > 0$ runs in time polynomial in $|V_{G_v}|$ and $1/\varepsilon$ and outputs a sample from a distribution $\varepsilon$-close in TV distance to $\mu_{G_v,\beta}(\cdot\mid \sigma_{P_v})$.

\begin{lemma}
	\label{lemma:sampling:couting}
	Let $\sigma_{P_v} \in \{+1,-1\}^{P_v}$ be an arbitrary spin configuration on $P_v$.
	For all sufficiently large $d=O(1)$, with probability $1-o(1)$ over the choice of the random multi-graph $G_v$, for all $\beta < 0$ 
	there is an FPRAS for $Z_{G_v,\beta}(\sigma_{P_v})$ and a polynomial-time sampling algorithm for $\mu_{G_v,\beta}(\cdot \mid \sigma_{P_v})$.
\end{lemma}

For $\mathcal Y \in \Omega$, 
let $\mu_{P,\beta}(\cdot\mid\mathcal Y)$ denote the marginal distribution of $\mu_{\widehat{F},\beta}(\cdot\mid\mathcal Y)$ on $P$.
When $\beta < \beta_c(d)$, in the non-uniqueness regime for the infinite $(d-1)$-ary tree, there are two semi-translation invariant measures, denoted $\mu^+$ and $\mu^-$.
These measures can be obtained by conditioning on
the leaves at level $2h$ (resp., $2h + 1$) to have spin $-1$, 
and then taking the weak limits as $h \rightarrow \infty$. 
Let $p^+$ (resp., $p^-$) be the probability that the root of the tree is assigned $-1$ under $\mu^+$ (resp., $p^-$).

Let $P_+^v = R_+^v \cap P_v$ and $P_-^v = R_-^v \cap P_v$.
For $i \in \{+1,-1\}$, we define the following product distribution over 
configurations $\sigma \in \{+1,-1\}^{P_v}$ on $P_v$:
$$
Q_v^i(\sigma) = (p^i)^{|\sigma^{-}(-1) \cap  P_+^v| } (1-(p^i))^{|\sigma^{-}(+1) \cap  P_+^v| }(p^{-i})^{|\sigma^{-}(-1) \cap  P_-^v| } (1-(p^{-i}))^{|\sigma^{-}(+1) \cap  P_-^v| },
$$	
where $\sigma^{-}(i)$ denotes the set of vertices from $P_v$
assigned spin $i$ in $\sigma$.

The product distribution $Q_v^+$ (resp., $Q_v^-$) is known to be a good approximation for $\mu_{G_v,\beta}(\cdot \mid \mathcal Y_v = +1)$ (resp., $\mu_{G_v,\beta}(\cdot \mid \mathcal Y_v = -1)$), as formalized in the following lemma. Here $\mathcal Y_v$ denote the phase of the gadget $G_v$.

\begin{lemma}[Lemma~22~\cite{cai2016hardness} \& Lemma 19~\cite{GSV:ising}]
	\label{lemma:phase}	
	Let $\beta < \beta_c(d)$. Then,  
	there exists $\theta$ and $\psi$ such that for $s \in \{+1,-1\}$, with probability $1-o(1)$ over the choice of the random $n$-vertex multi-graph $G_v$, 
	for any $\sigma_{P_v} \in \{+1,-1\}^{P_v}$ we have 
	\begin{equation}
	\label{eq:lemma:prod}
	1-n^{-2\theta} \le \frac{\mu_{G_v,\beta}(\sigma_{P_v}\mid\mathcal Y_v = s)}{Q_{v}^{s}(\sigma_{P_v})} \le 1+n^{-2\theta}.
	\end{equation}
	Moreover,  for $s \in \{+1,-1\}$  we have $\mu_{G_v,\beta} (Y_v = s) \ge \frac 1 n$.
\end{lemma}

Next, for a phase vector $\mathcal Y \in \Omega$ we define another product measure this time over 
configurations $\sigma \in \{+1,-1\}^P$ on $P$. Let 
$$
w_P^\mathcal Y(\sigma) = \prod_{v \in V_F} Q_v^{\mathcal Y_v}(\sigma_{P_v}) \prod_{\{u,v\} \in E(P)} e^{\beta \sigma_{u}\sigma_{v}},
$$	
where  $E(P)$  is the set of edges with both endpoints in $P$.
Let $Z_P^\mathcal Y = \sum_{\sigma \in \{+1,-1\}^P} w_P^\mathcal Y(\sigma)$ and define
$$
Q_P^\mathcal Y(\sigma) = \frac{w_P^\mathcal Y(\sigma)}{Z_P^\mathcal Y}.
$$

We have the following approximation for $\mu_{P,\beta}(\cdot \mid \mathcal Y)$ in terms of $Q_P^\mathcal Y$.

\begin{lemma}
	\label{lemma:prod:close}
	Let $\beta < \beta_c(d)$.
	For every $\mathcal Y \in \Omega$ and $\sigma \in \{+1,-1\}^P$, we have
	\[
	\left| \frac{\mu_{P,\beta}(\sigma \mid \mathcal Y)}{Q_P^\mathcal Y(\sigma)} - 1 \right| = o(1). 
	\]
\end{lemma}

Finally, we will also use the following fact.

\begin{lemma}
	\label{lemma:phase:aggreement}
	Let $\mathcal Y \in \Omega$. 
	Suppose $\sigma_P \in \{+1,-1\}^P$ is sampled from 
	$Q_P^{\mathcal Y}$
	and that $\sigma \in \{+1,-1\}^{V_{\widehat{F}}}$ is then sampled from $\mu_{\widehat{F},\beta}(\cdot \mid \sigma_P)$.
	Then, the phase vector of $\sigma$ is $\mathcal Y$ with probability at least $1 - e^{-\Omega(n^{3\theta/4})}$.
\end{lemma}

We can now provide the proof of Lemma~\ref{lemma:sampling}.

\begin{proof}[Proof of~Lemma~\ref{lemma:sampling}]
	Let $\sigma^+$ denote the all-plus configuration on $P$. 		
	For $\sigma \in \{+1,-1\}^P$, recall that we use 
	$Z_{\widehat{F},\beta}(\sigma,\mathcal Y)$ to denote the total weight of the configurations
	of $\widehat{F}$ that agree with $\sigma$ on $P$
	and have phase vector $\mathcal Y$.
	The algorithm is as follows:
	\begin{enumerate}
		\item  Sample $\mathcal Y \in \{ \xi_{st}^+, \xi_{st}^-\}$ uniformly at random. Note that by ignoring all other phase vectors in $\Omega_0$, the error is at most ${1}/{2^{n^{\theta/4}}}$ by Lemma~\ref{lemma:dominant:phase}.
		\item Sample $\sigma_P \in \{+1,-1\}^P$ from a distribution $\varepsilon/3$-close in TV distance to $\mu_{P,\beta}(\cdot\mid\mathcal Y)$ with the following rejection sampling algorithm:
		\begin{enumerate}
			\item[2.1] Generate $\sigma_P \in \{+1,-1\}^P$ from the product distribution $Q_P^\mathcal Y$;
			
			\item[2.2]  Compute the approximation $\hat Z(\sigma_P)$ for $Z_{\widehat{F},\beta}(\sigma_P,\mathcal Y)$ such that
			\[
			\left( 1 - \frac{\eps}{10} \right) Z_{\widehat{F},\beta}(\sigma_P,\mathcal Y) \le \hat Z(\sigma_P) \le \left( 1 + \frac{\eps}{10} \right) Z_{\widehat{F},\beta}(\sigma_P,\mathcal Y);
			\]
			this can be done in time $\poly(|V_{\widehat F}|,1/\varepsilon)$ 
			with success probability at least $1-\varepsilon/3$
			by Lemma~\ref{lemma:sampling:couting}.

			\item[2.3] Accept $\sigma_P$ with probability:
			\[
			r(\sigma_P) = \frac{1}{10} \cdot \frac{Q_P^\mathcal Y(\sigma^+)}{Q_P^\mathcal Y(\sigma_P)} \cdot \frac{\hat Z(\sigma_P)}{\hat Z(\sigma^+)}.
			\]
			\item[2.4] Repeat until accept or exceed $T= c\ln(1/\eps)$  rounds, for a suitable constant $c>0$, in which case we let $\sigma_P = \sigma^+$. 
		\end{enumerate}
		
		\item Sample the configuration of each gadget $G_v$ conditional on the port configuration $\sigma_{P_v}$ on $P_v$ from a distribution $\frac{\varepsilon}{3|V_F|}$-close to $\mu_{G_v,\beta}(\cdot\mid \sigma_{P_v})$ with the algorithm from Lemma~\ref{lemma:sampling:couting}
		\item Output the resulting configuration $\sigma$.
	\end{enumerate}
	
	For the analysis of this algorithm, let us focus first on the rejection sampling process in Step 2.
	First, note that the process is well-defined since:
	\begin{align*}
	r(\sigma_P) &= \frac{1}{10} \cdot \frac{Q_P^\mathcal Y(\sigma^+)}{Q_P^\mathcal Y(\sigma_P)} \cdot \frac{\hat Z(\sigma_P)}{\hat Z(\sigma^+)} 
	\le \frac{1}{5} \cdot \frac{Q_P^\mathcal Y(\sigma^+)}{Q_P^\mathcal Y(\sigma_P)} \cdot \frac{Z_{\widehat{F},\beta}(\sigma_P,\mathcal Y)}{Z_{\widehat{F},\beta}(\sigma^+,\mathcal Y)} \\		
	&= \frac{1}{5} \cdot \frac{\mu_{P,\beta}(\sigma_P \mid \mathcal Y)}{Q_P^\mathcal Y(\sigma_P)} \cdot \frac{Q_P^\mathcal Y(\sigma^+)}{\mu_{P,\beta}(\sigma^+ \mid \mathcal Y)} \le 1,
	\end{align*}
	where the last inequality follows from Lemma~\ref{lemma:prod:close}.
	Second, each iteration of the rejection sampling algorithm can be implemented in polynomial time; in particular, we can compute 
	compute $r(\sigma_P)$ using the FPRAS from Lemma~\ref{lemma:sampling:couting} to obtain $\hat Z(\sigma_P)$ and $\hat Z(\sigma^+)$.		
	Finally, we claim that the output of this algorithm is at least $\eps/3$-close to $\mu_{P,\beta}(\cdot \mid \mathcal Y)$ in TV distance.
	To see this, note that 
	for each $\sigma_P \in \{+1,-1\}^P$, the probability that process outputs $\sigma_P$ in one round is: 
	\[
	Q_P^\mathcal Y(\sigma_P) r(\sigma_P)
	= \frac{1}{10} \cdot Q_P^\mathcal Y(\sigma^+) \cdot \frac{\hat Z(\sigma_P)}{\hat Z(\sigma^+)}  \propto \hat{Z}(\sigma_P).
	\]
	Therefore, conditioned on the algorithm accepting on the first $T = c \ln (1/\varepsilon)$ rounds, the probability that $\sigma_P$ is the output is
	\[
	\frac{\hat{Z}(\sigma_P)}{\sum_{\sigma'} \hat{Z}({\sigma'})} 
	\le \left( 1+\frac{\eps}{5} \right) \frac{Z_{\widehat{F},\beta}(\sigma_P,\mathcal Y)}{\sum_{\sigma'} Z_{\widehat{F},\beta}(\sigma',\mathcal Y)}
	= \left( 1+\frac{\eps}{5} \right) \mu_{P,\beta}(\sigma_P \mid \mathcal Y),
	\]
	and similarly for the lower bound.
	Moreover, since by Lemma~\ref{lemma:prod:close} 
	\begin{align*}
	r(\sigma_P) &= \frac{1}{10} \cdot \frac{Q_P^\mathcal Y(\sigma^+)}{Q_P^\mathcal Y(\sigma_P)} \cdot \frac{\hat Z(\sigma_P)}{\hat Z(\sigma^+)} 
	\ge \frac{1}{20}\cdot \frac{Q_P^\mathcal Y(\sigma^+)}{Q_P^\mathcal Y(\sigma_P)} \cdot \frac{Z_{\widehat{F},\beta}(\sigma_P,\mathcal Y)}{Z_{\widehat{F},\beta}(\sigma^+,\mathcal Y)} 
	\\ &= \frac{1}{20} \cdot\frac{\mu_{P,\beta}(\sigma_P \mid \mathcal Y)}{Q_P^\mathcal Y(\sigma_P)} \cdot \frac{Q_P^\mathcal Y(\sigma^+)}{\mu_{P,\beta}(\sigma^+ \mid \mathcal Y)}  \ge \frac{1}{100},
	\end{align*}
	the probability that the algorithm accepts in the first $T=c \ln (1/\varepsilon)$ rounds is at least
	$$
	1-\left(1-\frac{1}{100}\right)^{c \ln (1/\varepsilon)} \ge 1 -  \frac{\eps}{10},
	$$
	for a suitable constant $c > 0$.
	
	Now, note that by Lemma~\ref{lemma:phase:aggreement} and union bound, the phase vector of $\sigma$ agrees with $\mathcal Y$ with probability at least $1 - |T|e^{-\Omega(n^{3\theta/4})}$; hence, the output distribution of the algorithm satisfies:
	$$
	\tv{\mu_{\textsc{alg}}}{\mu_{\widehat{F},\beta}(\cdot \mid \mathcal Y)} \le \frac{\varepsilon}{3}+\frac{\varepsilon}{3} + |V_F| \cdot \frac{\varepsilon}{3|V_F|} + |T|e^{-\Omega(n^{3\theta/4})} + {2^{-n^{\theta/4}}} \le \varepsilon + \ln(1/\varepsilon) e^{-\Omega(n^{\theta/4})},
	$$
	as claimed.
\end{proof}	

\subsection{Sampling conditional on the phase vector: proof of auxiliary facts}
\label{subsection:sampling:proofs}

We provide in this section the proofs of Lemmas~\ref{lemma:dominant:phase} \ref{lemma:sampling:couting}, \ref{lemma:prod:close}, and \ref{lemma:phase:aggreement}.

\begin{proof}[Proof of~Lemma~\ref{lemma:dominant:phase}]
	We have
	$$
	\tv{\mu_{\widehat{F},\beta} (\cdot \mid Y(\sigma) \in \Omega_0)}{\mu_{\widehat{F},\beta}(\cdot \mid Y(\sigma) \in \Omega_{st})} = \sum_{\sigma: Y(\sigma) \in \Omega_0\setminus\Omega_{st}} \mu_{\widehat{F},\beta}(\sigma \mid Y(\sigma) \in \Omega_0 ).
	$$
	This is the probability of obtaining a phase vector in $\Omega_0\setminus\Omega_{st}$ under $\mu_{\widehat{F},\beta} (\cdot \mid Y(\sigma) \in \Omega_0$.
	Observe that among all the cuts of $F$ corresponding to phase vectors in $\Omega_0$, the largest ones are those corresponding to the phase vectors in $\Omega_{st}$. Hence,
	following the argument in the proof of Lemma~\ref{lemma:tv:1}, 
	we get from the results in~\cite{Sly,GSV:ising,cai2016hardness} 
	that this probability is at least $1-1/2^{n^{\theta/4}}$ as desired.
\end{proof}

%

\begin{proof}[Proof of~Lemma~\ref{lemma:prod:close}]
	Let $\sigma \in \{+1,-1\}^P$. Let $Z_{\widehat{F},\beta}(\sigma,\mathcal Y)$ be the sum of the weights of all the configurations of $\widehat{F}$ with phase vector $\mathcal Y$ that agree with $\sigma$ in $P$ and, similarly, 
	define $Z_{\widehat{F},\beta}(\mathcal Y)$ as the sum of the weights of all configurations with phase vector $\mathcal Y$, so that
	\begin{align*}
	\mu_{P,\beta}(\sigma \mid \mathcal Y) = \frac{Z_{\widehat{F},\beta}(\sigma,\mathcal Y)}{Z_{\widehat{F},\beta}(\mathcal Y)} = \frac{1}{Z_{\widehat{F},\beta}(\mathcal Y)} \cdot \prod_{v \in V_F} Z_{G_v,\beta}(\sigma_{P_v},\mathcal Y_v)\cdot \prod_{\{u,v\} \in E(P)} e^{\beta \sigma_{u}\sigma_{v}}.
	\end{align*}
	
	By Lemma~\ref{lemma:phase}, we have
	\begin{align*}
	\mu_{P,\beta}(\sigma \mid \mathcal Y) &\le  \frac{(1+n^{-2\theta})^{|V_F|}}{Z_{\widehat{F},\beta}(\mathcal Y)} \prod_{v \in V_F} Z_{G_v,\beta}(\mathcal Y_v) Q_{v}^{\mathcal Y_v}(\sigma_{P_v}) \prod_{\{u,v\} \in E(P)} e^{\beta \sigma_{u}\sigma_{v}} \\
	&= \frac{(1+n^{-2\theta})^{|V_F|}}{Z_{\widehat{F},\beta}} \cdot  w_P^\mathcal Y(\sigma) \cdot \prod_{v \in V_F} Z_{G_v,\beta}(\mathcal Y_v).
	\end{align*}
	Then, 
	$$
	\frac{\mu_{P,\beta}(\sigma \mid \mathcal Y)}{Q_P^\mathcal Y(\sigma)} \le (1+n^{-2\theta})^{|V_F|} \cdot \frac{Z_P^\mathcal Y}{Z_{\widehat{F},\beta}} \cdot \prod_{v \in V_F} Z_{G_v,\beta}(\mathcal Y_v).
	$$
	Now,
	\begin{align*}
	Z_P^\mathcal Y \cdot \prod_{v \in V_F} Z_{G_v,\beta}(\mathcal Y_v) &= \sum_{\sigma\in\{+1,-1\}^P} \prod_{\{u,v\} \in E(P)} e^{\beta \sigma_{u}\sigma_{v}} \cdot \prod_{v \in V_F} Q_v^{\mathcal Y_v}(\sigma_{P_v})  Z_{G_v,\beta}(\mathcal Y_v) \notag.
	\end{align*}
	From~\eqref{eq:lemma:prod}, we have
	$$
	\frac{1}{1+n^{-2\theta}}  \mu_{P_v,\beta}(\sigma_{P_v} \mid \mathcal Y_v) \le Q_v^{\mathcal Y_v}(\sigma_{P_v})  \le \frac{1}{1-n^{-2\theta}} \mu_{P_v,\beta}(\sigma_{P_v} \mid \mathcal Y_v),
	$$
	and so
	\begin{align*}
	Z_P^\mathcal Y \cdot \prod_{v \in V_F} &Z_{G_v,\beta}(\mathcal Y_v) 
	\\
	&\le \frac{1}{(1-n^{-2\theta})^{|V_F|}} \sum_{\sigma\in\{+1,-1\}^P} \prod_{\{u,v\} \in E(P)} e^{\beta \sigma_{u}\sigma_{v}} \cdot \prod_{v \in V_F} \mu_{P_v,\beta}(\sigma_{P_v} \mid \mathcal Y_v)  Z_{G_v,\beta}(\mathcal Y_v) \notag\\
	&= \frac{1}{(1-n^{-2\theta})^{|V_F|}} Z_{\widehat{F},\beta}.
	\end{align*}
	Thus, we have obtained the upper bound
	$$
	\frac{\mu_{P,\beta}(\sigma \mid \mathcal Y)}{Q_P^\mathcal Y(\sigma)} \le \left(\frac{1+n^{-2\theta}}{1-n^{-2\theta}}\right)^{|V_F|},
	$$
	and in we can deduce analogously that
	$$
	\frac{\mu_{P,\beta}(\sigma \mid \mathcal Y)}{Q_P^\mathcal Y(\sigma)} \ge \left(\frac{1-n^{-2\theta}}{1+n^{-2\theta}}\right)^{|V_F|}.
	$$
	Recall that $|V_F| = n^{\theta/12}+2$, so that
	$$
	1-o(1) \le \frac{\mu_{P,\beta}(\sigma \mid \mathcal Y)}{Q_P^\mathcal Y(\sigma)}  \le 1+o(1)
	$$
	and the result follows.
\end{proof}

\begin{proof}[Proof of~Lemma~\ref{lemma:phase:aggreement}]
	Consider a gadget $G_v$ of $\widehat{F}$ such that $\mathcal Y_v = +1$.
	Let $u \in P_v^+$.
	We claim that since $p^+ > 1/2$, then $Q_P^\mathcal Y(\sigma_u = -1) > 1/2$.
	To see this, let $w$ be the neighbor of $u$ in $P$ and suppose that the phase of the gadget containing $w$ is $+1$.
	Then,
	$$
	\frac{Q_P^\mathcal Y(\sigma_u = -1)}{Q_P^\mathcal Y(\sigma_u = +1)} = \frac{(p^+)^2e^\beta + p^+(1-p^+)e^{-\beta}}{(1-p^+)^2e^{\beta}+p^+(1-p^+)e^{-\beta}} > 1,
	$$
	which implies that $Q_P^\mathcal Y(\sigma_u = -1) > 1/2$. An analogous calculation shows that the same holds when the gadget containing $w$ is in the $-1$ phase.		
	With the same reasoning, we can similarly deduce that
	when $u \in P_v^-$, then $Q_P^\mathcal Y(\sigma_u = +1) > 1/2$.
	This implies by a Chernoff bound that if $\sigma_P \sim Q_P^\mathcal Y$ and $\mathcal Y_v = +1$, then there exists a constant $\delta > 0$ such that:
	$$
	|\sigma^{-}_{P_v}(-1) \cap P_v^+|  -  |\sigma_{P_v}^{-}(-1) \cap  P_-^v| \ge  \delta|P_v|,
	$$
	and 
	$$
	|\sigma_{P_v}^{-}(+1) \cap  P_-^v| - |\sigma_{P_v}^{-}(+1) \cap  P_+^v| \ge  \delta|P_v|,
	$$
	with probability at least $1 - \exp(-\Omega(|P_v|))$.
	
	Now, 
	\begin{align*}
	\frac{\mu_{G_v,\beta}(\mathcal Y_v = +1\mid \sigma_{P_v})}{\mu_{G_v,\beta}(\mathcal Y_v = -1\mid \sigma_{P_v})} 
	&= \frac{\mu_{G_v,\beta}(\sigma_{P_v} \mid \mathcal Y_v = +1 )}{\mu_{G_v,\beta}(\sigma_{P_v} \mid \mathcal Y_v = -1 )} \cdot \frac{\mu_{G_v,\beta}(\mathcal Y_v = +1 )}{\mu_{G_v,\beta}(\mathcal Y_v = -1)} \\
	&\ge \left(1-\frac{c}{n^{2\theta}}\right)\cdot \frac{Q_v^+(\sigma_{P_v})}{Q_v^-(\sigma_{P_v})} \cdot \frac{\mu_{\widehat{F},\beta}(\mathcal Y_v = +1 )}{\mu_{\widehat{F},\beta}(\mathcal Y_v = -1)},
	\end{align*}
	by Lemma~\ref{lemma:phase}. 
	Lemma~\ref{lemma:phase} also implies that $\frac{\mu_{G_v,\beta}(\mathcal Y_v = +1 )}{\mu_{G_v,\beta}(\mathcal Y_v = -1)} \ge \frac 1n$. Moreover, from the definition of $Q_v^+$ and $Q_v^-$ we have
	\begin{align*}
	\frac{Q_v^+(\sigma_{P_v})}{Q_v^-(\sigma_{P_v})} &= 		
	\frac{(p^+)^{|\sigma_{P_v}^{-}(-1) \cap  P_+^v| } (1-p^+)^{|\sigma_{P_v}^{-}(+1) \cap  P_+^v| }(p^{-})^{|\sigma_{P_v}^{-}(-1) \cap  P_-^v| } (1-p^{-})^{|\sigma_{P_v}^{-}(+1) \cap  P_-^v| }}{(p^-)^{|\sigma_{P_v}^{-}(-1) \cap  P_+^v| } (1-p^-)^{|\sigma_{P_v}^{-}(+1) \cap  P_+^v| }(p^{+})^{|\sigma_{P_v}^{-}(-1) \cap  P_-^v| } (1-p^{+})^{|\sigma_{P_v}^{-}(+1) \cap  P_-^v| }}	\\
	&= \left(\frac{p^+}{p^-}\right)^{|\sigma_{P_v}^{-}(-1) \cap  P_+^v| - |\sigma_{P_v}^{-}(-1) \cap  P_-^v|}	\left(\frac{1-p^-}{1-p^+}\right)^{|\sigma_{P_v}^{-}(+1) \cap  P_-^v| - |\sigma_{P_v}^{-}(+1) \cap  P_+^v| } \\
	&\ge a^{\delta |P_v|},
	\end{align*}
	for a suitable constant $a > 1$ since $p^+ > 1/2$ and $p^-<1/2$.
	This implies that for a suitable constant $c_0 >0$, we have
	$$
	\mu_{G_v,\beta}(\mathcal Y_v = +1\mid \sigma_{P_v}) \ge 1 - \frac{c_0 n}{a^{\delta|P_v|}} \ge 1 - \frac{c_0 n}{a^{\delta n^{3\theta/4}}},
	$$
	since $|P_v| \ge n^{3\theta/4}$.
	Finally, we note that
	$$
	\mu_{\widehat{F},\beta}(\mathcal Y \mid \sigma_P) = \prod_{v \in V_F} \mu_{G_v,\beta}(\mathcal Y_v \mid \sigma_{P_v}) \ge \left(1 - \frac{c_0 n}{a^{\delta n^{3\theta/4}}}\right)^{|V_F|} \ge 1 - \frac{c_0 n^{1+\theta/12}}{a^{\delta n^{3\theta/4}}}
	$$
	since $|V_F| = n^{\theta/12}+2$, and the result follows.		
\end{proof}

\subsubsection{Sampling from the degree reducing gadget}

We focus now in proving Lemma~\ref{lemma:sampling:couting}.
Let $\mu_{G,\beta}$ and $Z_{G,\beta}$ be the antiferromagnetic Ising distribution and its corresponding partition function on a degree reducing gadget $G$.
We need to show to prove Lemma~\ref{lemma:sampling:couting} how approximately sample  from $\mu_{G,\beta}$ and how to compute $Z_{G,\beta}$ when conditioning on an arbitrary configuration on the ports $P$ of $G$. (Note that with a slight abuse of notation we are using $P$ for the set of ports of a single gadget $G$ throughout this section.)
Let $\tau \{+1,-1\}^{P}$ be a configuration on the ports.
Let $Z_{G,\beta}^\tau$ and $\mu_{G,\beta}^\tau$ denote the conditional Ising distribution distribution and the corresponding partition function.

To establish Lemma~\ref{lemma:sampling:couting} we provide two different algorithms: one based on the recent results from~\cite{KLR} that works when $\beta \ge -1/\sqrt{10 d}$,
and another based on polymer models that works when $\beta \le -\frac{c \ln d}{d}$ (for a sufficiently large constant $c >0$), so that each value of the regime $\beta < 0$ is covered by one of these algorithms provided $d$ is a large enough. 

Both algorithms use facts about the spectrum of the multi-graph induced by $V_G \setminus P$.
Hence,
let  $H = (V_H,E_H)$ be the multi-graph that results from removing $P$ from $V_G$. 
Let $A_H$ be the adjacency matrix for the multi-graph $H$; that is, $A_H(u,v)$ is the multiplicity of the edge $\{u,v\}$ in $H$.
For $S \subseteq V_{H}$, let $\partial_e (S)$ be the set of edges from $E_{H}$ with one endpoint in $S$ and one $V_{H} \setminus S$.	
For any real symmetric matrix $Q$, let $\lambda_i(Q)$ denote its $i$-th largest eigenvalue.

\begin{fact}
	\label{fact:spectrum}
	Suppose $d = O(1)$. Then, with probability $1-o(1)$:
	\begin{enumerate}
		\item $d - 2\sqrt{d} - 2 \le \lambda_1(A_H) \le d+2\sqrt{d}$;
		\item $-d - 2\sqrt{d}  \le \lambda_{|V_H|}(A_H) \le -d+2\sqrt{d}+2$;
		\item For $i \ge 2$, $-4\sqrt{d}-2 \le \lambda_2(A_H) \le 4\sqrt{d}+2$;
		\item For every $S \subseteq V_{H}$ such that $|S| \le |V_{H}|/2$, we have $|\partial_e (S)| \ge \frac{d-4\sqrt{d}-2}{2} |S|$.
	\end{enumerate}
\end{fact}	
\begin{proof}
	Consider the symmetric matrices $A$, $B$, and $T$, of dimension $|V_H| \times |V_H|$ defined by:
	\begin{itemize}
		\item $A(u,v) = \kappa$ if $u\in W_+ \cup U_+$ and $v \in W_- \cup U_-$ (or vice versa) and the edge $\{u,v\}$ appears $\kappa$ times in $\bigcup_{i=1}^{d-1} M_i$; all other entries of $A$ are $0$.
		\item $B(u,v) = 1$ if $u\in W_+$ and $v \in W_-$ (or vice versa) and $\{u,v\} \in M$; all other entries of $B$ are~$0$.
		\item $T(u,v) = 1$ if $\{u,v\} \in E_{H}$ and either $u$ or $v$ (or both) are vertices in $V_{H} \setminus ( W_+ \cup U_+ \cup  W_- \cup U_-)$; all other entries of $T$ are $0$.
	\end{itemize} 
	Note that $A_H = A + B + T$, so
	it follows from Weyl's inequality (see~\cite{Weyl}) that
	$$
	\lambda_i(A) + \lambda_{|V_H|}(B) + \lambda_{|V_H|}(T) \le \lambda_i(A_H) \le \lambda_i(A) +  \lambda_1(B) + \lambda_1(T) . 
	$$		
	From Theorem~4 in~\cite{BDH} and contiguity (see Theorem 4 and Corollary 1 in~\cite{MRRW}), we know that $A$ has real eigenvalues 
	$\lambda_1(A) \ge \lambda_2(A) \ge \dots \ge \lambda_{|V_{H}|}(A)$, where 
	$\lambda_1(A) = d-1$, 
	$\lambda_i(A) = -\lambda_{|V_{H}|-i+1}(A)$,
	and 
	$2\sqrt{d} - 1 \le \lambda_2(A) \le 2\sqrt{d} + 1$
	with probability $1-o(1)$.
	The matrix $B$ has eigenvalues $1$ and $-1$.
	Also, all the eigenvalues of the matrix $T$ are real and belong to the interval $[-2\sqrt{d},2\sqrt{d}]$ (see Theorem 3 in~\cite{Godsil}).
	Combining these facts, we obtain parts 1, 2 and 3; part 4 follows from Cheeger's inequality (for multi-graphs).
\end{proof}

\begin{proof}[Proof of~Lemma~\ref{lemma:sampling:couting}]
	Let $J$ be a $|V_H| \times |V_H|$ matrix
	indexed by the vertices of $H$ with entries $J(u,v) =  \beta \cdot A_H(u,v)$ for $u \neq v$ and $J(u,u) = \alpha$ where $\alpha$ is a real number we choose later.
	Let $\partial P \subset V_H$ be the set of vertices of $H$ that were incident to $P$ in $G$. 
	Define a magnetic field $h$ by letting $h_v = \beta$ (resp., $h_v = -\beta$)
	if $v \in \partial P$ and the vertex adjacent to $v$ in $P$
	has $+1$ (resp., $-1$) spin in $\tau$; we set $h_v = 0$ otherwise.
	The Ising model on $H$ with edge interaction $\beta$ and external field $h$ assigns to each configuration $\sigma$ on $H$ probability:
	\begin{align}
	\mu_{H,\beta} (\sigma) &= \frac{1}{Z_{H,\beta}} \exp\Big( \beta \sum_{\{u,v\} \in E_{H}} \sigma_u\sigma_v + \sum_{v \in V_{H}} h_v \sigma_v\Big) 
	= \frac{1}{\hat Z_{H,\beta}} \exp\Big( \frac{1}{2}\langle \sigma,J\sigma\rangle + \langle h,\sigma\rangle
	\Big), \label{eq:ising:matrix}
	\end{align}
	where $\hat Z_{H,\beta} = e^{\alpha |V_{H}|} Z_{H,\beta}$, and we interpret $\sigma$ and $h$ as vectors indexed by the vertices of $H$.
	By construction,
	$Z_{G,\beta}^\tau = Z_{H,\beta}$ and 
	$\mu_{G,\beta}^\tau(\sigma) = \mu_{H,\beta} (\sigma)$ for every $\sigma \in \{+,-\}^{V_{H}}$.
	
	The matrix $J$ has real eigenvalues $\lambda_1(J) \ge \lambda_2(A) \ge \dots \ge \lambda_{|V_{H}|}(J)$.
	To bound the spectrum of $J$,
	we 
	note that $J = \beta A_H + \alpha I$ and
	so 
	$
	\lambda_i(J) = \beta \lambda_i(A_H) + \alpha.
	$
	Hence, setting 
	$\alpha = -\beta(4\sqrt{d} + 2)$ and assuming that $0 > \beta \ge -{1}/(10\sqrt{d})$ and that $d$ is sufficiently large, we obtain from Fact~\ref{fact:spectrum} that with probability $1-o(1)$: 
	$\lambda_1(J) = \Theta(\sqrt{d})$, $\lambda_{|V(H)|}(J) = -\Theta(\sqrt{d})$ and that every other eigenvalue of $J$ is in the interval $[0,1]$.
	Then, Theorem 1.1 from~\cite{KLR} implies that:
	\begin{enumerate}
		\item There is an algorithm that 
		with probability $1 - e^{-{|V_H|}}$
		produces an $e^{\varepsilon}$-multiplicative approximation for $\hat Z_{H,\beta}=Z_{G,\beta}^\tau$ with running time $ \poly(|V_H|,{1}/{\varepsilon})$; and
		\item There is an algorithm to sample from a distribution within $\varepsilon$ TV distance from $\mu_{H,\beta} = \mu_{G,\beta}^\tau$ with running time $\poly(|V_H|,\log(1/\varepsilon))$.
	\end{enumerate}
	Hence, we have established the result for the case when $\beta \ge -\frac{1}{10\sqrt{d}}$.		
	We consider next the case when $\beta \le -\frac{c \ln d}{d}$, for a suitably large constant $c>0$.
	For this, we introduce the notion of polymer models.	
	
	For a fixed configuration $\tau$ in $P$,
	let $P^+ \subset \partial P$ be the set of vertices of $\partial P$
	adjacent to a vertex assigned ``+'' in $\tau$; define $P^- \subset \partial P$ similarly.
	For a configuration $\sigma$ on $H$, let 
	$p^+(\sigma)$ (resp., $p^-(\sigma)$) denote the number of vertices from $P^+$ (resp., $P^-$) that are assigned spin $-1$ (resp. $+1$) in $\sigma$. Let also $D(H,\sigma)$ denote the number of edges incident two vertices with different spins in $\sigma$.
	Then, we can renormalize the Ising distribution~\eqref{eq:ising:matrix} as:
	$$
	\mu_{H,\beta}(\sigma) = \frac{1}{\tilde Z_{H}} e^{-2\beta( D(H,\sigma) + p^+(\sigma) + p^-(\sigma))} =: \frac{w(\sigma)}{\tilde Z_{H,\beta}},
	$$
	where $\tilde Z_{H,\beta} = e^{-\beta (|E_H|+|\partial P|)} Z_{H,\beta}$.
	
	%
	%
	Let $\Omega = \{+1,-1\}^{V_H}$. 
	Observe the graph $H$ is bipartite with partition $(L,R)$ where $L \cup R = V_H$ and $|L|=|R|$. 
	Let $\Omega^{\pm} \subset \Omega$ be the subset of configurations where the number vertices that are assigned $+1$ in $L$
	and $-1$ in $R$ is more than $|V_H|/2$.
	Define $\Omega^{\mp}$ analogously.
	Let $Z_H^{\pm} = \sum_{\sigma\in\Omega^{\pm}} w(\sigma)$ and $Z_H^{\mp} = \sum_{\sigma \in \Omega^{-=}} w(\sigma)$ so that
	$\tilde Z_{H,\beta} = Z_H^{\pm} + Z_H^{\mp}$.
	We define a polymer model whose partition function will serve as a good approximation for $Z_H^{\pm}$ and $Z_H^{\mp}$.
	
	We say $\gamma \subset V_{H}$ is a polymer if 
	the subgraph induced by $\gamma$ is connected and $|\gamma| < |V_H|/2$.
	Two polymers are compatible if the graph distance between them is at least $2$.
	Let $\mathcal G$ be the family of all sets of mutually compatible polymers.
	To each polymer $\gamma$ we assign the weight: 
	$$
	w_\gamma = e^{-2\beta(-|\partial_e(\gamma)| + |P^- \cap L \cap \gamma| + |P^+ \cap R \cap \gamma| - |P^+ \cap L \cap \gamma|- |P^- \cap R \cap \gamma|)}.
	$$
	Define the polymer partition function:
	$$
	\Phi =  \sum_{\Gamma \in \mathcal G} \prod_{\gamma \in \Gamma} w_\gamma.
	$$
	We say $S \subset V_H$ is sparse if every connected component of $S$ has size less than $|V_H|/2$.
	Note that there is a one-to-one correspondence between the sparse subsets of $V_H$ and polymer configurations from $\mathcal G$.
	Then:
	\begin{align*}
	\hat \Phi &:= e^{-2\beta(|E_H|+ |P^- \cap R|+ |P^+ \cap L|)} \cdot \Phi \\
	&= e^{-2\beta(|E_H|+ |P^- \cap R|+ |P^+ \cap L|)}\sum_{\Gamma \in \mathcal G} \prod_{\gamma \in \Gamma} w_\gamma \\
	&= \sum_{S~\text{sparse}} e^{-2\beta(|E_H|-|\partial_e(S)|+
		|P^- \cap R|+ |P^+ \cap L|+ |P^- \cap L \cap S| + |P^+ \cap R \cap S| - |P^+ \cap L \cap S|- |P^- \cap R \cap S|)} \\
	&=
	\sum_{S~\text{sparse}} e^{-2\beta(|E_H|-|\partial_e(S)|+
		|((R\setminus S)\cup(S \cap L))\cap P^-|+|((L\setminus S)\cup(S \cap R)) \cap P^+|)} 
	\end{align*}
	Now, we say $S \subset V_H$ is small if $|S| < |V_H|/2$ (otherwise we say it is large), so that
	$$
	Z_H^{\mp} = \sum_{S~\text{small}} e^{-2\beta(|E_H|-|\partial_e(S)| + |((R\setminus S)\cup(S \cap L))\cap P^-|+|((L\setminus S)\cup(S \cap R)) \cap P^+|)}.
	$$		
	Hence, 
	$$
	0 \le \hat \Phi - Z_H^{\mp} \le \sum_{S~\text{sparse, large}} e^{-2\beta(|E_H|-|\partial_e(S)| + |((R\setminus S)\cup(S \cap L))\cap P^-|+|((L\setminus S)\cup(S \cap R)) \cap P^+|)} .
	$$
	
	If $S$ is sparse, by part 4 of Fact~\ref{fact:spectrum}, each connected component $S_i$ of $S$ satisfies $\partial_e(S_i) \ge \theta |S_i|$ with $\theta = \frac{d-4\sqrt{d}-2}{2}$.
	Summing over the components of $S$ we get
	$\partial_e(S) \ge  \theta |S| \ge \theta |V_H|/2$ when $S$ is large. Then,
	\begin{align*}
	|\hat \Phi - Z_H^{\mp}| &\le \sum_{S~\text{sparse, large}} e^{-2\beta(|E_H|-\theta |V_H|/2+|((R\setminus S)\cup(S \cap L))\cap P^-|+|((L\setminus S)\cup(S \cap R)) \cap P^+|) }
	\end{align*}
	and since $Z_H^{\mp} \ge e^{-2\beta(|E_H| + |L \cap P^+|+|R \cap P^-|)}$ and $|S| < |V_H|/2$, we have
	\begin{equation}	
	\label{eq:final:appox}
	\left|1 - \frac{\hat \Phi}{Z_H^{\mp}}\right| \le 2^{|V_H|} e^{-2\beta (-\theta |V_H|/2 + |V_H|/2)} \le e^{-|V_H|},
	\end{equation}		
	provided $\theta > 1$ and $-\beta \ge \frac{1+\ln 2}{\theta-1}$.
	An analogous argument yields the same bound for $Z_H^{\pm}$.
	
	Our goal now is to use Theorem 8 from~\cite{JKP} to obtain an approximation for $\Phi$ and consequently for $\hat \Phi$, $Z_H^{\mp}$, $Z_H^{\pm}$ and ultimately for $Z_{H,\beta} = Z_H^{\pm}+Z_H^{\mp}$.
	For this, it suffices 
	to check that our polymer model satisfies the so-called Koteck\'{y}-Preiss condition (see, e.g., equation~(3) from~\cite{JKP}). This condition requires that for every polymer $\gamma$:
	\begin{equation}
	\label{eq:kp}
	\sum_{\gamma': d(\gamma,\gamma') \le 1} w_{\gamma'} e^{2|\gamma'|} \le |\gamma|,	
	\end{equation}
	where $d(\cdot,\cdot)$ denotes graph distance. 
	First note that 
	$$
	w_{\gamma'} = e^{-2\beta(-|\partial_e(\gamma')| + |P^- \cap L \cap \gamma'| + |P^+ \cap R \cap \gamma'| - |P^+ \cap L \cap \gamma'|- |P^- \cap R \cap \gamma'|)}\le e^{-2\beta(-\theta/2+1)|\gamma'| }.
	$$
	Hence,
	\begin{align*}
	\sum_{\gamma': d(\gamma,\gamma') \le 1} w_{\gamma'} e^{2|\gamma'|} &\le \sum_{\gamma': d(\gamma,\gamma') \le 1} e^{|\gamma'|(2-2\beta(-\theta/2 +1))} \le \sum_{v \in \gamma \cup \partial_v(\gamma)}\,\sum_{\gamma': v \in \gamma'} e^{|\gamma'|(2-2\beta(-\theta/2 +1))}.
	\end{align*}
	The number of polymers of size $k$ that contain a given vertex is at most $(ed)^{k}$ (see Lemma 2.1 in~\cite{GK}), so 
	$$
	\sum_{\gamma': v \in \gamma'} e^{|\gamma'|(2-2\beta(-\theta/2 +1))} \le \sum_{t \ge 1} \big(d e^{(3-2\beta(-\theta/2 +1))}\big)^t \le \frac{1}{d+1}
	$$
	when $-\beta \ge \frac{3+\ln(d(d+2))}{\theta-2}$ and $\theta > 2$. (Note that the latter is true when $d$ is large enough.) Since $|\gamma \cup \partial_v(\gamma)| \le (d+1)|\gamma|$,~\eqref{eq:kp} follows. Hence, for sufficiently large $d=O(1)$, for a suitable constant $c>0$, Theorem 8 from~\cite{JKP} gives an FPTAS for
	$\Phi$ when $-\beta \ge \frac{c\ln d}{d}$.
	This yields the desired FPTAS for $Z_{H,\beta}$.
	Note that if the desired approximation factor is smaller than $e^{-|V_H|}$, which is the best approximation for $Z_{H,\beta}$ we could using the polymer function $\Phi$ (see~\eqref{eq:final:appox}), then we could use instead brute force for counting and sampling, since the running time would be allowed to be exponential in $|V_H|$. Finally, Theorem 9 from~\cite{JKP} gives the a polynomial-time approximate sampling algorithm for the distribution:
	$$
	\nu(\Gamma) = \frac{\prod_{\gamma \in \Gamma} w_\gamma}{\Phi}.
	$$
	Once a polymer configuration $\Gamma$ is sampled from $\nu$, it can be easily transformed into an Ising configuration by setting the vertices in $L \setminus \Gamma$ and $R \cap \Gamma$ to $+1$ with
	probability $\frac{Z_H^{\pm}}{Z_H^{\pm}+Z_H^{\mp}}$ and
	all other vertices to $-1$, and doing the opposite with the remaining probability.
\end{proof}

\section{Statistical Lower Bounds}
In this section we establish lower bounds on the number of samples required to perform uniformity testing over the hypercube $\{0,1\}^n$ (i.e., $k=2$), with a focus on comparisons between testing for KL divergence and TV distance, and between testing with $\coorora+\stanora$ and $\subcora$. 
Throughout this section, we assume $k=2$ and $\QQ = \{0,1\}$. 
Let $\uu_n$ denote the uniform distribution over $\XX_n = \{0,1\}^n$ for an integer $n \in \N^+$. 
We omit the subscript $n$ when it is clear from context.


\subsection{Statistical lower bounds for \texorpdfstring{$\coorora$}{Coordinate Oracle} model: Proof sketch}

We provide next an overview of our proof approach for Theorem \ref{thm:lb-coor-tv} in which we establish an information-theoretic lower bound for uniformity testing over the binary hypercube $\{0,1\}^n$ in the $\coorora$ model; our proof of Theorem \ref{thm:lb-subc-kl} for the $\subcora$ is similar and we comment on it below.
Our proof follows a well-known strategy. 
We construct a family of ``bad'' distributions $\BB$, each of which has TV distance (or KL divergence) at least~$\eps$ from the uniform distribution over $\{0,1\}^n$. 
Then, the lower bounds follow from, roughly speaking, the fact that the joint distributions of $L$ independent samples from the uniform distribution and of $L$ independent samples from a distribution from $\BB$ (chosen uniformly at random) are close to each other. 

Such an argument works nicely for non-adaptive identity testing algorithms, where the queries are pre-determined before receiving any sample.
In the presence of conditional sampling oracles, we are required to show the lower bounds for adaptive testing algorithms which is necessary with, and so we need to consider the whole \emph{query history}, as in \cite{CRS,Nar}. 
Informally speaking, a query history is a sequence of queries that the testing algorithm asks the oracle along with the outputs from the oracle. 
Each step, the tester determines, possibly at random, a new query based on all previous queries that have been asked and the corresponding outputs from the oracle. 
The output of the testing algorithm can be viewed as a function (possibly randomized) of the query history.

Consequently, we need to show that the following two processes generate close query histories in TV distance. 
In the first process, in each step the algorithm computes a query and the oracle outputs a sample using the uniform distribution. 
In the second one, we first pick a bad distribution $\pi \in\BB$ uniformly at random, and then the oracle outputs samples using~$\pi$.
To show that the two generated query histories are close, we use ideas from \cite{CRS} and also
the so-called \emph{hybrid argument} in cryptography (see, e.g.,~\cite{hybrid}). 
For each $\ell \le L$, we consider a hybrid query history where the first~$\ell$ queries are answered by the oracle using the uniform distribution, while the other $L -\ell$ queries are answered by a single $\pi \in \BB$ chosen uniformly at random. 
It then suffices to show that every pair of ``adjacent'' hybrid query histories are close to each other. 
Since two adjacent hybrid query histories differ only at one step, 
this can be done for a carefully constructed family $\BB$ of bad distributions. 

Our family of bad distributions for the $\coorora$ model is the same as in earlier works \cite{DDK19,CDKS20}.
Each distribution in $\BB$ is constructed by taking a perfect matching of all coordinates (we may assume $n$ is even) 
and considering the distribution such that coordinates from different matched pairs are independent of each other while within each pair the two coordinates are correlated with covariance $\Theta(\eps/\sqrt{n})$. 
Then, one can show that the joint distribution of $O(n/\eps^2)$ samples from the uniform distribution and that from a bad distribution corresponding to a uniformly random perfect matching are close to each other. 
Furthermore, the $\coorora$ does not help in the following sense: 
for the uniform distribution, the $\coorora$ outputs uniform Bernoulli $\Ber(1/2)$ random variables, and
for any distribution from $\BB$ it outputs a sample from $\Ber(1/2 + \xi)$ 
or $\Ber(1/2 - \xi)$ 
where $\xi = \Theta(\eps/\sqrt{n})$. 
We show that to distinguish between a sequence of $\Ber(1/2)$ and a sequence of adaptively chosen $\Ber(1/2 \pm \xi)$ one needs $\Omega(1/\xi^2) = \Omega(n/\eps^2)$ samples in this specific setting; this is proved in Section~\ref{subsec:lb-coor-tv}.

Finally, we briefly describe our construction of the family $\BB$ for establishing Theorem \ref{thm:lb-subc-kl}.
It is inspired by studying approximate tensorization of entropy. 
In particular, our identity testing algorithm fails within $O(n/\eps)$ steps if for most of pairs $(i,x)$ it holds $\pi_i(\cdot \mid x) = \Ber(1/2)$ but only for an $\eps/n$ fraction of the pairs the KL divergence is large, which means we need $\Omega(n/\eps)$ steps to be able to see it. 
For $\subcora$ we would like to construct bad distributions with similar behavior. 
Namely, for most (random) conditionings on a (random) subset of coordinates, the conditional distribution is the same as what one gets from the uniform distribution, and with probability $O(\eps/n)$ the KL divergence between the two conditional distributions is as large as $\Theta(n)$. 
We achieve this using the following type of construction. We pick a random subset $A$ of size $t$ such that $2^t = O(n/\eps)$, and pick a vector $\sigma \in \QQ^n$. 
To generate a sample $x$ from the bad distribution $\pi = \pi_{A,\sigma}$, we first sample $x_A$ uniformly at random. If $x_A \neq \sigma_A$ then the other coordinates are sampled randomly, but if $x_A = \sigma_A$ then we take $x = \sigma$.
One can check, with careful calculations, that such bad distributions satisfy our requirements.
In particular, while the KL divergence for any such bad distribution to the uniform distribution is $\eps$, the TV distance is $\eps/n$ instead, and so a $\Omega(n/\eps)$ lower bound is not a surprise for this construction of family of bad distributions.

\subsection{Uniformity testing with \texorpdfstring{$\subcora$}{Subcube Oracle} for KL divergence}
\label{subsec:lb-subc-kl}

In this subsection we consider uniformity testing over $\{0,1\}^n$ with access to $\subcora$ for KL divergence, and give an information-theoretic lower bound of $\Omega(n/\eps)$ on the number of samples needed. 

Let $\alg$ denote an arbitrary uniformity testing algorithm (possibly randomized and adaptive), and for simplicity let $\ora{\pi}$ denote the $\subcora$ with respect to a distribution $\pi$ over $\{0,1\}^n$.

\begin{definition}
A \emph{pinning} $\tau$ is a partial configuration on a subset of coordinates, namely $\tau \in \{0,1\}^\Lambda$ where $\Lambda \subseteq [n]$.
\end{definition}

Observe that pinnings are exactly the inputs to the subcube oracle $\ora{\pi}$.

\begin{definition}[Query History for $\subcora$]\label{qh:sub}
Let $\mathcal{T}$ be the collection of all pinnings on all subsets of coordinates.
For $L \in \N^+$, define the \emph{(subcube) query history} with respect to $\alg$ and $\ora{\pi}$ of length $L$ to be the random vector in $(\mathcal{T} \times \{0,1\}^n)^L$ generated as follows:

\begin{itemize}
	\item For $i = 1,\dots,L$:
	\begin{itemize}
		\item $\alg$ receives $\left( (\tau_1,x_1),\dots,(\tau_{i-1},x_{i-1}) \right)$ as input and generates $\tau_i \in \mathcal{T}$ (randomly) as output;
		\item $\ora{\pi}$ receives $\tau_i$ as input and generates $x_i \in \{0,1\}^n$ as output.
	\end{itemize}
	\item The (subcube) query history is $H = \left( (\tau_1,x_1),\dots,(\tau_L,x_L) \right)$.
\end{itemize}
\end{definition}

Note that the $\mathsf{output}$ of $\alg$ with sample complexity $L$ is a (randomized) function of the query history $H$ of length $L$.

Our main theorem is stated as below in terms of the query history, from which Theorem \ref{thm:lb-subc-kl} follows immediately.

\begin{theorem}\label{thm:lb-sub-kl}
	Let $n \in \N^+$ be a sufficiently large integer and $\eps > 0$ be a real.
	Let $\uu = \uu_n$ denote the uniform distribution over $\{0,1\}^n$.
	There is no algorithm which can achieve the following properties using only $L \le n/(64\eps)$ samples:
	\begin{itemize}
		\item $\Pr_H\left( \mathsf{output} = \mathsf{Yes} \right) \ge 2/3$ for a random query history $H$ of length $L$ with respect to $\alg$ and $\ora{\uu}$;
		\item $\Pr_{H'}\left( \mathsf{output} = \mathsf{No} \right) \ge 2/3$ for a random query history $H'$ of length $L$ with respect to $\alg$ and $\ora{\pi}$ where $\pi$ is any distribution such that $\kl{\pi}{\uu} \ge \eps$.
	\end{itemize}
\end{theorem}

Our plan, as in many previous works, is to construct a family $\BB$ of bad distributions that are all $\eps$ far away from $\uu$ in KL divergence, such that when picking a bad distribution from $\BB$ uniformly at random and drawing limited number of samples, the joint distributions of these samples are close to that of samples draw from $\uu$. 
We present now our construction of the bad family $\BB$. 
Let $t = \ceil{\log_2(n/\eps)} - 3$ for sufficiently large $n$.
For any $A \subseteq [n]$ with $|A| = t$ and any $\sigma \in \{0,1\}^n$, define the distribution $\pi_{A,\sigma}$ in the following way.
A sample from $\pi_{A,\sigma}$ is generated by:
\begin{itemize}
	\item For each $i \in A$ independently sample $x_i \in \{0,1\}$ uniformly at random;
	\item If $x_A \neq \sigma_A$, then for each $j \in [n] \setminus A$ independently sample $x_j \in \{0,1\}$ uniformly at random and output $x$;
	\item If $X_A = \sigma_A$ then output $x = \sigma$.
\end{itemize}
We remark that all steps are independent.
Finally, we define
\[
\BB = \left\{ \pi_{A,\sigma} : A \in \binom{[n]}{t}, \sigma \in \{0,1\}^n \right\}.
\]

We first show that the distributions in $\BB$ are all bad in the sense that their KL divergence to the uniform distribution is at least $\eps$. 
A key intuition in our construction of $\pi_{A,\sigma}$ here is that while the KL divergence $\kl{\pi_{A,\sigma}}{\uu}= \Theta(\eps)$, the TV distance is much smaller than $\eps$ and is $\tv{\pi_{A,\sigma}}{\uu} = \Theta(2^{-t}) = \Theta(\eps/n)$. 
Hence, intuitively, it will take $\Theta(1/\tv{\pi}{\uu}) = \Theta(n/\eps)$ samples to test between the family $\BB$ and the uniform distribution $\uu$. 
\begin{clm}
	For all $\pi \in \BB$ one has
	\[
	\kl{\pi}{\uu} \ge \eps.
	\]
\end{clm}
\begin{proof}
	Suppose $\pi = \pi_{A,\sigma} \in \BB$ is a bad distribution. 
	By definition we have $\pi(x) = \uu(x) = 2^{-n}$ if $x_A \neq \sigma_A$, and $\pi(\sigma) = 2^{-t}$. 
	Hence, we get
	\[
	\kl{\pi}{\uu} = \pi(\sigma) \ln\left( \frac{\pi(\sigma)}{\uu(\sigma)} \right)
	= \frac{\ln 2}{2^t} (n-t)
	\ge \frac{2 \eps}{n}(n-t) 
	\ge \eps,
	\]
	for $n$ sufficiently large. 
\end{proof}

Define $H$ to be the random query history of length $L$ with respect to $\alg$ and $\ora{\uu}$, and let $\mathsf{output}$ denote the random output with respect to $H$ and $\alg$.
Define $H'$ to be the random query history of length $L$ generated by
\begin{itemize}
	\item Pick $\pi \in \BB$ uniformly at random;
	\item Let $H'$ be the random query history of length $L$ with respect to $\alg$ and $\ora{\pi}$.
\end{itemize} 
Further, let $\mathsf{output}'$ denote the random output with respect to $H'$ and $\alg$.
Our goal is to show that the TV distance between the two query histories $H$ and $H'$ is small and therefore by the data processing inequality the TV distance between $\mathsf{output}$ and $\mathsf{output}'$ is also small so the two properties in Theorem \ref{thm:lb-sub-kl} cannot simultaneously hold. 
\begin{lemma}\label{lem:tv-H-H'}
	For the family $\BB$ of bad distributions, query histories $H,H'$ of length $L \le n/(64\eps)$, and $\mathsf{output},\mathsf{output}'$ defined as above, we have
	\[
	\tv{\mathsf{output}}{\mathsf{output}'} \le \tv{H}{H'} \le \frac{1}{4}.
	\]
\end{lemma}

We present next the proof of Theorem \ref{thm:lb-sub-kl} provided Lemma \ref{lem:tv-H-H'}. 
The proof of the latter is postponed to Section~\ref{subsec:tv-H-H'}.
\begin{proof}[Proof of Theorem~\ref{thm:lb-sub-kl}]
	Suppose for sake of contradiction that $\alg$ satisfies both properties as in Theorem \ref{thm:lb-sub-kl}.
	Then for the family $\BB$ of bad distributions, query histories $H,H'$ of length $L \le n/(64\eps)$, and $\mathsf{output},\mathsf{output}'$ defined as above, we know from these two properties that
	\[
	\Pr\nolimits\left( \mathsf{output} = \mathsf{Yes} \right) \ge 2/3
	\quad\text{and}\quad
	\Pr\nolimits\left( \mathsf{output}' = \mathsf{Yes} \right) \le 1/3.
	\]
	This implies $\tv{\mathsf{output}}{\mathsf{output}'} \ge 1/3$ which contradicts Lemma \ref{lem:tv-H-H'}.  
\end{proof}

\subsubsection{Proof of \texorpdfstring{Lemma \ref{lem:tv-H-H'}}{Lemma 6.3}}
\label{subsec:tv-H-H'}

Our proof is inspired by the hybrid argument from cryptography as in~\cite{CRS}; we flesh out the details of the proof in what follows.

For $0 \le \ell \le L$, define the \emph{hybrid query history} $H^{(\ell)}$ with respect to $\alg$, $\ora{\uu}$, and $\ora{\pi}$ to be the random vector in $(\mathcal{T} \times \{0,1\}^n)^L$ generated as follows:

\begin{itemize}
	\item For $i = 1,\dots,\ell$:
	\begin{itemize}
		\item $\alg$ receives $\left( (\tau_1,x_1),\dots,(\tau_{i-1},x_{i-1}) \right)$ as input and generates $\tau_i \in \mathcal{T}$ (randomly) as output;
		\item $\ora{\uu}$ receives $\tau_i$ as input and generates $x_i \in \{0,1\}^n$ as output.
	\end{itemize}
	
	\item Pick $\pi \in \BB$ uniformly at random;
	
	\item For $i = \ell+1,\dots,L$:
	\begin{itemize}
		\item $\alg$ receives $\left( (\tau_1,x_1),\dots,(\tau_{i-1},x_{i-1}) \right)$ as input and generates $\tau_i \in \mathcal{T}$ (randomly) as output;
		\item $\ora{\pi}$ receives $\tau_i$ as input and generates $x_i \in \{0,1\}^n$ as output.
	\end{itemize}
	
	\item The hybrid query history is $H^{(i)} = \left( (\tau_1,x_1),\dots,(\tau_L,x_L) \right)$.
\end{itemize}
Observe that $H^{(0)} = H'$ and $H^{(L)} = H$ in distribution.
We will prove the following lemma regarding the distance between two adjacent hybrid query histories. 
\begin{lemma}\label{lem:tv-adj-H}
	For every $1 \le \ell \le L$, we have
	\[
	\tv{H^{(\ell-1)}}{H^{(\ell)}} \le \frac{16\eps}{n}.
	\]
\end{lemma}

Note that Lemma \ref{lem:tv-H-H'} is an immediate consequence of Lemma \ref{lem:tv-adj-H}. 
\begin{proof}[Proof of Lemma \ref{lem:tv-H-H'}]
	By the triangle inequality and Lemma \ref{lem:tv-adj-H}, we have that
	\[
	\tv{H}{H'} \le \sum_{\ell = 1}^L \tv{H^{(\ell-1)}}{H^{(\ell)}}
	\le L \cdot \frac{16\eps}{n} \le \frac{1}{4},
	\]
	as claimed.
\end{proof}

It remains to prove Lemma \ref{lem:tv-adj-H}. 
Inspecting the definitions of $H^{(\ell-1)}$ and $H^{(\ell)}$, we see that they only differ locally at one place,
which we describe as follows.
For $0\le i \le L$ let $H_i = \left( (\tau_1,x_1),\dots,(\tau_i,x_i) \right)$ denote the first $i$ entries of a random hybrid query history (notice that $H_0 = \emptyset$).
We write $H^{(\ell-1)}$ and $H^{(\ell)}$ in the following form.

\begin{multicols}{2}
	\begin{enumerate}[(1)]
		\item[] Generation of $H^{(\ell-1)}$:
		\item $H_0 \xrightarrow{\alg,\,\ora{\uu}} H_{\ell-1}$;
		\item $\pi \sim \mathrm{unif}(\BB)$;
		\item $H_{\ell-1} \xrightarrow{\alg} \tau_\ell \xrightarrow{\ora{\pi}} x_\ell$;
		\item $H_\ell \gets H_{\ell-1}$ append $(\tau_\ell,x_\ell)$;
		\item $H_\ell \xrightarrow{\alg,\,\ora{\pi}} H_L = H^{(\ell-1)}$.
		
		\setcounter{enumi}{0}
		
		\item[] Generation of $H^{(\ell)}$:
		\item $H_0 \xrightarrow{\alg,\,\ora{\uu}} H_{\ell-1}$;
		\item $H_{\ell-1} \xrightarrow{\alg} \tau_\ell \xrightarrow{\ora{\uu}} x_\ell$;
		\item $H_\ell \gets H_{\ell-1}$ append $(\tau_\ell,x_\ell)$;
		\item $\pi \sim \mathrm{unif}(\BB)$;
		\item $H_\ell \xrightarrow{\alg,\,\ora{\pi}} H_L = H^{(\ell)}$.
	\end{enumerate}
\end{multicols}

%
%
%
%
%
%
%

In fact the ordering of the steps (2)-(4) can be changed appropriately without having any influence on the final 
distribution of both $H^{(\ell-1)}$ and $H^{(\ell)}$, which will be helpful for a coupling argument. 
We rewrite the generating processes of $H^{(\ell-1)}$ and $H^{(\ell)}$ equivalently as follows:

\begin{multicols}{2}
	\begin{enumerate}[(1)]
		\item[] Generation of $H^{(\ell-1)}$:
		\item $H_0 \xrightarrow{\alg,\,\ora{\uu}} H_{\ell-1}$;
		\item $H_{\ell-1} \xrightarrow{\alg} \tau_\ell$;
		\item $\pi \sim \mathrm{unif}(\BB)$, $\tau_\ell \xrightarrow{\ora{\pi}} x_\ell$;
		\item $H_\ell \gets H_{\ell-1}$ append $(\tau_\ell,x_\ell)$;
		\item $H_\ell \xrightarrow{\alg,\,\ora{\pi}} H_L = H^{(\ell-1)}$.
		
		\setcounter{enumi}{0}
		
		\item[] Generation of $H^{(\ell)}$:
		\item $H_0 \xrightarrow{\alg,\,\ora{\uu}} H_{\ell-1}$;
		\item $H_{\ell-1} \xrightarrow{\alg} \tau_\ell$;
		\item $\pi \sim \mathrm{unif}(\BB)$, $\tau_\ell \xrightarrow{\ora{\uu}} x_\ell$;
		\item $H_\ell \gets H_{\ell-1}$ append $(\tau_\ell,x_\ell)$;
		\item $H_\ell \xrightarrow{\alg,\,\ora{\pi}} H_L = H^{(\ell)}$.
	\end{enumerate}
\end{multicols}

%
%
%
%
%
%
%
%
%

Note that before and after the third step, the two processes have exactly the same steps. 
In the third step for $H^{(\ell-1)}$, we pick a bad distribution $\pi \in \BB$ uniformly at random, and $\ora{\pi}$ receives the pinning $\tau_\ell$ as input and generates $x_\ell \in \{0,1\}^n$ according to $\pi$ conditioned on $\tau_\ell$. 
Meanwhile, in the third step for $H^{(\ell)}$, we still pick a bad distribution $\pi \in \BB$ but do not use it (for now), and $\ora{\uu}$ receives $\tau_\ell$ as input and generates $x_\ell \in \{0,1\}^n$ according to $\uu$ instead of $\pi$. 
It is enough to show that, in this step, conditional on that $H_{\ell-1}$ and $\tau_\ell$ are the same, the $x_\ell$ generated in the two processes are the same with high probability. Since before and after this step the two processes are doing the same thing, we can then couple these two processes to produce the same hybrid query history with high probability, i.e., $H^{(\ell-1)} = H^{(\ell)}$.  

The following technical lemma bounds the probability that $x_\ell$'s are the same in both processes in the third step, which is crucial to us as explained earlier. The proof of it can be found in Section~\ref{subsec:lem:tech}. 

\begin{lemma}\label{lem:tech}
	Let $\tau \in \mathcal{T}$ be an arbitrary pinning on some subset $\Lambda \subseteq V$ of size $m$.
	Then for a random distribution $\pi$ chosen uniformly at random from $\BB$, we have
	\[
	\E_{\pi \sim \mathrm{unif}(\BB)} \left[ \tv{\uu\left( \cdot \mid \tau \right)}{\pi\left( \cdot \mid \tau \right)} \right] \le \frac{16\eps}{n}.
	\]
\end{lemma}

We give below the proof of Lemma \ref{lem:tv-adj-H}.

\begin{proof}[Proof of~Lemma \ref{lem:tv-adj-H}]
	We construct a coupling of $H^{(\ell-1)}$ and $H^{(\ell)}$ via coupling step-by-step the two processes generating $H^{(\ell-1)}$ and $H^{(\ell)}$. 
	Initially $H_0 = \emptyset$ for both processes.
	Then we can couple $H_{\ell-1}$ and $\tau_\ell$ since they are generated in the same way in both processes.
	For the third step, the bad distribution $\pi$ can be chosen to be the same and we deduce from Lemma \ref{lem:tech} that $x_\ell$'s can be coupled with probability at least $1-\eps/n$. 
	After that, suppose we couple $x_\ell$, $H_\ell$ and then the final outputs are coupled. Hence, for this coupling $\mathbb{P}$ we have
	\[
	\tv{H^{(\ell-1)}}{H^{(\ell)}} \le \mathbb{P} \left( H^{(\ell-1)} \neq H^{(\ell)} \right) \le \frac{16\eps}{n},
	\]
	as wanted.
\end{proof}

\subsubsection{Proof of \texorpdfstring{Lemma \ref{lem:tech}}{Lemma 6.5}}
\label{subsec:lem:tech}

Here we give the proof of the technical lemma, Lemma \ref{lem:tech}. 

\begin{proof}[Proof of Lemma \ref{lem:tech}]
	The distribution $\pi \in \BB$ depends on $A$ and $\sigma$. We will show that for any choice of $A \in \binom{[n]}{t}$ one has
	\[
	\E_{\sigma} \left[ \tv{\uu\left( \cdot \mid \tau \right)}{\pi_{A,\sigma}\left( \cdot \mid \tau \right)} \right] \le \frac{16\eps}{n},
	\]
	where $\sigma$ is a uniformly random configuration in $\{0,1\}^n$.
	
	Suppose $|\Lambda| = \ell$. 
	Suppose $|A \cap \Lambda| = j$ and hence $|A \setminus \Lambda| = t-j$. 
	Notice that
	$j \le \min\{t,\ell\}$.
	We partition $\XX = \{0,1\}^n$ into three disjoint subsets.
	
	\medskip
	\noindent\textbf{Case 1.}
	$\XX_1 = \{ \sigma \in \{0,1\}^n: \sigma_{A \cap \Lambda} \neq \tau_{A \cap \Lambda} \}$.
	We have
	\[
	\Pr \left( \sigma \in \XX_1 \right) = \frac{|\XX_1|}{2^n} = 1 - \frac{1}{2^j},
	\]
	and also
	\[
	\tv{\uu\left( \cdot \mid \tau \right)}{\pi_{A,\sigma}\left( \cdot \mid \tau \right)} = 0, \quad\forall \sigma \in \XX_1.
	\]
	
	\medskip
	\noindent\textbf{Case 2.}
	$\XX_2 = \{ \sigma \in \{0,1\}^n: \sigma_{A \cap \Lambda} = \tau_{A \cap \Lambda},\, \sigma_{\Lambda \setminus A} \neq \tau_{\Lambda \setminus A} \}$.
	We have
	\[
	\Pr\nolimits_{\sigma} \left( \sigma \in \XX_2 \right) = \frac{|\XX_2|}{2^n} = \frac{1}{2^j} - \frac{1}{2^\ell}.
	\]
	By definition we have
	\[
	\pi_{A,\sigma}\left( x \mid \tau \right)
	= 
	\begin{cases}
	0, & \text{if}~ x_{A \setminus \Lambda} = \sigma_{A \setminus \Lambda};\\
	\dfrac{1}{2^{n-\ell} - 2^{n-\ell-t+j}}, & \text{if}~ x_{A \setminus \Lambda} \neq \sigma_{A \setminus \Lambda}.
	\end{cases}
	\]
	It follows that
	\[
	\tv{\uu\left( \cdot \mid \tau \right)}{\pi_{A,\sigma}\left( \cdot \mid \tau \right)} = \frac{1}{2^{t-j}}, 
	\quad\forall \sigma \in \XX_2.
	\]
	
	\medskip
	\noindent\textbf{Case 3.}
	$\XX_3 = \{ \sigma \in \{0,1\}^n: \sigma_{\Lambda} = \tau_{\Lambda} \}$.
	We have
	\[
	\Pr\nolimits_{\sigma} \left( \sigma \in \XX_3 \right) = \frac{|\XX_3|}{2^n} = \frac{1}{2^\ell}.
	\]
	By definition we have
	\[
	\pi_{A,\sigma}\left( x \mid \tau \right)
	= 
	\begin{cases}
	\dfrac{\frac{1}{2^n}}{\frac{1}{2^t} + \frac{1}{2^\ell} - \frac{1}{2^{t+\ell-j}}}, & \text{if}~ x_{A \setminus \Lambda} \neq \sigma_{A \setminus \Lambda};\\
	0, & \text{if}~ x_{A \setminus \Lambda} = \sigma_{A \setminus \Lambda} ~\text{and}~ x_{[n] \setminus \Lambda \setminus A} \neq \sigma_{[n] \setminus \Lambda \setminus A};\\
	\dfrac{\frac{1}{2^t}}{\frac{1}{2^t} + \frac{1}{2^\ell} - \frac{1}{2^{t+\ell-j}}}, & \text{if}~ x_{[n] \setminus \Lambda} = \sigma_{[n] \setminus \Lambda}.
	\end{cases}
	\]
	It follows that
	\[
	\tv{\uu\left( \cdot \mid \tau \right)}{\pi_{A,\sigma}\left( \cdot \mid \tau \right)} 
	= \frac{\frac{1}{2^t}}{\frac{1}{2^t} + \frac{1}{2^\ell} - \frac{1}{2^{t+\ell-j}}} - \frac{1}{2^{n-\ell}}
	= \frac{2^\ell}{2^t + 2^\ell - 2^j} - \frac{1}{2^{n-\ell}}, 
	\quad\forall \sigma \in \XX_3.
	\]

	Therefore, combining all three cases we get from the law of total expectation that
	
	\begin{align*}
	\E_{\sigma} \left[ \tv{\uu\left( \cdot \mid \tau \right)}{\pi_{A,\sigma}\left( \cdot \mid \tau \right)} \right] 
	&= \left(1 - \frac{1}{2^j}\right) 0 + \left(\frac{1}{2^j} - \frac{1}{2^\ell}\right) \frac{1}{2^{t-j}} + \frac{1}{2^\ell} \left( \frac{2^\ell}{2^t + 2^\ell - 2^j} - \frac{1}{2^{n-\ell}} \right) \\
	&\le \frac{1}{2^t} + \frac{1}{2^t + 2^\ell - 2^j}.
	\end{align*}
	Note that the second term is monotone increasing in $j$ and by definition $j \le \min\{t,\ell\}$.
	Hence, we deduce that
	\[
	\frac{1}{2^t + 2^\ell - 2^j} \le \frac{1}{2^t + 2^\ell - 2^{\min\{t,\ell\}}} = \frac{1}{2^{\max\{t,\ell\}}} \le \frac{1}{2^t}.
	\]
	We conclude that for any $A$, 
	\[
	\E_{\sigma} \left[ \tv{\uu\left( \cdot \mid \tau \right)}{\pi_{A,\sigma}\left( \cdot \mid \tau \right)} \right]  
	\le \frac{1}{2^{t-1}} \le \frac{16\eps}{n},
	\]
	where in the last inequality we recall that $t = \ceil{\log_2(n/\eps)} - 3 \ge \log_2(n/\eps) - 3$.
\end{proof}

\subsection{Uniformity testing with \texorpdfstring{$\coorora$}{Coordinate Oracle} and \texorpdfstring{$\stanora$}{General Oracle} for TV distance}
\label{subsec:lb-coor-tv}

In this section we consider uniformity testing over the binary hypercube for TV distance when we have access to $\coorora$ and $\stanora$. 
We assume the binary hypercube is denoted by $\XX_n = \{+1,-1\}^n$ instead of $\{0,1\}^n$, since our bad distributions will be Ising models where $+1,-1$ are more often used. 

Let $\alg$ denote an arbitrary uniformity testing algorithm (possibly randomized and adaptive) with $\coorora$ and $\stanora$ access. 
We assume that $\alg$ receives $L$ independent full samples from the $\stanora$, and is allowed to make $L$ queries to the $\coorora$. 
For ease of notation we denote by $\ora{\pi}$ the $\subcora$ with respect to a distribution $\pi$ over $\{+1,-1\}^n$.

\begin{definition}[Query History for $\coorora$ and $\stanora$]\label{qh:coor}
Let $\TT$ denote the set of all pinnings on $n-1$ coordinates (which is exactly all possible inputs to the $\coorora$). 
For integer $L \in \N^+$, we define the \emph{query history} with respect to $\alg$ and $\ora{\pi}$ of length $2L$ to be the random vector in $\XX_n^L \times (\TT \times \{+1,-1\})^L$ generated as follows:

\begin{itemize}
	\item Let $x_1,\dots,x_L$ be $L$ independent samples from $\pi$;
	\item For $i = 1,\dots,L$:
	\begin{itemize}
		\item $\alg$ receives $\left( x_1,\dots,x_L \right)$ and $\left( (\tau_1,a_1),\dots,(\tau_{i-1},a_{i-1}) \right)$ as input and generates $\tau_i \in \TT$ (randomly) as output;
		\item $\ora{\pi}$ receives $\tau_i$ as input and generates $a_i \in \{+1,-1\}$ as output.
	\end{itemize}
	\item The (coordinate and general) query history is $H = \left( x_1,\dots,x_L ;  (\tau_1,a_1),\dots,(\tau_L,a_L)  \right)$.
\end{itemize}
\end{definition}
Definition~\ref{qh:coor} is analogous to (in fact, a special case of) Definition~\ref{qh:sub}; throughout this subsection, we consider query history only with respect to $\coorora$ and $\stanora$.
 
The $\mathsf{output}$ of $\alg$ with sample complexity $2L$ is a (randomized) function of the query history $H$ of length $2L$.
Our main theorem is then stated as follows.

\begin{theorem}\label{thm:lb-coo-tv}
	There exists a universal constant $c>0$ such that the following holds. 
	Let $n \in \N^+$ be a sufficiently large integer and $\eps > 0$ be a real.
	Let $\uu = \uu_n$ denote the uniform distribution over $\{+1,-1\}^n$.
	Then there is no algorithm which can achieve the following properties using $L$ samples from $\stanora$ and $L$ queries from $\coorora$ where $L \le cn/\eps^2$:
	\begin{itemize}
		\item $\Pr_H\left( \mathsf{output} = \mathsf{Yes} \right) \ge 2/3$ for a random query history $H$ of length $2L$ with respect to $\alg$ and $\ora{\uu}$;
		\item $\Pr_{H'}\left( \mathsf{output} = \mathsf{No} \right) \ge 2/3$ for a random query history $H'$ of length $2L$ with respect to $\alg$ and $\ora{\pi}$ where $\pi$ is any distribution such that $\tv{\pi}{\uu} \ge \eps$.
	\end{itemize}
\end{theorem}

We observe that Theorem \ref{thm:lb-coor-tv} follows immediately from Theorem \ref{thm:lb-coo-tv}.

In \cite[Theorem 14]{DDK19} it was shown that $\Omega(n/\eps^2)$ samples are necessary for uniformity testing with only $\stanora$ access but assuming the hidden distribution $\pi$ is an Ising model. 
See also \cite[Theorem 14]{CDKS20} for very similar lower bounds in the setting of Bayesian networks. 
Note that though Theorem 14 from \cite{DDK19} is stated for \emph{symmetric} KL divergence, it actually works for TV distance as well, see \cite[Remark 4]{DDK19}. 
We use the same constructions from \cite{DDK19,CDKS20} for the family of bad distributions for our purpose. 
Assume that $n$ is even; the case of odd $n$ can be easily reduced to even $n$ by adding an extra uniform, independent coordinate. 
Suppose $M$ is a perfect matching of $n$ coordinates, i.e., $M$ is a collection of $n/2$ pairs of coordinates such that each coordinate appears in exactly one pair. 
Let $\MM$ be the set of all perfect matchings on $[n]$. 
Each bad distribution $\pi_M$ where $M \in \MM$ corresponds to an Ising model on the graph $G=([n],M)$ of $n/2$ edges, with the edge coupling set to be $\beta = \rho\eps/\sqrt{n}$ where $\rho$ is a universal constant sufficiently large. 
The following are established in \cite{CDKS20,DDK19}.  
\begin{clm}[\cite{CDKS20,DDK19}]
	\label{clm:DDK}
	\begin{enumerate}[(1)]
		\item For $\rho > 0$ sufficiently large, for all $M \in \MM$, it holds
		\[
		\tv{\pi_M}{\uu} \ge \eps.
		\]
		\item For any $\rho > 0$ there exists $c_1 = c_1(\rho) > 0$ such that the following holds.
		Suppose $L \le c_1 n/\eps$.
		Let $X = (x_1,\dots,x_L)$ be $L$ independent samples from $\uu$. 
		Independently, let $M \in \MM$ be chosen uniformly at random, and let $X' = (x'_1,\dots,x'_L)$ be $L$ independent samples from $\pi_M$.
		Then $\tv{X}{X'} \le 0.98$. 
	\end{enumerate}
\end{clm}

\begin{proof}
	(1) follows from Lemma 8 in \cite{CDKS20}. 
	(2) is proved in Section 8.3.2 in \cite{DDK19}. 
	See also in Section 8.1 from \cite{CDKS20} the same result for a slightly different construction of $\pi_M$, where every edge is set to be ferromagnetic with probability $1/2$ and antiferromagnetic otherwise.
	%
\end{proof}

Define $H$ to be the random query history of length $2L$ with respect to $\alg$ and $\ora{\uu}$, and let $\mathsf{output}$ denote the random output with respect to $H$ and $\alg$.
Define $H'$ to be the random query history of length $2L$ generated by
\begin{itemize}
	\item Pick $M \in \MM$ uniformly at random and let $\pi = \pi_M$;
	\item Let $H'$ be the random query history of length $2L$ with respect to $\alg$ and $\ora{\pi}$.
\end{itemize} 
Further, let $\mathsf{output}'$ denote the random output with respect to $H'$ and $\alg$.
Then we can show the following key lemma.
\begin{lemma}\label{lem:22tv-H-H'}
	For query histories $H,H'$ of length $2L$ where $L \le cn/\eps$ and $\mathsf{output},\mathsf{output}'$ defined as above, we have
	\[
	\tv{\mathsf{output}}{\mathsf{output}'} \le \tv{H}{H'} \le 0.99.
	\]
\end{lemma}

\begin{proof}
	The first inequality follows from the data processing inequality. 
	We focus on the second one.
	For $M \in \MM$ and $t \in \{0,1\}$, let $\pi_{M,t}$ denote the Ising model on $G=([n],M)$ with edge coupling $t\beta = t\rho\eps/\sqrt{n}$.
	Observe that $\pi_{M,0} = \uu$ and $\pi_{M,1} = \pi_M$.
	We rewrite the process 
	for generating the query histories $H$ and $H'$ of length $2L$ in the following equivalent form:
	
	\begin{itemize}
		\item Let $M \in \MM$ be chosen uniformly at random from $\MM$;
		\item Let $X_t = ( x_1,\dots,x_L ) \in \XX_n^L$ be $L$ independent samples from $\pi_{M,t}$;
		\item Let $R_t = (r_1,\dots,r_L) \in \{0,1\}^L$ be $L$ independent Bernoulli random variables with mean $(1+\tanh (t\rho\eps/\sqrt{n}))/2$;
		\item For $i = 1,\dots,L$:
		\begin{itemize}
			\item $\alg$ receives $\left( x_1,\dots,x_L \right)$ and $\left( (\tau_1,a_1),\dots,(\tau_{i-1},a_{i-1}) \right)$ as input and generates $\tau_i \in \TT$ (randomly) as output;
			\item $\ora{\pi}$ receives $\tau_i$ as input, which fixes all coordinates but one say $j$, and suppose $j'$ is matched to $j$ in $M$; 
			then $\ora{\pi}$ outputs $a_i = (\tau_i)_{j'}$ (the $j'$-th coordinate of $\tau_i$) as the sampled value at the $j$-th coordinate if $r_i = 1$, and outputs $a_i = -(\tau_i)_{j'}$ otherwise;
		\end{itemize}
		\item The query history is $H_t = \left( x_1,\dots,x_L ; (\tau_1,a_1),\dots,(\tau_L,a_L) \right)$.
	\end{itemize}
	
	Observe that if $t=0$, then the final query history $H_0$ is distributed as $H$; 
	meanwhile, if $t =1$, then it is distributed as $H'$.
	Moreover, the process above can be viewed as a random mapping from the vector $(M,X_t,R_t)$ to the query history $H_t$ where, for fixed $(M,X_t,R_t)$, the randomness purely comes from the decision-making of $\alg$. 
	Therefore, we can apply the data processing inequality and obtain
	\begin{equation*}
	\tv{H}{H'} \le \tv{(M,X_0,R_0)}{(M,X_1,R_1)} \le \tv{X_0}{X_1} + \tv{R_0}{R_1}.
	\end{equation*}
	Note that $\tv{X_0}{X_1} \le 0.98$ by Claim~\ref{clm:DDK}. 
	For the second term, we have
	\[
	\tv{R_0}{R_1} = 
	\tv{\mathrm{Bin}\left(L,\frac{1}{2}\right)}{\mathrm{Bin} \left(L,\frac{1}{2}\left(1+\tanh \frac{\rho\eps}{\sqrt{n}} \right)\right)}
	\le c' \cdot \sqrt{L} \cdot \frac{\rho\eps}{\sqrt{n}} \le 0.01,
	\]
	where $c'>0$ is a universal large constant, and $L \le cn/\eps^2$ for $c$ sufficiently small.
	Therefore, we deduce that $\tv{H}{H'} \le 0.98 + 0.01 = 0.99$ as claimed.
\end{proof}

We end this section with the proof of Theorem \ref{thm:lb-coo-tv}.
\begin{proof}[Proof of~Theorem \ref{thm:lb-coo-tv}]
	Suppose for sake of contradiction that $\alg$ satisfies both properties as in Theorem \ref{thm:lb-coo-tv}.
	Then by a standard amplification technique for failure probability, one can decrease the failure probability from $1/3$ to $0.001$ with the number of samples needed increases only by a constant factor; see \cite[Lemma 1.1.1]{Canonne22}).  
	In particular, for query histories $H,H'$ of length $2L$ where $L \le cn/\eps$ and $\mathsf{output},\mathsf{output}'$ defined as above, we have
	\[
	\Pr\nolimits\left( \mathsf{output} = \mathsf{Yes} \right) \ge 0.999
	\quad\text{and}\quad
	\Pr\nolimits\left( \mathsf{output}' = \mathsf{Yes} \right) \le 0.001.
	\]
	This implies $\tv{\mathsf{output}}{\mathsf{output}'} \ge 0.998$ which contradicts Lemma \ref{lem:22tv-H-H'}.  
\end{proof}

\section{Identity Testing with \texorpdfstring{$\subcora$}{Subcube Oracle}}
\label{sec:subcube}

In this section we give our algorithmic results for identity testing with access to the $\subcora$. In particular, we establish slightly more general versions of Theorem \ref{thm:alg-subcora} and Theorem \ref{thm:intro:kl-estimate-subc} which relax the assumption that $\mu$ is fully supported and only require that the support of $\pi$ is a subset of the support of $\mu$.
For the case when $\mu$ is not fully supported we instead require the slightly stronger assumption that $\mu$ is $b$-marginally bounded (this notion is equivalent to balancedness when $\mu$ is fully supported).
%
%
%



\subsection{Identity testing with exact conditional marginal distributions}
\label{subsec:subc-approx}

Recall that $[i] = \{1,\dots,i\}$ for an integer $i \in \N^+$. 
The following factorization of (relative) entropy is well-known, see e.g. \cite{Cesi,MSW03,CP}. 

\begin{lemma}
 For any distribution $\pi$ over $\QQ^n$ such that $\pi \ll \mu$ we have
	\begin{equation}\label{eq:ent-factorization}
	\kl{\pi}{\mu} = \sum_{i=1}^n \E_{x \sim \pi_{[i-1]}}\left[ \kl{\pi_i(\cdot \mid x)}{\mu_i(\cdot \mid x)} \right].
	\end{equation}
\end{lemma}



We now give our testing algorithm with $\subcora$. 
\begin{theorem}
	\label{thm:alg-subc}
	Let $k = k(n)$ be an integer and let $\bb = \bb(n) \in (0,1/2]$ be a real.
	Suppose that $\log\log(1/\bb) = O(\log n)$.
	There is an identity testing algorithm for all $b$-marginally bounded distributions with query access to $\subcora$ and for KL divergence with distance parameter $\eps > 0$. 
	The query complexity of the identity testing algorithm is 
	\[
	O\left( \min\left\{ \frac{1}{\sqrt{\bb}} \cdot \frac{n}{\eps} \log^3 \left(\frac{n}{\eps}\right),\, \sqrt{k} \log\Big(\frac{1}{\bb}\Big) \cdot \frac{n^2}{\eps^2} \log^2 \left(\frac{n}{\eps}\right) \right\} \right). 
	\]
	The running time of the algorithm is polynomial in all parameters and also proportional to the time of computing the conditional marginal distributions $\mu_i(\cdot \mid x)$ for any $i \in [n]$ and any feasible $x \in \QQ^{[i-1]}$. 
	Furthermore, if $k=2$, i.e., we have a binary domain $\QQ = \{0,1\}$, the query complexity can be improved to
	\[
	O\left( \log\Big(\frac{1}{\bb}\Big) \cdot \frac{n}{\eps} \log^3 \left(\frac{n}{\eps}\right) \right).
	\]
\end{theorem}

\begin{proof}
	We observe that \eqref{eq:ent-factorization} can be equivalently written as 
	\[
	\kl{\pi}{\mu} = n \; \E_{(i,x)} \left[ \kl{\pi_i(\cdot \mid x)}{\mu_i(\cdot \mid x)} \right],
	\] 
	where $i \in [n]$ is a uniformly random coordinate and $x$ is generated from $\pi_{[i-1]}$. 
	Therefore, Algorithm~\ref{alg:id-test} still works once we generate the pair $(i,x)$ in Line~\ref{line:sample-i-x} from the correct distribution as just described, and define $p^x_i = \pi_i(\cdot \mid x)$, $q^x_i = \mu_i(\cdot \mid x)$ correspondingly. 
	The analysis is exactly the same with the constant $C$ for approximate tensorization replaced by $1$. 
	We omit the proofs here
	and only highlight the differences: the coordinate balancedness $\eta$ is now replaced by the marginal boundedness $b$, and the running time depends on the time to compute the conditional marginal distributions $\mu_i(\cdot \mid x)$ for any $i \in [n]$ and any $x \in \QQ^{[i-1]}$ such that $\mu_{[i-1]}(x) > 0$.
\end{proof}

\begin{remark}
	\label{rmk:marginal-bound}
	We remark that the assumption of marginal boundedness can be relaxed to the following slightly weaker version: 
	for a fixed ordering of the coordinates, for every $i \in [n]$, every $x \in \QQ^{[i-1]}$ with $\mu_{[i-1]}(x) > 0$, and every $a \in \QQ$, one has
	\[
	\text{either~~} \mu_i(a \mid x) = 0,
	\quad
	\text{or~~} \mu_i(a \mid x) \ge b. 
	\]
	In some circumstances, this weaker notion of marginal boundedness can give a better bound on the sample complexity.
\end{remark}

Theorem \ref{thm:alg-subc} that identity testing can be done efficiently for a wide variety of families of distributions with the power of $\subcora$, assuming that one can efficiently compute the exact marginal probabilities under any conditioning. 
Below we give a few examples where Theorem \ref{thm:alg-subc} applies:

\begin{itemize}
	\item Consider any undirected graphical model (e.g., Ising model, Potts model) defined on trees of constant degrees. Then the distributions are $\Omega(1)$-marginally bounded, and one can efficiently compute the marginal probabilities under any pinning via, e.g., Belief Propagation. 
	Hence, there is a polynomial-time identity testing algorithm for undirected graphical models on bounded-degree trees with $\subcora$ access. 
	The sample complexity is $O((n/\eps) \log^3(n/\eps))$ where $n$ is the number of vertices. 
	If the degree is unbounded, then the marginal bound $b$ can be as small as $e^{-\Theta(n)}$. 
	Still, by the second bound in Theorem \ref{thm:alg-subc} the number of samples needed is at most $O((n^3/\eps^2) \log^2(n/\eps))$.
	
	\item Consider the Bayesian network on DAGs, and assume without loss of generality that $[n] = \{1,\dots,n\}$ is the topological ordering of the DAG. 
	In particular, all conditional marginal probabilities at any coordinate $i \in [n]$ and conditioned on any feasible pinning $x \in \QQ^{[i-1]}$ are given by the Bayesian network. 
	If these conditional marginal probabilities are lower bounded by $b=\Omega(1)$, then there is a polynomial-time identity testing algorithm for such Bayesian networks with $\subcora$ access, and the sample complexity is $O((n/\eps) \log^3(n/\eps))$. 
	If $b$ is exponentially small, then similarly as before the sample complexity is $O((n^3/\eps^2) \log^2(n/\eps))$. 
	See also Remark~\ref{rmk:marginal-bound} above on relaxing the marginal boundedness condition to specifically the topological ordering. 
	
	\item Consider mixtures of polynomially many product distributions, each of which has $\eta(\mu) = \Omega(1)$ as defined in Section~\ref{subsub:prod}. 
	One can efficiently compute the conditional marginal probabilities by the simple nature of mixtures of product distributions. 
	Then by Theorem \ref{thm:alg-subc}, we have an efficient identity testing algorithm with $\subcora$ access and the sample complexity is $O((n/\eps) \log^3(n/\eps))$. 
	Similarly as before, the sample complexity becomes $O((n^3/\eps^2) \log^2(n/\eps))$ when the minimum $\eta(\mu)$ is exponentially small. 
\end{itemize}


\subsection{Identity testing with approximate conditional marginal distributions}
\label{subsec:subc-exact}
In Theorem~\ref{thm:alg-subc} we assume that one can compute exactly any conditional marginal distribution in polynomial time. 
In some applications the exact computation is not possible and one can get, at the best, an estimator of the conditional marginal probabilities. 
As we will show in this subsection, identity testing can still be done efficiently in this setting.   

We first need more robust versions of Lemma \ref{lem:KL-id-test} and Lemma \ref{lem:Ber-KL-id-test}.
We say there is an FPRAS for a distribution $q$ over $\QQ$ if for any $\eps > 0$ and $\delta \in (0,1)$, one can compute a distribution $\hat{q}$ over $\QQ$ as an approximation of $q$ such that, with probability $1-\delta$, we have that for every $a \in \QQ$, 
\[
e^{-\eps} \le \frac{\hat{q}(a)}{q(a)} \le e^\eps,
\]
and $\hat{q}$ can be computed with running time polynomial in $k$, $1/\eps$, $\log(1/\delta)$, and the input size of $q$ (e.g., the number of parameters representing $q$). 
We remark that if $q(a) = 0$ then $\hat{q}(a) = 0$. 

\begin{lemma}
	\label{lem:robost-KL-id}
	Let $k \in \N^+$ be an integer, 
	and let $\eps> 0$, $\bb \in (0,1/2]$ be reals. 
	Given an FPRAS for a target distribution $q$ over domain $\QQ$ of size $k$ such that either $q(a) = 0$ or $q(a) \ge \bb$ for any $a \in \KK$, and given sample access to an unknown distribution $p \ll q$ over $\QQ$, there exists a polynomial-time identity testing algorithm that distinguishes with probability at least $2/3$ between the two cases 
	\begin{equation}
	p = q
	\quad\text{and}\quad
	\kl{p}{q} \ge \eps.
	\end{equation}
	For $k \ge 3$, the sample complexity of the identity testing algorithm is
	\[
	O\left( \min\left\{ \frac{1}{\eps \sqrt{\bb}},\, \frac{\sqrt{k} \ln(1/\bb)}{\eps^2} \right\} \right). 
	\]
	For $k=2$, the sample complexity of the identity testing algorithm is
	\[
	O\left( \frac{\ln(1/\bb)}{\eps} \right).
	\]
\end{lemma}

\begin{proof}
	Let $m$ be an upper bound for the number of samples required in Lemma \ref{lem:KL-id-test} and Lemma \ref{lem:Ber-KL-id-test}, with the assumption being either $q(a) = 0$ or $q(a) \ge \bb/2$ for any $a \in \KK$, distance parameter $\eps/2$, and failure probability $1/10$. 
	Let $\xi = O(\min\{\eps,1/m\})$ be a small constant, and let $\hat{q}$ be an approximation of $q$ such that with probability $9/10$ we have $e^{-\xi} \le \hat{q}(a)/q(a) \le e^\xi$ for every $a \in \QQ$. 
	Notice that if this holds then
	\[
	\left|\kl{p}{\hat{q}} - \kl{p}{q}\right|
	\le \sum_{a \in \QQ} p(a) \left| \ln \left( \frac{q}{\hat{q}} \right) \right| \le \xi. 
	\]
	
	We then apply the identity testing algorithm $\AA_{\textsc{kl-id}}$ from Lemma \ref{lem:KL-id-test} and Lemma \ref{lem:Ber-KL-id-test} to the distributions $p,\hat{q}$ with distance parameter $\eps/2$ and failure probability $1/10$, and returns the output of $\AA_{\textsc{kl-id}}$ as our output. 
	Note that $\hat{q}(a) = 0$ if $q(a) = 0$ and $\hat{q}(a) \ge e^{-\xi} q(a) \ge \bb/2$ if $q(a) \ge \bb$, assuming $\hat{q}$ is a $\xi$-approximation of $q$. 
	Thus, the number of samples required by $\AA_{\textsc{kl-id}}$ is at most $m$. 
	If $\kl{p}{q} \ge \eps$, then 
	\[
	\kl{p}{\hat{q}} \ge \kl{p}{q} - \left|\kl{p}{\hat{q}} - \kl{p}{q}\right| \ge \eps -\xi \ge \frac{\eps}{2}.
	\]
	Hence, the testing algorithm wrongly outputs Yes only if at least one of the following happens
	\begin{enumerate}[(1)]
		\item $\hat{q}$ is not a $\xi$-approximation of $q$, which happens with probability at most $1/10$;
		\item $\AA_{\textsc{kl-id}}$ makes a mistake, which happens with probability at most $1/10$.
	\end{enumerate}
	This shows that the failure probability is at most $1/5$. 
	If $p = q$, then notice that 
	\[
	\tv{p}{\hat{q}} = \tv{q}{\hat{q}} = O(\xi) \le \frac{1}{10m}.
	\]
	We consider an optimal coupling between $m$ independent samples from $p$ and $m$ independent samples from $\hat{q}$, so the probability that these two sets of $m$ samples are not exactly the same is at most $1/10$.
	One can think of the testing process as follows: we try to send $m$ samples from $\hat{q}$ to $\AA_{\textsc{kl-id}}$, and it succeeds only when the samples are coupled with those from $p$.
	Therefore, the failure probability, in addition to (1) and (2) above, also includes this uncoupled probability, and hence is at most $3/10$.
	Finally, the number of samples needed, $m$, is bounded in Lemma \ref{lem:KL-id-test} and Lemma \ref{lem:Ber-KL-id-test}.
\end{proof}

Lemma \ref{lem:robost-KL-id}, combined with the proof of Theorem \ref{thm:alg-subc}, immediately implies the following theorem. 
See also Remark~\ref{rmk:marginal-bound} for the discussion on relaxing marginal boundedness.

\begin{theorem}
	\label{thm:robust-alg-subc}
	Let $k = k(n)$ be an integer and let $\bb = \bb(n) \in (0,1/2]$ be a real.
	Suppose that $\log\log(1/\bb) = O(\log n)$.
	There is an identity testing algorithm for all $b$-marginally bounded distributions with query access to $\subcora$ and for KL divergence with distance parameter $\eps > 0$. 
	The query complexity of the identity testing algorithm is 
	\[
	O\left( \min\left\{ \frac{1}{\sqrt{\bb}} \cdot \frac{n}{\eps} \log^3 \left(\frac{n}{\eps}\right),\, \sqrt{k} \log\Big(\frac{1}{\bb}\Big) \cdot \frac{n^2}{\eps^2} \log^2 \left(\frac{n}{\eps}\right) \right\} \right). 
	\]
	The running time of the algorithm is polynomial in all parameters assuming that there is an FPRAS for the conditional marginal distributions $\mu_i(\cdot \mid x)$ for any $i \in [n]$ and any feasible $x \in \QQ^{[i-1]}$. 
	Furthermore, if $k=2$, i.e., we have a binary domain $\QQ = \{0,1\}$, the query complexity can be improved to
	\[
	O\left( \log\Big(\frac{1}{\bb}\Big) \cdot \frac{n}{\eps} \log^3 \left(\frac{n}{\eps}\right) \right).
	\]
\end{theorem}

Again, we give a few examples as applications of Theorem~\ref{thm:robust-alg-subc}, omitting all the technical details. 

\begin{itemize}
	\item Consider the Ising model with the interaction matrix $J$ (with entries being $\beta_{uv}$'s and assumed to be positive semidefinite). 
	We know from recent works \cite{EKZ,AJKPV,KLR} that one can efficiently estimate all conditional marginal probabilities when $\norm{J}_2 < 1$ under any external fields. 
	There are two special features for this application. The first is that the marginal bounds could potentially be as small as $e^{-\Theta(\sqrt{n})}$. 
	The second is that we can only approximate the conditional marginal probabilities rather than get the exact values, and hence we should apply Theorem~\ref{thm:robust-alg-subc} instead of Theorem~\ref{thm:alg-subc}.
	With access to the $\subcora$, one can obtain a polynomial-time identity testing algorithm for this family of Ising models with sample complexity $O((n^{3/2}/\eps) \log^3(n/\eps))$ (note that $k=2$).
	\item Consider the monomer-dimer model (weighted matchings) on arbitrary (unbounded-degree) graphs. We know from the classical work \cite{JS} that one can approximate the conditional marginal distributions for all pinnings. 
	Similar to the previous example, the marginal probabilities can be exponentially small (in the number of vertices) and one can at the best approximate them efficiently rather than computing them exactly. Still, we can apply Theorem~\ref{thm:robust-alg-subc} to obtain an efficient identity testing algorithm with access to the $\subcora$ with sample complexity $O((mn/\eps) \log^3(n/\eps))$ where $m$ is the number of edges of the graph (which is the dimension) and $n$ is the number of vertices (note that $k=2$). 
\end{itemize}


\subsection{Estimating KL divergence with additive error}

With access to the $\subcora$, we can also estimate the KL divergence from an unknown distribution $\pi$ to a given distribution $\mu$ within an arbitrary additive error in polynomial time.
This corresponds to the \emph{tolerant identity testing} problem for KL divergence;
that is, given $s,\eps > 0$, we want to distinguish between $\kl{\pi}{\mu} \le s$ and $\kl{\pi}{\mu} \ge s+\eps$. 

We first consider estimating KL divergence for distributions on a finite domain of size $k$. 

\begin{lemma}
	\label{lem:approx-KL}
	Let $k \in \N^+$ be an integer, 
	and let $\eps> 0$, $\bb \in (0,1/2]$ be reals. 
	Given an FPRAS for a target distribution $q$ over domain $\QQ$ of size $k$ such that either $q(a) = 0$ or $q(a) \ge \bb$ for any $a \in \KK$, and given sample access to an unknown distribution $p \ll q$ over $\QQ$, there exists a polynomial-time algorithm that computes $\widehat{R}$ such that with probability at least $2/3$ it holds 
	\begin{equation}
	\left| \widehat{R} - \kl{p}{q} \right| \le \eps,
	\end{equation}
	with sample complexity 
	\[
	O\left( \frac{k}{\eps \log(k/\eps)} + \frac{\log^2(1/\bb)}{\eps^2} \right). 
	\]
\end{lemma}

For a distribution $p$ over a finite domain $\QQ$, the (Shannon) entropy of $p$ is defined as
\[
H(p) = \sum_{a \in \QQ} p(a) \ln \left( \frac{1}{p(a)} \right). 
\]
Observe that if $p \ll q$ are two distributions over $\QQ$, then
\begin{equation}\label{eq:kl-two-terms}
\kl{p}{q} = \sum_{a \in \QQ} p(a) \ln\left( \frac{1}{q(a)} \right) - \sum_{a \in \QQ} p(a) \ln\left( \frac{1}{p(a)} \right)
= \E_{a \sim p} \left[ \ln\left( \frac{1}{q(a)} \right) \right] - H(p).
\end{equation}
It suffices to estimate the two terms on the right-hand side of \eqref{eq:kl-two-terms} respectively with good enough accuracy. 

We need the following well-known result from \cite{VV17} for estimating the entropy of an unknown distribution from samples, see also \cite{VV11,JVHW15,WY16}. 

\begin{lemma}[\cite{VV17}]
	\label{lem:ent-est}
	Let $k \in \N^+$ be an integer, 
	and let $\eps> 0$ be a real. 
	Given sample access to an unknown distribution $p$ over domain $\QQ$ of size $k$, there exists a polynomial-time algorithm that computes $\widehat{H}$ such that with probability at least $9/10$ it holds 
	\begin{equation}
	\left| \widehat{H} - H(p) \right| \le \eps,
	\end{equation}
	with sample complexity 
	\[
	O\left( \frac{k}{\eps \log(k/\eps)} + \frac{\log^2k}{\eps^2} \right). 
	\]
\end{lemma}

For the fist term in \eqref{eq:kl-two-terms}, we show the following estimator.
\begin{lemma}\label{lem:kl-1st-term}
	Let $k \in \N^+$ be an integer, 
	and let $\eps> 0$, $\bb \in (0,1/2]$ be reals. 
	Given an FPRAS for a target distribution $q$ over domain $\QQ$ of size $k$ such that either $q(a) = 0$ or $q(a) \ge \bb$ for any $a \in \KK$, and given sample access to an unknown distribution $p \ll q$ over $\QQ$, there exists a polynomial-time algorithm that computes $\widehat{G}$ such that with probability at least $4/5$ it holds 
	\begin{equation}\label{eq:widehat-G}
	\left| \widehat{G} - \E_{a \sim p} \left[ \ln\left( \frac{1}{q(a)} \right) \right] \right| \le \eps,
	\end{equation}
	with sample complexity 
	\[
	O\left( \frac{\ln^2(1/b)}{\eps^2} \right). 
	\]
\end{lemma}

\begin{proof}
	Compute an approximation $\hat{q}$ of $q$ such that, with probability $9/10$, we have that $e^{-\eps/2} \le \hat{q}(a)/q(a) \le e^{\eps/2}$ for every $a \in \QQ$ with $q(a) > 0$, and $\hat{q}(a) = 0$ for $q(a) = 0$. 
	Generate $m$ independent samples from $p$, denoted by $a_1,\dots,a_m$. 
	Then our estimator is defined as
	\[
	\widehat{G} = \frac{1}{m} \sum_{j=1}^m \ln\left( \frac{1}{\hat{q}(a_j)} \right).
	\]
	We will show that $\widehat{G}$ satisfies \eqref{eq:widehat-G} with probability at least $4/5$ for
	\[
	m = \ceil{\frac{8 \ln^2(1/b)}{\eps^2}}.
	\]
	
	Observe that
	\begin{align}
	&\left| \widehat{G} - \E_{a \sim p} \left[ \ln\left( \frac{1}{q(a)} \right) \right] \right| \nonumber\\
	\le{}& \left| \frac{1}{m} \sum_{j=1}^m \ln\left( \frac{1}{\hat{q}(a_j)} \right) - \frac{1}{m} \sum_{j=1}^m \ln\left( \frac{1}{q(a_j)} \right) \right|
	+ \left| \frac{1}{m} \sum_{j=1}^m \ln\left( \frac{1}{q(a_j)} \right) - \E_{a \sim p} \left[ \ln\left( \frac{1}{q(a)} \right) \right] \right|
	\label{eq:wideG-error}
	\end{align}
	Assuming $\hat{q}$ is an $(\eps/2)$-approximation of $q$, we can upper bound the first term in \eqref{eq:wideG-error} by
	\[
	\frac{1}{m} \sum_{j=1}^m \left| \ln\left( \frac{\hat{q}(a_j)}{q(a_j)} \right) \right| \le \frac{1}{m} \sum_{j=1}^m \frac{\eps}{2} = \frac{\eps}{2}. 
	\]
	Meanwhile, for the second term in \eqref{eq:wideG-error}, since $0 \le \ln(1/q(a)) \le \ln(1/b)$ for all $a \in \QQ$ with $q(a) > 0$ and since $p\ll q$,
	we deduce from Hoeffding's inequality that
	\begin{align*}
	\Pr\left( \left| \frac{1}{m} \sum_{j=1}^m \ln\left( \frac{1}{q(a_j)} \right) - \E_{a \sim p} \left[ \ln\left( \frac{1}{q(a)} \right) \right] \right| \ge \frac{\eps}{2} \right) 
	\le 2 \exp\left( -\frac{\eps^2 m}{2\ln^2(1/b)} \right) \le \frac{1}{10},
	\end{align*}
	provided 
	$
	m \ge (8/\eps^2) \ln^2(1/b).
	$
	Therefore, we deduce from \eqref{eq:wideG-error} that
	\[
	\left| \widehat{G} - \E_{a \sim p} \left[ \ln\left( \frac{1}{q(a)} \right) \right] \right|
	\le \frac{\eps}{2} + \frac{\eps}{2} = \eps
	\]
	with failure probability at most $1/10 + 1/10 = 1/5$ by the union bound, as wanted.
\end{proof}

Lemma \ref{lem:approx-KL} then follows easily from Lemma \ref{lem:ent-est} and Lemma \ref{lem:kl-1st-term}. 

\begin{proof}[Proof Lemma \ref{lem:approx-KL}]
	Since the sample complexity upper bound we want to show is monotone increasing in $k$, we can safely assume without loss of generality that $q$ is fully supported on $\QQ$, i.e., $q(a) \ge \bb$ for each $a \in \QQ$. 
	In particular, this implies that $b \le 1/k$. 
	Take $\widehat{G}$ from Lemma \ref{lem:kl-1st-term} and $\widehat{H}$ from Lemma \ref{lem:ent-est}, and let $\widehat{R} = \widehat{G} - \widehat{H}$.
	We then deduce from \eqref{eq:kl-two-terms} that 
	\[
	\left| \widehat{R} - \kl{p}{q} \right| \le 
	\left| \widehat{G} - \E_{a \sim p} \left[ \ln\left( \frac{1}{q(a)} \right) \right] \right|
	+ \left| \widehat{H} - H(p) \right|
	\le \frac{\eps}{2} + \frac{\eps}{2} = \eps,
	\]
	which fails with probability at most $1/5 + 1/10 = 3/10$ by the union bound. 
	The running time is polynomial in all parameters and depends on the given FPRAS for $q$. 
	The sample complexity is given by
	\[
	O\left( \frac{k}{\eps \log(k/\eps)} + \frac{\log^2k}{\eps^2} \right) + O\left( \frac{\log^2(1/\bb)}{\eps^2} \right)
	= O\left( \frac{k}{\eps \log(k/\eps)} + \frac{\log^2(1/\bb)}{\eps^2} \right), 
	\]
	since we have $k \le 1/b$. 
\end{proof}

We now give our main theorem for estimating KL divergence with $\subcora$ access. 

\begin{theorem}
	\label{thm:kl-estimate-subc}
	Let $k = k(n)$ be an integer and let $\bb = \bb(n) \in (0,1/2]$ be a real.
	Suppose that $k = O(n)$ and $\log\log(1/\bb) = O(\log n)$.
	Given a visible distribution $\mu$ over $\QQ^n$ that is $b$-marginally bounded, and given access to $\subcora$ for a hidden distribution $\pi \ll \mu$ over $\QQ^n$, 
	there is an algorithm that for any $\eps > 0$ computes $\widehat{S}$ such that with probability at least $2/3$ it holds 
	\begin{equation}
	\left| \widehat{S} - \kl{\pi}{\mu} \right| \le \eps.
	\end{equation}
	The query complexity of the algorithm is 
	\[
	O\left( \log^4\Big(\frac{1}{\bb}\Big) \cdot \frac{n^4}{\eps^4}\log\left(\frac{n}{\eps}\right) \right). 
	\]
	The running time of the algorithm is polynomial in all parameters assuming that there is an FPRAS for the conditional marginal distributions $\mu_i(\cdot \mid x)$ for any $i \in [n]$ and any feasible $x \in \QQ^{[i-1]}$. 
\end{theorem}

\begin{proof}
	From \eqref{eq:ent-factorization} we observe that
	\[
	\kl{\pi}{\mu} = n \; \E_{(i,x)} \left[ \kl{\pi_i(\cdot \mid x)}{\mu_i(\cdot \mid x)} \right],
	\]
	where $(i,x)$ is a random pair generated by taking a uniformly random coordinate $i \in [n]$ and sampling $x \in \QQ^{[i-1]}$ from the marginal of $\pi$ on the first $i-1$ coordinates. 
	Hence, it suffices to estimate $\E_{(i,x)} \left[ \kl{\pi_i(\cdot \mid x)}{\mu_i(\cdot \mid x)} \right]$ with additive error $\eps/n$. 
	
	Let $(i_1,x_1), \dots, (i_L,x_L)$ be $L$ independent random pairs generated via the $\stanora$ (which is contained in the power of $\subcora$), where we define
	\[
	L = \ceil{\frac{8n^2\ln^2(1/b)}{\eps^2}}. 
	\] 
	For $1 \le \ell \le L$, we let 
	$$ R_\ell = \kl{\pi_{i_\ell}(\cdot \mid x_\ell)}{\mu_{i_\ell}(\cdot \mid x_\ell)}. $$ 
	Furthermore, for each $\ell$ let $\widehat{R}_\ell$ be an estimate of $R_\ell$ which is obtained from Lemma \ref{lem:approx-KL} via the $\subcora$, such that
	\[
	\Pr\left( \left| \widehat{R}_\ell - R_\ell \right| \ge \frac{\eps}{2n} \right) \le \frac{1}{10L}.
	\]
	Then, by the union bound we have
	\[
	\Pr\left( \left| \frac{1}{L} \sum_{\ell=1}^L \widehat{R}_\ell - \frac{1}{L} \sum_{\ell=1}^L R_\ell \right| \ge \frac{\eps}{2n} \right) 
	\le \sum_{\ell=1}^L \Pr\left( \left| \widehat{R}_\ell - R_\ell \right| \ge \frac{\eps}{2n} \right) 
	\le L \cdot \frac{1}{10L} = \frac{1}{10}.
	\]
	Note that using the standard amplification technique for the failure probability, the number of samples we need for each $\ell$ is 
	\[
	O\left( \frac{kn}{\eps \log(kn/\eps)} + \frac{n^2\log^2(1/\bb)}{\eps^2} \right) \cdot O\left(\log L\right)
	= O\left( \log^2\Big(\frac{1}{\bb}\Big) \cdot \frac{n^2}{\eps^2}\log\left(\frac{n}{\eps}\right) \right), 
	\]
	where we use the assumptions $k=O(n)$ and $\log\log(1/b) = O(\log n)$.  
	
	Meanwhile, we observe $0 \le \kl{\pi_i(\cdot \mid x)}{\mu_i(\cdot \mid x)} \le \ln(1/b)$ for any feasible pair $(i,x)$ since $\mu$ is $\bb$-marginally bounded and $\pi \ll \mu$. 
	Hence, Hoeffding's inequality implies that
	\[
	\Pr\left( \left| \frac{1}{L} \sum_{\ell=1}^L R_\ell - \E_{(i,x)} \left[ \kl{\pi_i(\cdot \mid x)}{\mu_i(\cdot \mid x)} \right] \right| \ge \frac{\eps}{2n} \right) 
	\le 2 \exp\left( - \frac{\eps^2 L}{2n^2\ln^2(1/b)} \right) \le \frac{1}{10},
	\]
	provided
	$
	L \ge (8n^2/\eps^2) \ln^2(1/b). 
	$
	
	Therefore, by letting our estimator to be
	\[
	\widehat{S} = \frac{n}{L} \sum_{\ell=1}^L \widehat{R}_\ell,
	\]
	we deduce that
	\begin{align*}
	| \widehat{S} - & \kl{\pi}{\mu} | \\
	&\le n \left| \frac{1}{L} \sum_{\ell=1}^L \widehat{R}_\ell - \frac{1}{L} \sum_{\ell=1}^L R_\ell \right| 
	+ n \left| \frac{1}{L} \sum_{\ell=1}^L R_\ell - \E_{(i,x)} \left[ \kl{\pi_i(\cdot \mid x)}{\mu_i(\cdot \mid x)} \right] \right| \\
	&\le n \cdot \frac{\eps}{2n} + n \cdot \frac{\eps}{2n} = \eps,
	\end{align*}
	with failure probability at most $1/10+1/10 = 1/5$ by the union bound. 
	Finally, the query complexity is given by
	\[
	O\left( \log^2\Big(\frac{1}{\bb}\Big) \cdot \frac{n^2}{\eps^2}\log\left(\frac{n}{\eps}\right) \right) \cdot L
	= O\left( \log^4\Big(\frac{1}{\bb}\Big) \cdot \frac{n^4}{\eps^4}\log\left(\frac{n}{\eps}\right) \right), 
	\]
	as claimed.
\end{proof}

We remark that Theorem~\ref{thm:kl-estimate-subc} is applicable to all the examples mentioned in Sections~\ref{subsec:subc-exact} and \ref{subsec:subc-approx}.
See also Remark~\ref{rmk:marginal-bound} on relaxing the marginal boundedness condition.

\section{Conclusion and Open Problems}
In this paper we give efficient algorithms for identity testing for the $\coorora$ model, and also establish matching computational hardness and information-theoretical lower bounds.
Our algorithmic result builds on the fact that the visible distribution satisfies approximate tensorization of entropy. 
While we show that for the antiferromagnetic Ising model, there is no polynomial-time identity testing algorithm when approximate tensorization fails, it is in general unclear if one can get a testing algorithm running in polynomial time without approximate tensorization, using either $\coorora$ or $\subcora$ in addition to $\stanora$. 
One important example is the ferromagnetic Ising model at all temperatures. 
We know that approximate tensorization fails at low temperature (large $\beta$) since the Glauber dynamics has exponential mixing time. 
We do not know whether an efficient identity testing algorithm exists or not even with access to the more powerful $\subcora$.
Note that our Theorem~\ref{thm:alg-subc} does not apply to ferromagnetic Ising models since we cannot estimate conditional marginal probabilities under an arbitrary pinning (corresponding to ferromagnetic Ising models with inconsistent local fields).
Another important example is mixtures of product distributions. 
It is easy to show that approximate tensorization could fail even for a mixture of two product distributions with equal weights. 
We know from Theorem~\ref{thm:alg-subc} that there is an efficient identity testing algorithm for the family of mixtures of polynomially many balanced product distributions given access to the $\subcora$. 
It is unclear to us, however, that if there is a polynomial-time testing algorithm using only the weaker $\coorora$.

\bigskip\noindent\textbf{Acknowledgements.} 
AB research was supported by NSF grant CCF-2143762. EV was supported by NSF grant CCF-2147094.

\bibliographystyle{alpha}
\bibliography{testing}

\end{document}